\documentclass[11pt]{article}
\usepackage[utf8]{inputenc}
\usepackage{bm}
\usepackage{hyperref}
\usepackage{xcolor}
\definecolor{iris}{HTML}{5D3FD3}

\hypersetup
{
   colorlinks=true,
  linkcolor=iris,
  urlcolor={iris},
  filecolor={iris},
  citecolor={iris},
  allcolors={iris}
}
\title{Dynamic Interventions for Networked Contagions\thanks{Alphabetical order of authors is followed by first placing students and then non-students.}}
\author{
    Marios Papachristou \\ Cornell University \\ \texttt{papachristoumarios@cs.cornell.edu}\footnote{Department of Computer Science.} \and
    Siddhartha Banerjee \\ Cornell University \\ \texttt{sbanerjee@cornell.edu}\footnote{School of Operations Research and Information Engineering (ORIE).} \and
    Jon Kleinberg \\ Cornell University \\ \texttt{kleinberg@cornell.edu}\footnote{Department of Computer Science and Department of Information Science.}
}

\date{}
\usepackage[margin=1in]{geometry}
\usepackage{xurl}
\usepackage{graphicx}
\usepackage{amsmath}
\usepackage{amsfonts}
\usepackage{amssymb}
\usepackage{amsthm}
\usepackage{appendix}
\usepackage{algorithmicx}
\usepackage{algorithm}
\usepackage{algpseudocode} 
\usepackage{caption}

\usepackage{paralist}
\usepackage{tikz-network}
\usepackage{booktabs} 
\usepackage[noabbrev, capitalise]{cleveref}
\usepackage{subfigure}
\usepackage{color-edits}
\addauthor{mp}{red}
\addauthor{jk}{blue}
\addauthor{sb}{teal}

\newtheorem{theorem}[]{Theorem}

\newtheorem{assumption}[]{Assumption}
\newtheorem{condition}[]{Condition}

\crefname{assumption}{Assumption}{Assumptions}

\theoremstyle{remark}

\newcommand{\ev}[2]{\mathbb E_{#1} \left [ #2 \right ]}

\newcommand{\one}[0]{\mathbf 1}
\newcommand{\zero}[0]{\mathbf 0}

\newcommand{\til}[0]{\widetilde}

\renewcommand{\Pr}[2]{\mathbb P_{#1} \left [ #2 \right ]}

\newcommand{\frbox}[1]{\medskip \noindent \fbox {\parbox{\textwidth} {#1}} \medskip}

\begin{document}

\maketitle

\begin{abstract}
We study the problem of designing dynamic intervention policies for minimizing networked defaults in financial networks.
Formally, we consider a dynamic version of the celebrated Eisenberg-Noe model of financial network liabilities and use this to study the design of external intervention policies.
Our controller has a fixed resource budget in each round and can use this to minimize the effect of demand/supply shocks in the network.
We formulate the optimal intervention problem as a Markov Decision Process and show how we can leverage the problem structure to efficiently compute optimal intervention policies with continuous interventions and provide approximation algorithms for discrete interventions. Going beyond financial networks, we argue that our model captures dynamic network intervention in a much broader class of dynamic demand/supply settings with networked inter-dependencies. 
To demonstrate this, we apply our intervention algorithms to various application domains, including ridesharing, online transaction platforms, and financial networks with agent mobility. In each case, we study the relationship between node centrality and intervention strength, as well as the fairness properties of the optimal interventions. 

\medskip


\end{abstract}

\section{Introduction}

\noindent \textbf{Motivation.} The world consists of interconnected entities that interact with one another through networks. Networks experience shocks due to adverse scenarios, such as (partial) failures of nodes and edges. When exogenous shocks hit networks, such shocks propagate through the edges of the network, causing cascades that may affect a significant population of nodes in the network. This phenomenon is known as \emph{network contagion}, and a particularly interesting category of networks that undergo contagion are \emph{financial networks}. In financial networks, financial entities, such as individuals, businesses, and banks have \emph{liabilities} to one another as well as assets which can be attributed either \emph{internally}, i.e., within the financial network in question, or \emph{externally}, i.e., outside the financial network. It is often the case that the entities within these networks do not have adequate means to pay off their financial obligations, so they become \emph{default}, in which case they are only able to respond and pay off, or ``clear'', a fraction of their debts. A \emph{planner} acts as an external force and is responsible for (optimally) allocating resources, also known as \emph{interventions} (or bailouts) subject to budget constraints so that defaults are averted. Interventions can be fractionally allocated \cite{ahn2019optimal, demange2018contagion} or discretely allocated \cite{papachristou2021allocating}. The former setting is a tractable problem, and the latter is computationally intractable. 

When modeling and studying such networked interactions, much of the literature assumes that the interactions are \emph{one-shot} (i.e., static)(e.g.~\cite{glasserman2015likely, papachristou2021allocating} etc.).
However, in many situations, networked interactions evolve, subject to an \emph{uncertain environment}, and where the planner's interventions at some point in time affect the system's state in future times. 
So far, limited attention has been given to contagion processes that evolve dynamically: First of all, such approaches have either been considered in limited settings of a financial system together with specific strategies for the mitigation of systemic risk that do not fall under an intervention regime~\cite{chen2016financial}, or are based on different modeling assumptions and are computationally more intensive~\cite{barratt2020multi}. Secondly, in terms of the nature of the allocated resources themselves, the \emph{fractional} and \emph{discrete} allocation schemes we propose in this paper have not been studied, as well as the case where there are constraints regarding the \emph{fair} distribution of interventions. 

\medskip

\noindent \textbf{A Model for Dynamic Interventions in Networks.}Our work addresses the above points. More specifically, we formulate and study systems subject to a contagion process over time, with the system's current state and interventions affecting its future state. In our model, the agents can be either \emph{solvent} or \emph{default}. The former state means that they can clear all of their debts, and the latter stage means that they cannot clear their debts but rather must split their debts \emph{proportionally} among their creditors. Unmet liabilities then \emph{accumulate over time}. 
So designing an intervention policy in such a regime requires leveraging both spatial and temporal information. We initially study the design of an optimal \emph{fractional} intervention policy, which is a computationally tractable problem (assuming that the system responds \emph{optimally} at every round). Next, we study the allocation of discrete resources, a computationally intractable problem whose static version has been studied in \cite{papachristou2021allocating}. Here, we design approximate discrete intervention policies based on the fractional intervention policy. 

While our models are motivated by (and build on) existing models of liabilities and defaults in financial networks, we posit that these models, and the corresponding resource allocation problems, can extend far beyond this setting. 
For instance, one place where one can utilize such a clearing mechanism is \emph{ridesharing}. The nodes of the network represent \emph{neighborhoods} of a city, the external world represents suburban areas of the city, and the internal liabilities between nodes represent the number of rides requested from a neighborhood $i$ to a neighborhood $j$. The external liabilities (resp. external assets) correspond to rides requested from a neighborhood $i$ to suburban areas (resp. rides from the suburban areas to a neighborhood $i$). The interventions correspond to ridesharing vehicles that can be allocated in any given neighborhood. The assets equal the conditions of the road network with shocks, representing adverse scenarios such as traffic jams.
Similarly, in a \emph{computing cluster} allocation regime, nodes represent computation nodes, and liabilities represent the amount of computation that can be displaced to adjacent nodes. Assets represent available compute power at each node. Shocks represent failures of computing resources, and interventions represent allocations of backup resources to the existing centers (which can also be utilized, e.g., because of high demand). 

More generally, any problem that corresponds to a \emph{supply and demand network that evolves} at which, when nodes cannot meet their corresponding obligations in full, they distribute them \emph{proportionally} can be modeled with our framework. A planner seeks to allocate resources in this environment subject to a \emph{budget constraint} can be captured by our framework. Our framework can support the maximization of a variety of objectives. This paper focuses on a particular objective function; however, the solutions produced are equivalent to every strictly increasing objective. We provide approximation guarantees of the computed solution compared to the optimal solution for the discrete case, subject to realistic monotonicity assumptions. As a result, our dynamic contagion framework is suitable for resolving resource allocation problems in a variety of domains (both through the lens of fractional and discrete allocations); financial transaction networks (physical or on the Web [e.g., Venmo, cryptocurrencies]), ridesharing, high-performance computing, ad-placement, to name a few. 

\medskip

\noindent \textbf{Our Technical Contributions.}
Formally, we study the problem of optimally allocating (fractional and discrete) resources subject to \emph{(i)} contagion effects and \emph{(ii)} an uncertain environment, i.e., an environment that experiences financial shocks. We generalize the model of Eisenberg and Noe \cite{eisenberg2001systemic} to discrete time as a \emph{Markov Decision Process} (MDP) and formulate the optimal intervention problem (\cref{sec:model}). 
The resulting MDP is very high-dimensional; nevertheless, we show how we can leverage the problem structure to compute near-optimal intervention policies for continuous interventions efficiently. Moreover, under discrete interventions, we demonstrate how we can use the above continuous intervention policies to derive heuristic control policies with formal approximation guarantees  (\cref{sec:discrete_allocations}). 
In addition, our framework also supports incorporating additional fairness constraints (based on generalizations of the Gini Coefficient) for the distribution of the resources such that interventions are more equitable across nodes. Surprisingly, we also show that incorporating such objectives has little effect on the welfare objective in our setting.

We supplement our theoretical results with an extensive experimental evaluation (see \cref{sec:experiments}). In particular, we use the derived algorithm to tackle problems regarding dynamic resource allocation on real-world data from the Venmo transaction platform, semi-artificial data from mobility patterns, real-world data for ridesharing applications from New York City's Taxi and Limousine Commission, and synthetic data with core-periphery structure. We also incorporate a variety of fairness (\cref{sec:fairness}) measures on the intervention allocations, which stem from generalizations of the Gini Coefficient \cite{gini1921measurement} in our model. In \cref{sec:insights}, we study the relationship between optimal interventions and the resulting clearing payments of the system and the relationship between optimal interventions and node centrality (i.e., financial connectivity). We find that generally, where there are no (spatial) fairness constraints on the system, the most ``important'' nodes get most of the interventions. This behavior changes when we consider fairness constraints; however, we empirically observe our hypothesis that under fair intervention policies, the overall social welfare value is close to the optimal one (in the case of fractional interventions) without the fairness constraint.

\medskip

\noindent \textbf{Source Code.} The \emph{source code} and \emph{data} used for the simulations and experiments of this paper can be found at \url{https://github.com/papachristoumarios/dynamic-clearing}.

\medskip

\subsection{Related Work} \label{sec:related_work}

\noindent \textbf{Financial Networks and the Eisenberg-Noe (EN) Model.} The EN model introduced in~\cite{eisenberg2001systemic} considers a financial network where each node has assets and liabilities both concerning the internal network (i.e., the other nodes in the network) and the external sector (the node may, for instance, owe taxes to the government, or have loans to external creditors, as well as may get social security benefits from the government). According to the EN model, when a node \emph{defaults}, i.e., when it is not able to pay out its creditors (internal and external), it then re-scales its obligations \emph{proportionally} (this is sometimes referred to as absolute priority or ``pro-rata payments'') and then pays the re-scaled responsibilities in total (this principle is known as ``liability over equity''). The EN model is used to calculate these payments, usually called \emph{clearing payments}, by computing a solution to a fixed-point problem, or equivalently, the optimal solution to an optimization program with a strictly increasing objective. The general problem may have multiple equilibria; however, for simplicity, we focus on settings where the equilibrium payments are unique.  

In the presence of shocks~\cite{glasserman2015likely}, the nodes similarly become default due to external shocks that disrupt their assets. Moreover, \cite{jackson2020credit, ahn2019optimal}, investigate \emph{optimal interventions} in the cases of defaulting.  The work of~\cite{ahn2019optimal} considers optimal intervention methods under budget constraints under the extended model of~\cite{glasserman2015likely} and formulate optimization problems that minimize the systemic losses as well as minimize the number of defaulting institutions and applies their methods to publicly available data on the Korean financial system. Their work considers \emph{fractional intervention} policies which differentiate it from ours. More specifically, we consider \emph{discrete interventions} (i.e., a node gets the full intervention or does not), in which case the model of~\cite{ahn2019optimal} can be seen as the fractional relaxation of the optimization objectives we study. Secondly, we experiment both with high-granularity financial institutions (banks) as well as view the problem from a societal lens, namely modeling the entities of the system as ``societal nodes'' (businesses, households, individuals).  \cite{feinstein2019obligations} generalizes the classical EN model to a multi-layered one, develops a financial contagion 
model with fire sales that allows institutions to buy and sell assets and studies its equilibria. 

Finally, the work closest to ours is that of \cite{papachristou2021allocating}, where the authors use the static Eisenberg-Noe model for studying the allocation of stimulus checks subject to financial shocks. However, the static nature of~\cite{papachristou2021allocating} means that it is hard to use their techniques and results directly in a dynamic setting. When the temporal dimension is introduced, the model's dynamics change significantly since accumulated liabilities from round $t - 1$ affect the future states of the model; in a sense, the behavior is closer to the dynamics in queuing network models~\cite{kelly2014stochastic,banerjee2016pricing}. 
On the other hand, this also suggests that insights from the financial contagion model can be used to study interventions in more general network allocation problems, such as in ridesharing~\cite{banerjee2016pricing}.

\medskip

\noindent \textbf{Dynamic Clearing.} The work of \cite{barratt2020multi} investigates multi-period liability clearing mechanisms by formulating convex optimal control problems on a different model from the Eisenberg-Noe we study in this paper. In their model, an initial liability matrix $L_0$ is about to be (perhaps partially) cleared at a finite horizon $T$. The entities have some cash held at each time $t$, and they try to clear as many liabilities as possible. The work of \cite{barratt2020multi} waives the obligation of each node to repay its debts in a \emph{pro-rata} fashion. We believe that in a resource allocation setting, equitably splitting the debts is fairer on first principles. Moreover, their work does not extend to discrete interventions and fairness constraints, as we consider in our work. 

In a similar flavor, the work of \cite{feinstein2019dynamic} considers a continuous-time version of the Eisenberg-Noe model and studies the default risk. The paper investigates defaults that can result from insolvency or illiquidity and identifies the default times of nodes. The authors, however, do not study network interventions.

Another work closely tied to ours is the result of \cite{capponi2015systemic}, where the authors study a multi-period clearing framework, where the level of systemic risk is mitigated through the provision of assistance. The network is stochastic, and defaulted entities are sold through a first-price sealed-bid auction. The authors also empirically find that in core-periphery network architectures targeting the systemically important entities is more effective, and in random topologies targeting entities that maximize the total liquidity is more effective. While being close to our work,~\cite{capponi2015systemic} differs from us in that while they mainly consider heuristic policies, we derive optimal policies for continuous interventions and approximately-optimal policies for discrete interventions.

Finally, we mention the contemporary work of \cite{calafiore2022control}, closely related to ours. The authors study a dynamic extension of the EN model and design optimal intervention policies, assuming that the relative liability matrix remains \emph{constant} during the dynamic process. We note here that our framework is more general than \cite{calafiore2022control} since we consider the more general case where the relative liability matrix (see \cref{sec:model}) varies with time which makes the resulting dynamics non-convex. In \cref{sec:general_dynamics}, we prove that the necessary and sufficient condition for the dynamics to be convex is that the relative liability matrix does not vary with time.

\section{Setting}
\label{sec:model}

\begin{figure}
    \centering
    \includegraphics[width=0.7\textwidth]{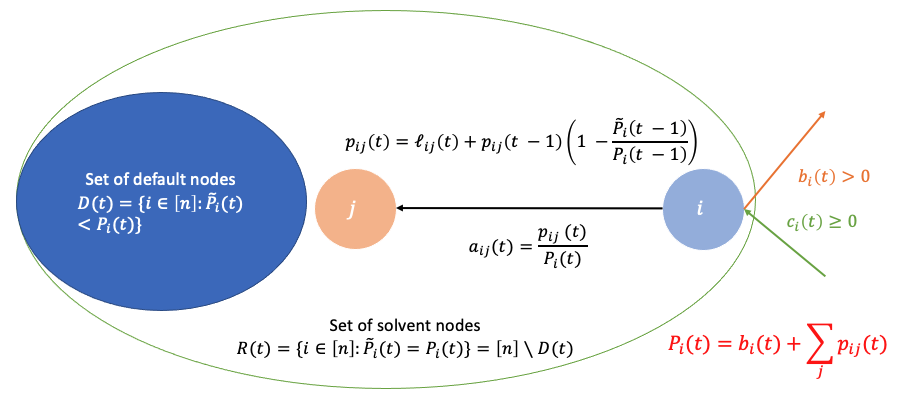}
    \caption{Notation and Problem Setting.}
    \label{fig:setting}
\end{figure}

\noindent \textbf{Notation.} 
We use $[n]$ to denote the set $\{ 1, \dots, n \}$. For any pair of sets $A,B$ we use $A \sqcup B$ to denote the disjoint union of $A$ and $B$ ($A \cap B = \emptyset$).  

For vectors, we use $\| x \|_p$ for the $p$-norm of $x$; for the Euclidean norm (i.e., $p = 2$), we omit the subscript. $\zero$ (resp. $\one$) denotes the all zeros (resp. all ones) column vector, and $\one_S$ represents the indicator column vector of the set $S$. 
We use $x \wedge y$ (resp. $x \vee y$) as shorthand for the coordinate-wise minimum (resp. maximum) of vectors $x$ and $y$, and $(x)^+ = \zero \vee x, (x)^- = - (x \wedge \zero)$ denotes the non-negative and non-positive parts of $x$. 
For a given vector $x$, we use the array notation $x(i : j)$ to denote a sub-vector of $x$ from $x_i$ to $x_j$ (inclusive range).
Finally order relations $\ge, \le, >, <$ denote coordinate-wise ordering. 

For matrices, we use $\| A \|_p$ to denote the induced $p$-norm of $A$. 
$A \odot B$ is the Hadamard (element-wise) product of $A, B$. $A^{(k)}$ is defined as $A^{(0)} = I$ and $A^{(k + 1)} = A^{(k)} \odot A$ for $k \ge 1$. 
Finally, we use ``round'' and ``time'' interchangeably to denote the temporal dimension of the problem, and use $\til O ( \cdot )$ to denote scaling up to poly-logarithmic factors.

\subsection{Model for Networked Interactions Over Time}
\label{subsec:controlmodel}

The model we study (see \cref{fig:setting}) is a natural dynamic extension of the Eisenberg-Noe model~\cite{eisenberg2001systemic}.
The system consists of $n$ entities $[n]$, connected via a dynamic network, where each \emph{directed} edge $(i,j)$ denotes that entity $i$ owes a liability to entity $j$. 

At the start of each (discrete) round $t$, new \emph{internal} liabilities $\ell_{ij}(t) \ge 0$ get generated in the system. Moreover, let $\til P(t - 1)$ denote the \emph{clearing vector} from round $t-1$, i.e., each agent clears $\til P_i(t-1)$ of its liabilities in round $t-1$ (with the rest of the liabilities getting forwarded in time). In accordance with the {proportional clearing rule} of the EN model, we have that the same fraction  $( 1 - \frac {\til P_i(t - 1)} {P_i(t - 1)})$ of each of agent $i$'s liabilities are cleared. 
Thus, the total liability between $i, j \in [n]$ \emph{at the start of round $t$} (i.e., before clearing) is given by
\begin{equation*}
    p_{ij}(t) = \underbrace {\ell_{ij}(t)}_{\text{liability generated in round $t$}} + \underbrace {p_{ij}(t - 1) \cdot \left ( 1 - \frac {\til P_i(t - 1)} {P_i(t - 1)} \right )}_{\text{liabilities carried over from round $t - 1$}},
\end{equation*}
In addition, agents also experience external liabilities $b_i(t) > 0$, generated at the start of each round $t$. 
Let $\ell_i(t) = \sum_{j \in [n]} \ell_{ij}(t)$; then the total liabilities owed by $i$ at the start of round $t$ are
\begin{equation*}
    P_i(t) = b_i(t) + \sum_{j \in [n]} p_{ij} (t) = \underbrace{b_i(t) + \ell_i(t)}_{\text{instantaneous liabilities}} + \underbrace{(P_i(t - 1) - \til P_i(t - 1))}_{\text{liabilities carried over from time $t - 1$}}.
\end{equation*}
This induces the liability network in round $t$, with the (weighted, directed) relative liability matrix $A(t)$ given by
\begin{equation*}
a_{ij}(t) = \begin{cases} \frac {p_{ij}(t)} {P_i(t)} & P_i(t) > 0 \\ 0 & \text{otherwise} \end{cases}.
\end{equation*}
{Now let $\beta_i(t) = \sum_{j \in [n]} a_{ij}(t)$ denote the \emph{fractional internal liability} for any node $i$, and $\beta(t)$ the vector of these internal liabilities.
Throughout the rest of the paper we assume the system has \emph{non-vanishing external liabilities}, i.e., that the following holds for $A(t)$:}
\begin{assumption}[Non-vanishing external liabilities] \label{assumption:uniqueness}
    $\| \beta(t) \|_\infty = \| A^T(t) \|_1 < 1$ for all $t \in [T]$.
\end{assumption}

In order for the overall system to clear their liabilities (in particular, given there are always external liabilities), the EN model assumes that each agent $i$ has some additional ``assets'' $c_i$ which can contribute to clearing their liabilities.
In the same vein, in a dynamic interaction setting, we assume each agent has external assets (or ``revenue streams''), which in each round $t$ provide an instantaneous revenue $c(t) = \{c_i(t)\}_{i \in [n]} \ge \zero$. 
Now, as in the EN model, the clearing vectors $\til P(t)$ must satisfy the following \emph{zero-input dynamics} constraints 
\begin{subequations} 
\label{eq:zero_input_dynamics}     
\begin{align} 
\til P(t) & \le P(t) = b(t) + \ell(t) + P(t - 1) - \til P(t - 1) 
\label{eq:zero_input_dynamics_solvency}     
\\
\til P(t)  &\le A^T(t) \til P(t) + c(t) \label{eq:zero_input_dynamics_default}     
\\ 
\til P(t) & \ge 0.\nonumber
\end{align}
\end{subequations}
We refer to first constraint (\cref{eq:zero_input_dynamics_solvency}) as the \emph{solvency constraint}, since if for some node $i$ it holds with equality, it means that this node is able to repay its debts in full. 
We refer to the second constraint (\cref{eq:zero_input_dynamics_default}) as the \emph{default constraint}, since when it holds with equality for some node $i \in [n]$, this means that $i$ partially repays its debts proportionally to its creditors. Finally, clearing vectors are always non-negative. Note by definition the bounds in~\cref{eq:zero_input_dynamics_solvency,eq:zero_input_dynamics_default} are non-negative; we can thus compress these constraints as $\til P(t) \in [0,P(t)\wedge (A^T(t) \til P(t) + c(t))]$.

\subsection{Dynamic Control of Networked Interactions Process}
\label{subsec:networkmodel}

Till now, we have assumed that all inputs (internal/external liabilities, and external asset values) are exogenous. We now augment this model with an additional \emph{centralized controller}, who is provided with some resource budget in each round, and can allocate this budget to try and minimize defaults.

\noindent \emph{Actions.} Consider a controller (or planner) who has access to a bounded amount of resource $B \ge 0$ at each round, and seeks to inject this into the network (i.e. allocate a fractional quantity $Z_i(t) \ge \zero$ to each agent $i$ subject to the constraint $\one^T Z(t) \le B$). 
The per-round allocation $Z(t)$ may be subject to additional bounds  $Z(t) \le L$ for some given maximum allowed allocation vector $L$; ignoring this is equivalent to letting $L \ge \one B$. 
Therefore, the action space $\mathcal Z$ is given by $\mathcal Z = \{ z \in \mathbb R^n : \| z \|_1 \le B, \; \zero \le z \le L \}$. We always assume the policy function $Z(t)$ is Markovian. 

\noindent \emph{State transitions.}
\cref{eq:zero_input_dynamics} can be modified to incorporate allocations $Z(t)$ to get
\begin{align}
\label{eq:dynamics}  
    \til P(t) & \in [0,P(t)\wedge(A^T(t) \til P(t) + c(t) + Z(t))] \\
    \one^T Z(t) & \le B \nonumber\\
    Z(t) &\in [\zero,L] \nonumber
\end{align}
{Note that in the above equations, $A(t)$ is implicitly a function of $P(t)$, and consequently, this makes the constraints non-linear.
In fact, in~\cref{sec:general_dynamics}, we show that the necessary and sufficient condition for the set of allocations satisfying~\cref{eq:dynamics} to be \emph{convex} (in $\til P(1:T), Z(1:T)$) is that the relative liability matrix $A(t)$ is constant over time (which makes the dynamics linear). }

Next, given the current state $P(t)$, exogenous input $c(t)$, and action $Z(t)$, a natural `maximal' choice of the clearing vector $\til P(t)$ is for it to be the \emph{fixed point} of the system 
\begin{equation}
    S(t) = \begin{pmatrix} P(t) \\ \til P(t) \end{pmatrix} \mapsto \begin{pmatrix} P(t) \\  P(t) \wedge \left (A^T(t) \til P(t) + c(t) + Z(t) \right ) \end{pmatrix} = \Phi \left (S(t), Z(t); S(t - 1), U(t) \right ).
\end{equation}
When the round is clear from the context we will use the abbreviation $\Phi_t(s, z)$ to denote the mapping with information up to time $t$ acting on the state action pair $(s, z)$, i.e. for all $z$ we have that $s(z) = \Phi_t(s(z), z)$. Under \cref{assumption:uniqueness}, we can use the Banach fixed-point theorem to assert that this has a unique solution (since $\Phi_t$ is a contraction with respect to $s$ for a given $z$; see \cite{glasserman2015likely}). 
We henceforth make the following assumption on the agents' response

\begin{assumption}[Maximal clearing] \label{assumption:optimal_response}
In each round $t$, the agents maximally clear their liabilities, i.e., with $\til P(t)$ as the unique fixed point of $\til P(t) = P(t) \wedge \left (A^T(t) \til P(t) + c(t) + Z(t) \right )$.
\end{assumption}

The above condition, taken from the EN model, is standard in the finance literature -- it imposes a natural requirement that agents try and clear liabilities as soon as possible, subject to the proportional clearing rule. If one wants to maximize flows (or minimize defaults), one may be tempted to think that this is without loss of generality. However, this is not the case in dynamic external interventions; one can create examples where dropping this assumption leads to higher overall rewards. These settings, however, are somewhat extreme, and it may be possible to eliminate them via other assumptions.
In \cref{sec:general_dynamics} we exhibit one such natural regime (which also relates to the recent result of \cite{calafiore2022control}); finding more general conditions remains a challenging direction for future work.

\subsection{A Stochastic Model for Exogenous Shocks}
\label{subsec:randomshocks}

Till now, we have been agnostic in our model description as to the exact nature of the exogenous shocks to the system, i.e., the per-round internal and external liabilities, and external asset payouts.
We assume that the environment is \emph{stochastic}, i.e. the instantaneous assets and (internal and external) liabilities induce uncertainty in the system in forms of a \emph{disturbance}. We denote the \emph{financial environment} at round $t$ as $U(t) = \left ( b(t), c(t), \{ \ell_{ij}(t) \}_{i, j \in [n]} \right )$. 
We assume that the financial environment is a \emph{Markov Chain}: 

\begin{assumption}[Markovian Exogenous Shocks] \label{assumption:financial_environment_mc}
    The financial environment $U(t)$ is a Markov Chain, i.e. $\Pr {} {U(t) = u(t) | U(t - 1) = u(t - 1), \dots, U(1) = u(1)} = \Pr {} {U(t) = u(t) | U(t - 1) = u(t - 1)}$.
\end{assumption}

The state space of $U(t)$ is denoted by $\mathcal U$. It is easy to observe that under \cref{assumption:financial_environment_mc} the sequence $S(t) = (P(t), \til P(t))$ is a \emph{Markov Chain}. More specifically, at time $t$ the only information needed to determine $p_{ij}(t)$ is the instantaneous liabilities (which is an MC), the action $Z(t - 1)$ and the remaining liabilities from times $t - 1$, therefore extra information from round $0$ up to $t - 2$ is redundant. Since the external liabilities are also an MC and the sum of $p_{ij}(t)$ only depends on the state of the system at $t - 1$ then $S(t)$ is a MC based wrt to $S (t - 1)$ and $Z(t - 1)$. Also \cref{eq:dynamics} depends only on the state of the system at time $t - 1$ and the calculated maximum liabilities at the start of round $t$ therefore the optimal clearing vector that occurs on the element-wise minimum of the RHS of the inequalities of \cref{eq:dynamics} is dependent on the previous state $S(t - 1)$ and the action $Z(t -  1)$. That defines a transition kernel $\mathcal T$,

\begin{equation}
    \mathcal T ((s, z) \to s') = \mathcal T(s' | s, z) = \Pr {} {S (t) = s' | S (t - 1) = s, Z(t - 1) = z}.
\end{equation}

We also denote the projection on $(s, z)$ of the kernel (which is a distribution itself) as $\mathcal T(\cdot | s, z) = \mathcal T_{s, z}(\cdot)$. The MC is also associated with an initial distribution over the state space $S(0) \sim \pi_0$. The state space of the MC is denoted by $\mathcal S$.

\section{Optimal Network Intervention with Continuous Controls} \label{sec:formulation}

Given the above setting of networked interactions over time with stochastic shocks, We can now formulate the problem of optimal dynamic interventions for maximizing various objectives as a Markov Decision Process (MDP). In this section, we formalize this and show that when interventions are allowed to be continuous, the MDP can be solved optimally for many objectives. Continuous actions are, however often infeasible in practice, and so in the next section, we turn to the question of designing approximately-optimal controllers given discrete actions.

\noindent \emph{Rewards \& Objective.} The stochastic reward incurred by a state-action pair at time $t$ is $R(t) = R(S(t), Z(t), U(t)) = \one^T \til P(t)$. Note that here we can use \emph{any function of $\til P(t)$ that is coordinate-wise strictly increasing} due to \cite[Lemma 4]{eisenberg2001systemic} and get the \emph{same solution}. Other possible reward functions can be, for example, $R(v(t), S(t), Z(t), U(t)) = v^T(t) \til P(t)$ for some function $v(t) > 0$, or $R(S(t), Z(t), U(t)) = \sum_{i \in [n]} \log (\til P_i(t) + \epsilon)$ for some $\epsilon > 0$. 

For simplicity, we have chosen $R(t) = \one^T \til P(t)$ since it corresponds to a measure of how much money circulates in and out of the network. The overall objective that is to be maximized is the sum of rewards over a finite horizon $[T]$, 

\begin{equation} 
\begin{split}
    \max_{\Pi} \quad & \ev {S(0) \sim \pi_0} {\sum_{t = 0}^{T - 1} R(S(t), Z(t) = \Pi(t, S(t)), U(t))} \\
    \text{s.t.} \quad & \eqref{eq:dynamics} \quad \forall t \in [T]
\end{split}
\end{equation}

where $\Pi: [T] \times \mathcal S \to \mathcal Z$ is a policy function. We also assume that there are no accumulated debts and interventions from time $t \le 0$. We let $r(s, z) = \ev {U(t)} {R(S(t) = s, Z(t) = z, U(t))}$. 

\medskip

\noindent \emph{Value Function.} We define the value function $V^{\Pi}(t, s)$ as the optimal reward we can collect from time $t$ onward, starting from state $s$ and applying policy $\Pi$. The value function obeys the HJB equations with respect to actions chosen from the action set, i.e. 

\begin{align*}
    V(t, s) & = \max_{z \in \mathcal Z} \left \{ r(s, z) + \ev {s' \sim \mathcal T_{s, z}} {V(t + 1, s')} \right \} \\
    \Pi^*(t, s) & = \mathrm {argmax}_{z \in \mathcal Z} \left \{ r(s, z) + \ev {s' \sim \mathcal T_{s, z}} {V(t + 1, s')} \right \}
\end{align*}

\cref{fig:ex_fractional_allocations_zero_input} and \cref{fig:ex_fractional_allocations} present an example with a horizon $T = 2$ in a deterministic environment, whereas in \cref{fig:ex_fractional_allocations_zero_input} we consider the \emph{zero-input} dynamics (i.e. no interventions), and in \cref{fig:ex_fractional_allocations} we consider the case when budget is available and is able to avert defaults. 

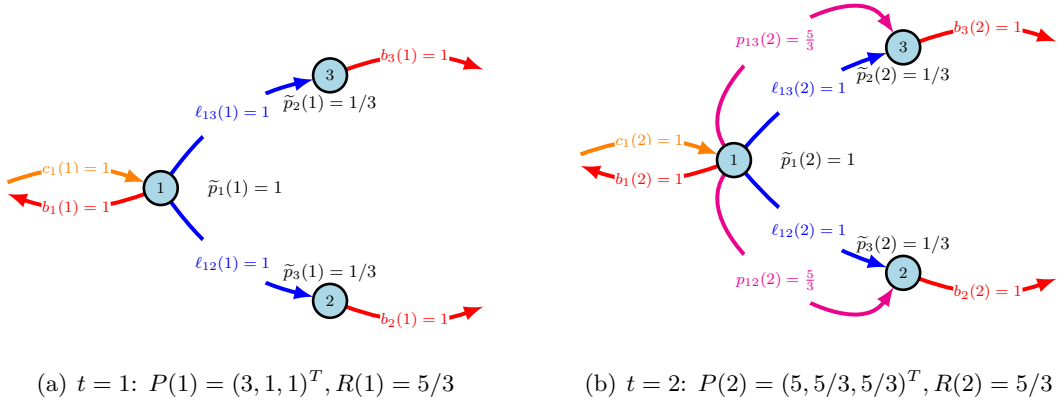
\begin{figure}[t]
    \centering
    \subfigure[$t = 1$: $P(1) = (3, 1, 1)^T, R(1) = 5/3$]{
        \begin{tikzpicture}[transform shape,scale=0.75]
        \Vertex[x=0,y=0,label=$1$]{1}
        \Vertex[x=3,y=-2,label=$2$]{2}
        \Vertex[x=3,y=2,label=$3$]{3}
        \Vertex[x=-3,y=0,Pseudo]{s1}
        \Vertex[x=6,y=-2,Pseudo]{s2}
        \Vertex[x=6,y=2,Pseudo]{s3}
        
        \Text[x=1.5,y=0]{{\footnotesize $\til p_1(1) = 1$}}
        \Text[x=3,y=1.5]{{\footnotesize $\til p_2(1) = 1/3$}}
        \Text[x=3,y=-1.5]{{\footnotesize $\til p_3(1) = 1/3$}}

        \Edge[Direct, bend=-20, label={$\ell_{12}(1) = 1$}, color=blue](1)(2)
        \Edge[Direct, bend=20, label={$\ell_{13}(1) = 1$}, color=blue](1)(3)

        \Edge[Direct, bend=20, label={$c_1(1) = 1$}, color=orange](s1)(1)
        \Edge[Direct, bend=20, label={$b_1(1) = 1$}, color=red](1)(s1)

        \Edge[Direct, bend=-20, label={$b_2(1) = 1$}, color=red](2)(s2)
        \Edge[Direct, bend=20, label={$b_3(1) = 1$}, color=red](3)(s3)
     
    \end{tikzpicture}
    }
    \subfigure[$t = 2$: $P(2) = (5, 5/3, 5/3)^T, R(2) = 5/3$]{
        \begin{tikzpicture}[transform shape,scale=0.75]
        \Vertex[x=0,y=0,label=$1$]{1}
        \Vertex[x=3,y=-2,label=$2$]{2}
        \Vertex[x=3,y=2,label=$3$]{3}
        \Vertex[x=-3,y=0,Pseudo]{s1}
        \Vertex[x=6,y=-2,Pseudo]{s2}
        \Vertex[x=6,y=2,Pseudo]{s3}
        
        \Text[x=1.5,y=0]{{\footnotesize $\til p_1(2) = 1$}}
        \Text[x=3,y=1.5]{{\footnotesize $\til p_2(2) = 1/3$}}
        \Text[x=3,y=-1.5]{{\footnotesize $\til p_3(2) = 1/3$}}

        \Edge[Direct, bend=-15, label={$\ell_{12}(2) = 1$}, color=blue](1)(2)
        \Edge[Direct, bend=15, label={$\ell_{13}(2) = 1$}, color=blue](1)(3)
        
        \Edge[Direct, bend=-90, label={$p_{12}(2) = \tfrac 5 3$}, color=magenta](1)(2)
        \Edge[Direct, bend=90, label={$p_{13}(2) = \tfrac 5 3$}, color=magenta](1)(3)

        \Edge[Direct, bend=20, label={$c_1(2) = 1$}, color=orange](s1)(1)
        \Edge[Direct, bend=20, label={$b_1(2) = 1$}, color=red](1)(s1)

        \Edge[Direct, bend=-20, label={$b_2(2) = 1$}, color=red](2)(s2)
        \Edge[Direct, bend=20, label={$b_3(2) = 1$}, color=red](3)(s3)
     
    \end{tikzpicture}
    }
    \caption{\small \em A contagion network over a horizon $T = 2$, with instantaneous internal liabilities $\ell$, external liabilities $b$, and external assets $c$ that are identical over the two rounds. The total budget is $B = 0$ (zero-input case). At time $t = 1$, the external assets of node 1 are not adequate to cover its liabilities so the node defaults and pays $1/3$ to each of its creditors (i.e. nodes 2, 3 and external). Also, nodes 2, and 3 pay their creditors $1/3$. At time $t = 2$, an excess of liabilities have accumulated due to the inability to pay creditors in full at $t = 1$. The total debt of node 1 is $5$ and nodes 2 and 3 have a total debt of $5/3$. The value function equals $R(1) + R(2) = 10/3$.}
    \label{fig:ex_fractional_allocations_zero_input}
\end{figure}

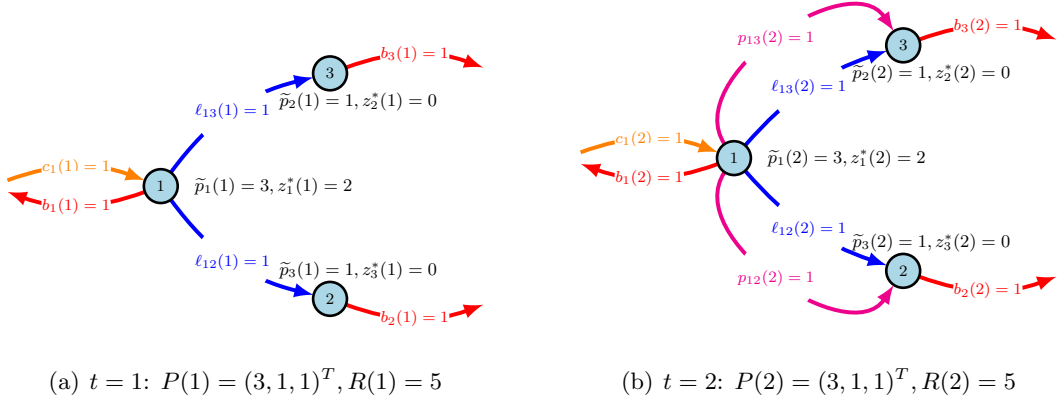
\begin{figure}[t]
    \centering
    \subfigure[$t = 1$: $P(1) = (3, 1, 1)^T, R(1) = 5$]{
        \begin{tikzpicture}[transform shape,scale=0.75]
        \Vertex[x=0,y=0,label=$1$]{1}
        \Vertex[x=3,y=-2,label=$2$]{2}
        \Vertex[x=3,y=2,label=$3$]{3}
        \Vertex[x=-3,y=0,Pseudo]{s1}
        \Vertex[x=6,y=-2,Pseudo]{s2}
        \Vertex[x=6,y=2,Pseudo]{s3}
        
         \Text[x=2,y=0]{{\footnotesize $\til p_1(1) = 3, z_1^*(1) = 2$}}
        \Text[x=3.5,y=1.5]{{\footnotesize $\til p_2(1) = 1, z_2^*(1) = 0$}}
        \Text[x=3.5,y=-1.5]{{\footnotesize $\til p_3(1) = 1, z_3^*(1) = 0$}}

        \Edge[Direct, bend=-20, label={$\ell_{12}(1) = 1$}, color=blue](1)(2)
        \Edge[Direct, bend=20, label={$\ell_{13}(1) = 1$}, color=blue](1)(3)

        \Edge[Direct, bend=20, label={$c_1(1) = 1$}, color=orange](s1)(1)
        \Edge[Direct, bend=20, label={$b_1(1) = 1$}, color=red](1)(s1)

        \Edge[Direct, bend=-20, label={$b_2(1) = 1$}, color=red](2)(s2)
        \Edge[Direct, bend=20, label={$b_3(1) = 1$}, color=red](3)(s3)
     
    \end{tikzpicture}
    }
    \subfigure[$t = 2$: $P(2) = (3, 1, 1)^T, R(2) = 5$]{
        \begin{tikzpicture}[transform shape,scale=0.75]
        \Vertex[x=0,y=0,label=$1$]{1}
        \Vertex[x=3,y=-2,label=$2$]{2}
        \Vertex[x=3,y=2,label=$3$]{3}
        \Vertex[x=-3,y=0,Pseudo]{s1}
        \Vertex[x=6,y=-2,Pseudo]{s2}
        \Vertex[x=6,y=2,Pseudo]{s3}
        
        \Text[x=2,y=0]{{\footnotesize $\til p_1(2) = 3, z_1^*(2) = 2$}}
        \Text[x=3.5,y=1.5]{{\footnotesize $\til p_2(2) = 1, z_2^*(2) = 0$}}
        \Text[x=3.5,y=-1.5]{{\footnotesize $\til p_3(2) = 1, z_3^*(2) = 0$}}

        \Edge[Direct, bend=-15, label={$\ell_{12}(2) = 1$}, color=blue](1)(2)
        \Edge[Direct, bend=15, label={$\ell_{13}(2) = 1$}, color=blue](1)(3)
        
        \Edge[Direct, bend=-90, label={$p_{12}(2) = 1$}, color=magenta](1)(2)
        \Edge[Direct, bend=90, label={$p_{13}(2) = 1$}, color=magenta](1)(3)

        \Edge[Direct, bend=20, label={$c_1(2) = 1$}, color=orange](s1)(1)
        \Edge[Direct, bend=20, label={$b_1(2) = 1$}, color=red](1)(s1)

        \Edge[Direct, bend=-20, label={$b_2(2) = 1$}, color=red](2)(s2)
        \Edge[Direct, bend=20, label={$b_3(2) = 1$}, color=red](3)(s3)
     
    \end{tikzpicture}
    }
    \caption{\small \em A contagion network over a horizon $T = 2$, with instantaneous internal liabilities $\ell$, external liabilities $b$, and external assets $c$ that are as in \cref{fig:ex_fractional_allocations_zero_input}. The total budget is $B = 2$ at each round. The optimal allocation for $t = 1$ is $z^*(1) = (2, 0, 0)^T$ in which case all nodes are able to cover their debts, and no liabilities are carried over from $t = 1$ to $t = 2$. Similarly, in $t = 2$ the optimal intervention vector is $z^*(2) = (2, 0, 0)^T$ and all liabilities are cleared. The value function equals $R(1) + R(2) = 10$.}
    \label{fig:ex_fractional_allocations}
\end{figure}

\medskip

\noindent\textbf{Efficiently Computing the Optimal Value Function}
{The above MDP has a very high-dimensional state and action space ($\mathbb{R}^{2nT}$ and $\mathbb{R}^{nT}$ respectively), so a priori it is unclear if it can be solved efficiently. Surprisingly, we show below that we can exploit the structure of the problem -- in particular, the fact that the random shocks are exogenous (\cref{assumption:financial_environment_mc}), and the maximal clearing assumption (\cref{assumption:optimal_response}) -- to give a \emph{closed-form expression} for the value function as an expectation over the exogenous shock vector $U(1:T)$; moreover, this also allows us to compute it efficiently (and thus find near-optimal policies) via Monte Carlo estimation.
} 

{We now proceed to show how to calculate the value function $V(t, s)$ and the optimal policy $Z^*(t)$. First, it is easy to check that \emph{given} a realization of the random shocks $u(t:T)$, the optimal reward (and policy) can be written as a sequence of \emph{nested linear programs}. 
More surprisingly, we prove that we can exchange the maximum and expectation operators in the value function, due to the structure of our model. 
Consequently, when the shocks are generated randomly, we get that the value function (\cref{theorem:value_function}) can be approximated by sampling $N$ sample paths and then, for each sample path $u(t:T)$ solving a sequence of $T - t + 1$ linear programs.}

{Our algorithm (\cref{alg:fractional_interventions}) is comprised of two routines: The first routine (\textbf{Compute-Value-Function-Given-Sample-Path}) takes as an input a realization $u(t:T)$ of exogenous shocks, a starting state $s(t - 1) = s$, and budget constraints $L$ and $B$ and solves a sequence of $T - t + 1$ nested linear programs, where the optimal solution at round $t$ is fed to calculate the optimal solution at round $t + 1$. The second routine (\textbf{Aggregate}) takes as input a natural number $N$, the budget constraints $L$ and $B$, a financial environment $\mathcal U$, and the starting state $s$. The algorithm then samples $N$ exogenous shock realizations from $\mathcal U$. Conditioned on any of the sample paths $u_i(t:T)$ with $i \in [N]$, it calls the first routine to compute the sample value function $\bar V_{u_i(t:T)}$. Finally it aggregates all solutions and outputs an estimate $\bar V(t, s)$.} 

\frbox {
\small
\captionof{algorithm}{\small \em Dynamic Clearing With Fractional Interventions}
\label{alg:fractional_interventions}
\noindent \textbf{Compute-Value-Function-Given-Sample-Path}($L$, $B$, $u(t:T)$, $s$)

\begin{compactenum}
    \item Given the initial state at round $t - 1$ calculate $A(t)$ and $P(t)$
    \item For each $t' \in [t,T]$ 
    \begin{compactenum}
        \item Let $u(t') = \left (b(t'), c(t'), \mathsf{vec}(\{ \ell_{ij}(t') \}_{i, j \in [n]}) \right )$ be the financial environment.
        \item Let $\til P^*(t'), Z^*(t')$ be the optimal solution to $\max_{\til P(t'), Z(t')} \one^T \til P(t')$ subject to the dynamics of \eqref{eq:dynamics} and the random shocks $b(t'), c(t'), \{ \ell_{ij}(t') \}_{i, j \in [n]}$. 
        \item If $t' < T$, use $\til P^*(t')$ to calculate $A(t'+1)$ and $P(t'+1)$. 
    \end{compactenum}
    \item Return $V_{u(t:T)} = \sum_{t' \in [t,T]} \one^T \til P^*(t')$.
\end{compactenum}

\medskip

\noindent \textbf{Aggregate}($N$, $L$, $B$, $\mathcal U$, $s$)

\begin{compactenum}
    \item Sample $N$ i.i.d. sample paths $\{ u_i(t:T) \}_{i \in [N]} \sim \mathcal U$ where a sample path consists a realization of the environment on $T - t + 1$ periods. 
    \item For every $i \in [N]$ compute $$V_{u_i(t:T)} = \text{\textbf{Compute-Value-Function-Given-Sample-Path}}(L, B, u_i(t:T), s)$$
    \item Return $\bar V(t, s) = \frac 1 N \sum_{i = 1}^N V_{u_i(t:T)}$.
\end{compactenum}
}

\begin{theorem}%
\label{theorem:value_function} 
Under~\cref{assumption:uniqueness,assumption:optimal_response,assumption:financial_environment_mc}, the following are true
\begin{compactenum}
\item The value function $V(t, s)$ satisfies
\begin{equation*} \label{eq:value_function}
    V(t, s) = \ev {U(t:T)} {\max_{z_t, \til p_t} \left \{ \one^T \til p_t + \max_{z_{t +1}, \til p_{t + 1}} \left \{ \one^T \til p_{t + 1} + \max_{z_{t + 2}, \til p_{t + 2}} \left \{ \one^T \til p_{t + 2} + \dots \right \} \right \} \right \}}
\end{equation*}
and corresponds to solving a sequence of linear programs.

\item For $N = \tfrac {\log (2 / \delta) (T - t + 1)^2 \Delta^2 } {2 \varepsilon^2}$ samples, $\varepsilon > 0$, and  $\Delta = \sup_{\mathcal U} \left ( \| b \|_1 + \| \ell \|_1 \right )$, \cref{alg:fractional_interventions} returns a solution $\bar V(t, s)$ such that $|\bar V(t, s) - V(t, s)| \le \varepsilon$ with probability at least $1 - \delta$. 
    \end{compactenum}

\end{theorem}

The proof of this result is given in \cref{sec:proofs}.

\section{Network Intervention with Discrete Controls} 
\label{sec:discrete_allocations}

We next focus on the problem of allocating discrete interventions. For the discrete interventions problem, each node can be allocated discrete resources up to some value $L_j \in \mathbb N$. A simpler version of the problem studied in~\cite{papachristou2021allocating} allowed the interventions to get two distinct values $\{ 0, L_j \}$. The analysis in this case is exactly the same with the general case. The total budget is again $B$ and does not change with time as well. We refer the the action space of this setting with $\mathcal Z_d = \left \{ z \in \mathbb N^n : \| z \|_1 \le B, \zero \le z \le L \right \}$. Note that $\mathcal Z$ defined on \cref{sec:formulation} corresponds to the \emph{fractional relaxation} of $\mathcal Z_d$.

We again seek to find the optimal policy which maximizes the value function at round $t = 0$, subject to the dynamics 

\begin{equation}
    S(t) = \begin{pmatrix} P(t) \\ \til P(t) \end{pmatrix} \mapsto \begin{pmatrix} P(t) \\  P(t) \wedge \left (A^T(t) \til P(t) + c(t) + Z(t) \right ) \end{pmatrix} = \Psi \left (S(t), Z(t); S(t - 1), U(t) \right ).
\end{equation}

Again, when the round is clear from the context we will use the abbreviation $\Psi_t(s, z)$ to denote the mapping with information up to time $t$ acting on the state action pair $(s, z)$, i.e. for all $z$ we have that $s(z) = \Psi_t(s(z), z)$. Similarly to $\Phi_t$, the operator $\Psi_t$ is a contraction according to \cref{assumption:uniqueness}. In \cite[Theorem 1]{papachristou2021allocating} it has been proven that the same problem is NP-Hard for $T = 1$ by reduction from the Set Cover problem and therefore, the problem in question is at least as hard as the combinatorial optimization problem of \cite{papachristou2021allocating}. Therefore, we seek an approximate policy which yields an \emph{approximate value function} $V^{SOL}(t, s)$ such that 

\begin{equation}
    V^{SOL} (t, s) \ge \left ( 1 - \gamma (t, s) \right ) \cdot V^{OPT}(t, s) \quad \text{for some} \quad \gamma(t,s) \in (0, 1).
\end{equation}

For the problem with $T = 1$ the work of \cite{papachristou2021allocating} obtains a $(1 - \beta_{\max})$-approximation for the problem with a randomized rounding algorithm where $\beta_{\max}$ is the maximum row sum of the relative liability matrix. Because the matrix $A$ will be different in both cases which would correspond to different clearing solutions. However, note that a rounding regime that iteratively rounds the fractional optimal solutions in a backward function will not yield correct results. The reason for that is that a suboptimal action at round $t + 1$ can have an effect on the optimal fractional action at round $t$.

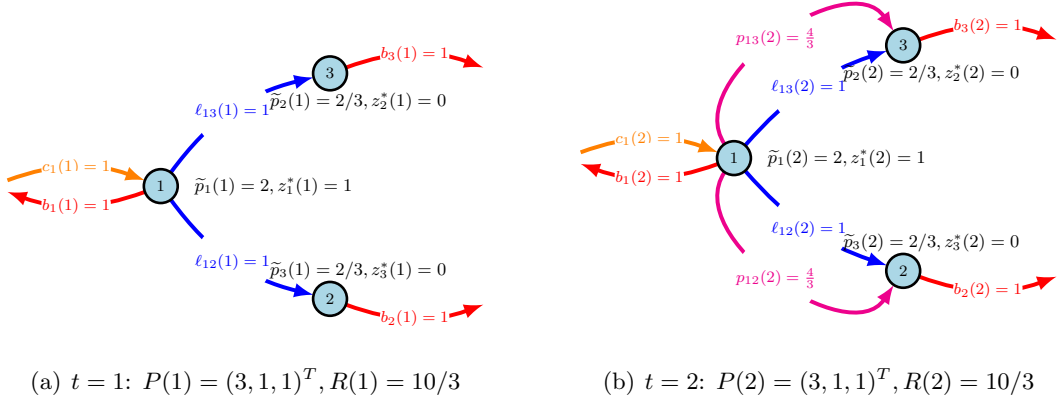
\begin{figure}[t]
    \centering
    \subfigure[$t = 1$: $P(1) = (3, 1, 1)^T, R(1) = 10/3$]{
        \begin{tikzpicture}[transform shape,scale=0.75]
        \Vertex[x=0,y=0,label=$1$]{1}
        \Vertex[x=3,y=-2,label=$2$]{2}
        \Vertex[x=3,y=2,label=$3$]{3}
        \Vertex[x=-3,y=0,Pseudo]{s1}
        \Vertex[x=6,y=-2,Pseudo]{s2}
        \Vertex[x=6,y=2,Pseudo]{s3}
        
         \Text[x=2,y=0]{{\footnotesize $\til p_1(1) = 2, z_1^*(1) = 1$}}
        \Text[x=3.5,y=1.5]{{\footnotesize $\til p_2(1) = 2/3, z_2^*(1) = 0$}}
        \Text[x=3.5,y=-1.5]{{\footnotesize $\til p_3(1) = 2/3, z_3^*(1) = 0$}}

        \Edge[Direct, bend=-20, label={$\ell_{12}(1) = 1$}, color=blue](1)(2)
        \Edge[Direct, bend=20, label={$\ell_{13}(1) = 1$}, color=blue](1)(3)

        \Edge[Direct, bend=20, label={$c_1(1) = 1$}, color=orange](s1)(1)
        \Edge[Direct, bend=20, label={$b_1(1) = 1$}, color=red](1)(s1)

        \Edge[Direct, bend=-20, label={$b_2(1) = 1$}, color=red](2)(s2)
        \Edge[Direct, bend=20, label={$b_3(1) = 1$}, color=red](3)(s3)
     
    \end{tikzpicture}
    }
    \subfigure[$t = 2$: $P(2) = (3, 1, 1)^T, R(2) = 10/3$]{
        \begin{tikzpicture}[transform shape,scale=0.75]
        \Vertex[x=0,y=0,label=$1$]{1}
        \Vertex[x=3,y=-2,label=$2$]{2}
        \Vertex[x=3,y=2,label=$3$]{3}
        \Vertex[x=-3,y=0,Pseudo]{s1}
        \Vertex[x=6,y=-2,Pseudo]{s2}
        \Vertex[x=6,y=2,Pseudo]{s3}
        
        \Text[x=2,y=0]{{\footnotesize $\til p_1(2) = 2, z_1^*(2) = 1$}}
        \Text[x=3.5,y=1.5]{{\footnotesize $\til p_2(2) = 2/3, z_2^*(2) = 0$}}
        \Text[x=3.5,y=-1.5]{{\footnotesize $\til p_3(2) = 2/3, z_3^*(2) = 0$}}

        \Edge[Direct, bend=-15, label={$\ell_{12}(2) = 1$}, color=blue](1)(2)
        \Edge[Direct, bend=15, label={$\ell_{13}(2) = 1$}, color=blue](1)(3)
        
        \Edge[Direct, bend=-90, label={$p_{12}(2) = \tfrac 4 3$}, color=magenta](1)(2)
        \Edge[Direct, bend=90, label={$p_{13}(2) = \tfrac 4 3$}, color=magenta](1)(3)

        \Edge[Direct, bend=20, label={$c_1(2) = 1$}, color=orange](s1)(1)
        \Edge[Direct, bend=20, label={$b_1(2) = 1$}, color=red](1)(s1)

        \Edge[Direct, bend=-20, label={$b_2(2) = 1$}, color=red](2)(s2)
        \Edge[Direct, bend=20, label={$b_3(2) = 1$}, color=red](3)(s3)
     
    \end{tikzpicture}
    }
    \caption{\small \em A contagion network over a horizon $T = 2$, with instantaneous internal liabilities $\ell$, external liabilities $b$, and external assets $c$ that are as in \cref{fig:ex_fractional_allocations_zero_input}. We consider discrete allocations with $B = 1$ and $L = \one$. The optimal discrete allocation for $t \in \{1, 2 \}$ is $z^*(t) = (1, 0, 0)^T$. Similarly with \cref{fig:ex_fractional_allocations}, the interventions increase node 1's assets and thus it's able to pay a total of $2/3$ to each creditor (on both steps), yielding a value function of $R(1) + R(2) = 20/3$.}
    \label{fig:ex_discrete_allocations}
\end{figure}

\subsection{Approximation Algorithms}

The first idea on deriving an approximation algorithm for the problem described is adapting the approximation algorithm presented in \cite{papachristou2021allocating} to the dynamic setting: More specifically let $t$ be a fixed round and $s(t)$ be a fixed state at round $t$ and let $u(t:T)$ be a sample path for the random sequence $U(t:T)$ from time $t$ onward. Conditioned on the realization of $u(t:T)$ the planner seeks to solve the following deterministic problem 

\begin{equation} \label{eq:optimization_discrete_interventions}
\begin{split}
    \max_{z(t:T)} \quad & V_{u(t:T)} (t, s(t)) = \sum_{t' = t}^T R (s(t'), z(t') = \Pi (t', s(t')), u(t')) \\
    \text{s.t.} \quad & s(t') = \Psi_{t'} (s(t'), z(t'), u(t')) \quad \forall t' \in [t, T] \\
    & z(t') \in \mathcal Z_d \quad \forall t' \in [t, T].
\end{split}
\end{equation}

We also consider the fractional relaxation where the decision variables lie in $\mathcal Z$. Let $V^{SOL}_{u(t:T)}(t, s(t))$ be the value that the approximation algorithm produces, let $V^{REL}_{u(t:T)}(t, s(t))$ be the optimal solution of the relaxation (i.e. when the interventions belong to $\mathcal Z$) and let $V^{OPT}_{u(t:T)}(t, s(t))$ be the optimal solution. From optimality we know that $V^{REL}_{u(t:T)}(t, s(t)) \ge V^{OPT}_{u(t:T)}(t, s(t))$. Note that, since $u(t:T)$ is given, the optimal policy for the relaxed program consists of solving $T - t + 1$ linear programs sequentially and finding the pair of the clearing vector and optimal policy. We produce a rounded policy $SOL$ that corresponds to an allocation vector $z_d(t:T)$ randomly by rounding the matrix $z_r^*(t:T)$ of the optimal relaxed policy by rounding all random variables $z_{d, i}(t')$ as Binomial i.i.d. vectors on $L_i$ trials with biases $z_{r, i}^*(t') / L_i$ for all $t' \in [t, T]$. Our algorithm is presented below: 

\frbox{
    \small
    \captionof{algorithm}{\small \em Randomized Rounding Approximation Algorithm} \label{alg:randomized_rounding}
      \noindent \textbf{Sample-Interventions}($L$, $B$, $u(t:T)$, $\tau$, $s$)
      
    \begin{compactenum}
        \item Until constraints are satisfied and the approximation guarantee is not violated or $\tau$ iterations have elapsed
        \begin{compactenum}
            \item For every agent $i \in [n]$ sample (independent) interventions
            
            $$z_{d, i}(t:T) \sim \mathrm {Bin} \left ( \frac {z_{r, i}^*(t:T)} {L_i}, L_i \right )$$.
            
        \end{compactenum}
        
        \item Return the value function  $V^{SOL}_{u(t:T)} = \sum_{t' \in [t, T]} \one^T \til P_d(t)$ given the calculated approximate ($SOL$) policy after calculating the clearing payments $\til P_d(t:T)$. 
        
    \end{compactenum}
    
     \medskip
    \noindent \textbf{Aggregate-Discrete}($N$, $L$, $B$, $\tau$, $\mathcal U$, $s$)
    
    \begin{compactenum}
  
    \item  Let $\Delta = \sup_{u \in \mathcal U} \left ( \| b \|_1 + \| \ell \|_1 \right ), \; \Theta = \min \{ B, \min_{u \in \mathcal U} \| b \|_1 \}$. 
    \item Sample $N$ exogenous shocks $\{u_j(t:T)\}_{j \in [N]} \sim \mathcal U$
    
    \item For $j \in [N]$
    \begin{compactenum}
        \item Call \textbf{Compute-Value-Function-Given-Sample-Path}($L, B, u_j(t:T), s$) and get the optimal fractional policy $z_r^*(t:T)$.
        \item Calculate 
        $$V_{u_j(t:T)}^{SOL} = \text{\textbf{Sample-Interventions}}(L, B, u_j(t:T), \tau, s)$$
    \end{compactenum}
    
    \item Return $ \bar V^{SOL} = \frac {1} {N} \sum_{j = 1}^{N} V^{SOL}_{u_{j}(t:T)} $
    \end{compactenum}
}

\begin{theorem} \label{theorem:approximation_randomized_rounding}
    \cref{alg:randomized_rounding} yields the following approximation guarantee (on expectation):
   
    \begin{equation*} 
        \ev {u(t:T), z_d(t:T)} {V^{SOL} (t, s(t))} \ge \left ( 1 -  \sup_{u(t:T)} \max_{t' \in [t,T], i \in [n]} \beta_i(t') \right ) \cdot \ev {u(t:T)} {V^{OPT} (t, s(t))},
    \end{equation*}

    where $\sup_{u(t:T)} \max_{t' \in [t,T], i \in [n]} \beta_i(t')$ is the worst-case financial connectivity of a deterministic contagion network produced in the environment space $\mathcal U$. Moreover, if the minimum external liability for any node $i \in [n]$ at any round $t' \in [t,T]$ is $\delta_b$, $\Delta = \sup_{\mathcal U} \left ( \| b \|_1 + \| \ell \|_1 \right )$, then the approximation ratio is lower bounded by
    
    \begin{equation*}
        1 - \sup_{u(t:T)} \max_{t' \in [t, T], i \in [n]} \beta_i(t') \ge \begin{cases}
            \frac {\delta_b} {\Delta} &  B > \Delta \\
            \frac {\delta_b} {(T - t + 1) \Delta} & \text{otherwise}
        \end{cases}.
    \end{equation*}
\end{theorem}

Note that the proof of \cref{theorem:approximation_randomized_rounding} there is no dependence on the states created by the rounded outcomes, and therefore we can compare and lower bound by the value of the reward at each round $t$ of the fractionally relaxed policy. Specifically, the proof of \cref{theorem:approximation_randomized_rounding} lies on the observation that under a realization of the financial environment a default node on the rounded solution can serve at least its external assets plus the intervention it gets, so, in expectation, it can serve at least its external assets plus the intervention of the fractional solution. For the optimal solution of the fractional relaxation we know \cite{glasserman2015likely, papachristou2021allocating} that for every $S \subseteq [n]$ the sum of the external assets and the fractional interventions over $S$ is at least the sum of payments weighted by the ``negated'' financial connectivities. Moreover, for all the solvent nodes the amount of payments they can serve is at least the corresponding fractional payment weighted by the connectivity of the node. Combining both observations, we get that the value function of the rounded solution is at least a factor of 1 minus the worst financial connectivity across (the remaining) rounds of the optimal solution. Simplifying the approximation factor we make two observations: firstly if the system is able to clear all the liabilities, then the approximation factor simplifies to $\delta_b / \Delta$, similar to the static case, has no dependency on the time horizon, and is only instance-dependent (i.e. it depends on $\mathcal U$). Otherwise, the approximation factor lower bound loses an $\Omega (T^{-1})$ factor.

For the runtime of the algorithm, we give the following guarantee when the interventions are equal for all nodes. A similar analysis (see also \cite{papachristou2021allocating}) can be performed when the interventions are different using the dependent rounding scheme of \cite{srinivasan2001distributions} and achieves the same runtime guarantee as \cref{theorem:runtime}: 

\begin{theorem}[Runtime] \label{theorem:runtime} 
    
    \cref{alg:randomized_rounding} run for $N = \tfrac {\tau \Delta^2} {2 \Theta^2}$ realizations, where for each realization \textbf{Sample-Interventions} is called at most $\tau$ times, with interventions $L = \lambda \cdot \one$ (for $\lambda \in \mathbb N^*$) and $\Delta, \Theta$ defined in \cref{alg:randomized_rounding} yields an approximation guarantee of $\big ( 1 - \max_{t' \in [t,T]} \max_{i \in [n]} \beta_i(t') - 2 \varepsilon \big )$ with probability $1 - O \left (\tfrac {(T - t + 1) \tau \Delta^2 e^{-\tau \varepsilon^2}} {\Theta^2} \right )$ and runs in time $O((T - t) (\mathcal T_{LP} + \tau n \lambda))$. $\mathcal T_{LP}$ is the runtime for solving the one-step clearing problem, assuming $O(1)$ access to $[0,1]$-uniform random variables.  
\end{theorem}

\medskip 

\noindent \emph{Discretizing the Clearing Payments.} Once a solution to the fractional (resp. discrete) intervention problem has been obtained an interesting direction is rounding the clearing payment vector, i.e. enforce integral clearing payments to the nodes. An algorithm to do this would be the following simple algorithm: For every $t \in [T]$ sort $\left (y_{i1}(t) = a_{i1}(t), \dots, y_{in}(t) = a_{in}(t), y_{i,n+1}(t) = 1 - \beta_{i}(t) \right )$ in decreasing order and obtain a permutation of the nodes $\sigma_{i, t} : [n + 1]\to [n + 1] $. Then we get the first element and subtract $\lfloor y_{i,\sigma_{i, t}(1)} \til P_i^*(t) \rfloor$ from $\til P_i^*(t)$, get the second element, and so on, until no further rounding of the payments can be performed (i.e. at $j' \in [n]$ such that $\sum_{j \le j'} \lfloor y_{i,\sigma_{i, t}(j)} \til P_i^*(t) \rfloor > \til P_i^*(t)$). Closing, it is an interesting problem to see how this solution compares to the optimal fractional solution (see also \cite{bomze2014rounding} for similar error bounds). 

\section{Incorporating Fairness Considerations} 
\label{sec:fairness}

We say that an allocation is fair across the nodes if the intervention each node gets ``does not differ a lot from its neighbors'', whereas the ``neighborhood'' of a node can be expressed in terms of the existing financial network or can be expressed in terms of an auxiliary network as in \cite{papachristou2021allocating}. In this way, we can measure fairness in allocations in various settings. For instance, we can compare the interventions between a node and all other nodes in the network, interventions between a node and its neighbors on the financial network, and interventions between nodes belonging to different population groups (such as minority groups). All fairness metrics have to be scale-invariant, namely, do not change when the budget provided changes from $B$ to $\alpha B$ for some $\alpha > 0$. 

Motivated by the above desiderata, we call an allocation rule $Z(1:T)$ in the model \emph{fair} if the allocations of a node do not \emph{``differ much''} from its neighbors. As a starting point, we consider the Gini Coefficient \cite{gini1921measurement} and generalize it accordingly to our model. In detail, we measure the deviation between a node and its neighbors on a graph sequence $\{ H_t \}_{t \in [T]}$ with weights $w_{ij}(t) \ge 0$ ($w_{ii}(t) = 0$ for all $i \in [n], t \in [T]$), with the following measure 

\begin{equation}
    \mathrm{GC}(t; H_t) = \frac {\sum_{(i, j) \in E(H_t)} w_{ij}(t) |Z_i(t) - Z_j(t)|} {\sum_{i \in [n]} Z_i(t) \left ( \sum_{j \in [n]} (w_{ij}(t) + w_{ji}(t)) \right ) }.
\end{equation}

Note that the above measure of inequality is well defined: when all nodes get the same allocations it equals zero and when one node gets all the allocations it equals one. Regarding our model, examples of the sequence $\{ H_t \}_{t \in [T]}$ include:

\begin{compactitem}
    \item Setting $w_{ij}(t) = \one \{ i \neq j \}$ yields the standard Gini Coefficient (GC). Such a measure does not capture the topology of the problem, and aims to distribute interventions equitably across all nodes in the network. 
    \item Setting $w_{ij}(t) = a_{ij}(t)$, i.e. the weights between the nodes represent the actual relative liabilities between such pairs of nodes. This measure takes into account spatial interactions and the strength of ties (i.e. the relative importance of liabilities) to distribute the resources. We call this fairness constraint the Spatial Gini Coefficient (SGC). Note that a similar fairness measure has been considered in \cite{papachristou2021allocating}, however the measure defined there is not symmetric like in our case.
    \item If every node is associated with a sensitive attribute and $q_i(t) \in [0, 1]$ is the probability of the node having this sensitive attribute, we can for example use $w_{ij}(t) = |q_i(t) - q_j(t))| \cdot \one \{ a_{ij}(t) > 0 \}$ to put high weights on adjacent pairs which deviate in this attribute. For instance, $q_i(t)$ can represent the probability of a node belonging to a minority group and thus the weights would give importance in mitigating inequalities between neighboring groups of minorities and non-minorities. Moreover, if we want to enforce a stronger version of fairness, we can consider $w_{ij}(t) = |q_i(t) - q_j(t)|$ which penalizes all deviations in allocations between nodes with high deviations in their sensitive attribute. If more than one sensitive attributes are present, the weights can be modified accordingly to capture the average (or maximum) deviation of the attributes between a pair of nodes. We call this coefficient the Property Gini Coefficient (PGC). 
\end{compactitem}

This motivates the definition of a $g(t)$-fair allocation to be the allocation policy $Z(1:T)$ for which $\mathrm{GC}(t; H_t) \le g(t)$ for all $t \in [T]$ for some function $g(1:T) \in [0, 1]$ that does \emph{not} depend on the clearing payments and the interventions. {In our experiments we use $g(t) = \text{constant}$}. This corresponds to the following additional linear constraints on the \emph{action space} for an additional set of decision variables $\varpi(t) \in \mathbb R^{|E(H_t)|}$

\begin{equation} \label{eq:sgc_constraints}
    \begin{split}
        \varpi_{ij}(t) \ge 0 & \quad  \forall (i, j) \in E(H_t), t \in [T] \\
        - \varpi_{ij}(t) \le Z_i(t) - Z_j(t) \le \varpi_{ij}(t) & \quad \forall (i, j) \in E(H_t), t \in [T] \\
        \sum_{(i, j) \in E(H_t)} a_{ij}(t) \varpi_{ij}(t)  \le g(t) \sum_{i \in [n]} \left ( \sum_{j \in [n]} (w_{ij}(t) + w_{ji}(t) \right ) Z_i(t) & \quad \forall t \in [T].
    \end{split}
\end{equation}

After imposing the fairness constraints, a question one might ask is \emph{``How does the optimal value function without the fairness constraints compares to the optimal value function with fairness constraints for any type of the (generalized) Gini Coefficient?''}. For this reason, we define the Price of Fairness (PoF) to be 

\begin{equation}
    \mathrm{PoF} = \frac {\ev {} {\text{OPT sans fairness}}} {\ev {} {\text{OPT with fairness}}}.
\end{equation}

Since the fairness-constrained setting has additional constrains compared to the no-fairness setting, the price of fairness is always at least 1. In the static version, \cite{papachristou2021allocating} proves that under discrete allocations there exist instances where the PoF can be unbounded, yet when the allocations are fractional the PoF is always bounded. Subsequently, in the dynamic setting it is easy to observe that the same result holds. Finally, in \cref{sec:experiments_fairness} we show how incorporating further fairness constraints to the problem affects the distribution of interventions in our datasets with respect to the nodes' financial connectivities and we also give quantitative results regarding the PoF. 

\section{Experiments} \label{sec:experiments}

\subsection{Experimental Setup} 

We run experiments in various settings, i.e., synthetic core-periphery data, ridesharing data, Web financial transaction data, and semi-artificial data from cellphone mobility networks. We generally run experiments in two settings: fractional and discrete interventions. For the fractional intervention setting, we report the payments, cumulative reward, and interventions for the datasets in question. Since we are solving the problem optimally, we do not report competing methods. For the discrete intervention setting, the work of \cite{papachristou2021allocating} has shown that the randomized rounding algorithm performs very well in practice, and it outperforms various heuristics regarding allocations, even in more limited settings, i.e., whereas the decision was either to take the resource or not.
For this reason, we expect that the various heuristics (and their subsequent adaptations) in \cite{papachristou2021allocating} would be significantly outperformed by the randomized rounding method in the dynamic setting, and therefore we omit comparisons with these heuristic methods. The work of \cite{papachristou2021allocating} also presents a greedy hill-climbing algorithm that, under specific conditions, gives an approximation guarantee. However, this algorithm is designed for allocating resources in a binary way (the resource is either taken or not taken) and extending the analysis to a dynamical setting, and a more general allocation rule is hard.

\subsection{Stochastic Blockmodel with Core-Periphery Structure} \label{sec:synthetic_data}

\begin{figure}
\centering
\subfigure[Clearing Payments (Fractional interventions)]{\includegraphics[width=0.49\textwidth]{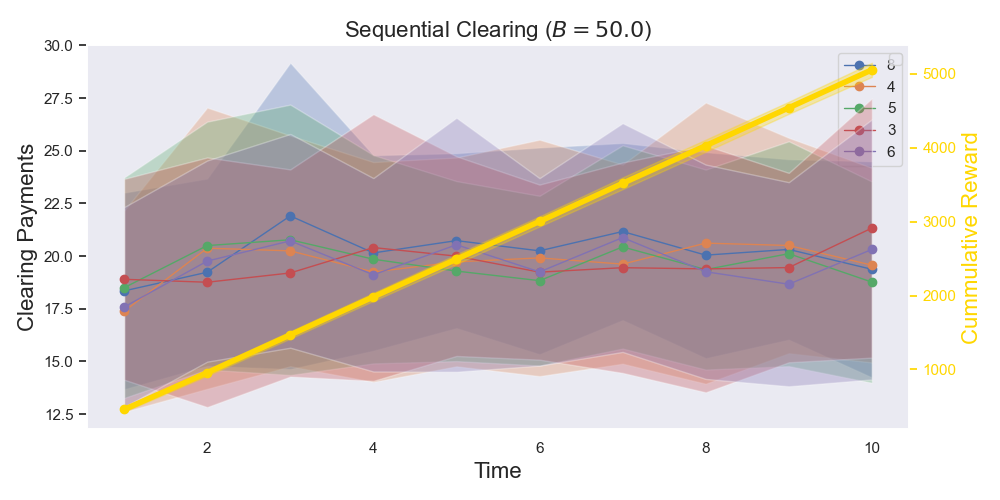}}
    \subfigure[Interventions (Fractional)]{\includegraphics[width=0.49\textwidth]{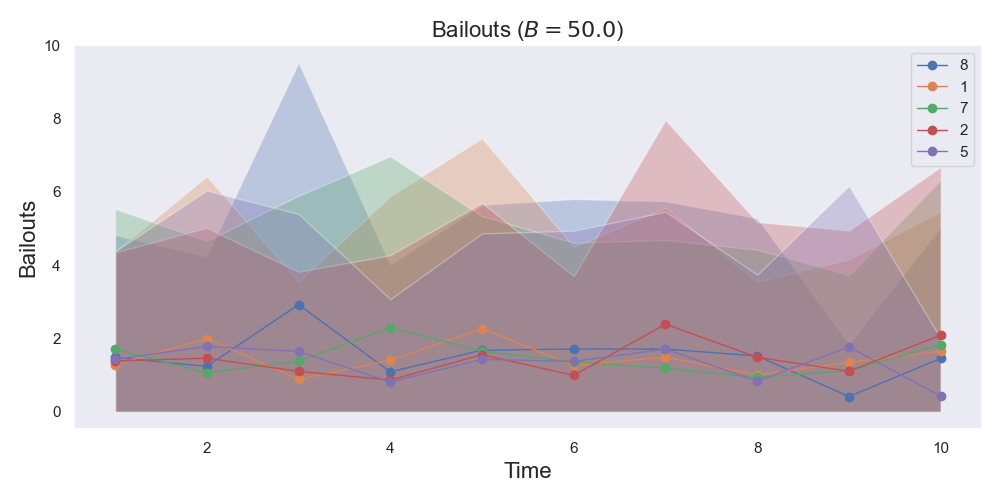}}
    \subfigure[Clearing Payments (Discrete interventions)]{\includegraphics[width=0.49\textwidth]{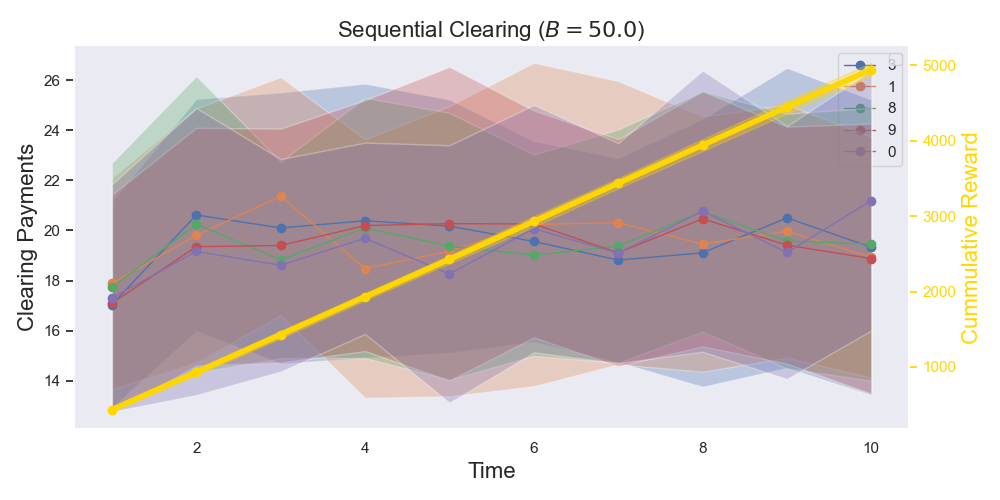}}
    \subfigure[Interventions (Discrete)]{\includegraphics[width=0.49\textwidth]{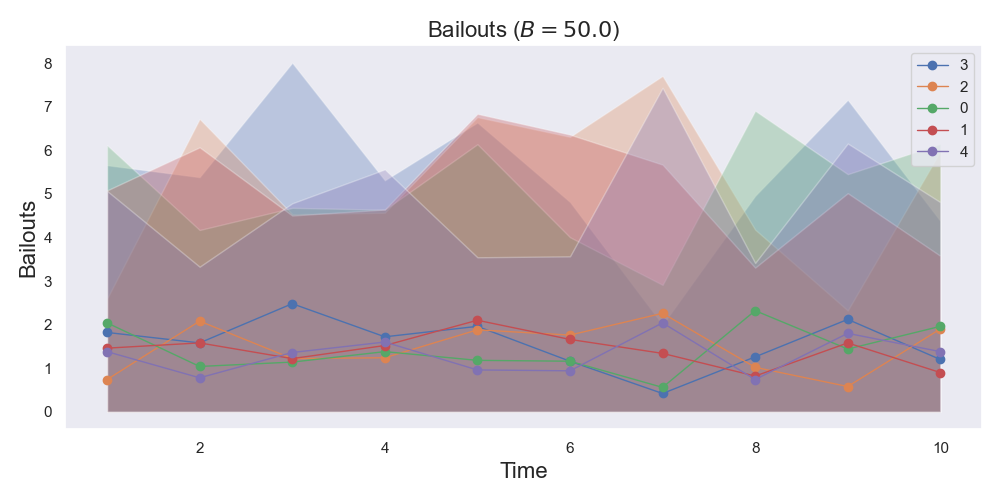}}
    \caption{\small \em Synthetic Core-Periphery Data Experiments averaged over 50 draws. Uncertainty corresponds to 1 standard deviation.}
    \label{fig:synthetic_data}
\end{figure}

We first run experiments with \emph{synthetic data} on networks that follow the stochastic blockmodel with core-periphery structure \cite{zhang2015identification}. Core-periphery networks have been widely observed in the modern financial (inter-bank) system \cite{chen2016financial, elliott2014financial, fricke2015core, craig2014interbank, in2020formation} whereas the core consists of a few \emph{large} banks (see e.g. \cite{papachristou2021sublinear, fricke2015core} for a characterization of the core size of core-periphery networks) and more \emph{small} banks. 

We generate a network of $n = 50$ nodes and $T = 10$ rounds whereas the structural graphs $G_t$ are drawn i.i.d. from a stochastic blockmodel with core-periphery structure with 2 blocks of size $n_{\mathrm{core}} = 10$ and $n_{\mathrm{periphery}} = 40$ and probabilities $\begin{pmatrix} p_{CC} = 0.6  & p_{CP} = 0.35 \\ p_{PC} = 0.35 & p_{PP} = 0.1 \end{pmatrix}$. The internal liabilities and the external liabilities are drawn i.i.d. from $\mathrm{Exp}(1)$, where the internal liability between $(i, j)$ at round $t$ is realized conditioned on the edge $(i, j)$ existing on $G_t$. and the asset vector is set to 0 for every time. We use a intervention budget of $B = 50$. \cref{fig:synthetic_data} shows the results of the clearing procedure with fractional interventions averaged over 50 draws of the financial environment as well as discrete interventions with $L = B \cdot \one$.  

\subsection{Ridesharing: For-Hire-Vehicles in New York City} \label{sec:experiment_ridesharing}

We use trip record data that are publicly available by the New York City (NYC) Taxi and Limousine Commission (TLC) and can e found at \url{https://www1.nyc.gov/site/tlc/about/tlc-trip-record-data.page}. The TLC data are split in time periods and each data-set contains the following fields

\begin{compactitem}
    \item \emph{Timestamp.} A timestamp indicating the time of the start of the ride.
    \item \emph{Source and destination Location ID.} The source location which belongs to a borough. We remove the rows where one of these quantities (or both) is missing. Location IDs correspond to \emph{zones} (e.g. Washingthon Heights South, East Harlem South etc.)
\end{compactitem}

We construct an instance of the dynamic clearing problem as follows 

\begin{compactenum}
    \item We define the network as rides between locations at the same borough (this can be extended to include rides from different boroughs; but here we focus on one borough for clarity of exposition). 
    \item The data is split into non-overlapping frames that correspond to some duration. For exposition clarity we have used 1-day intervals. Again, here we can use smaller intervals (e.g. 5min) to represent demand for rides realistically.
    \item We define $\ell_{ij}(t)$ as the total number of rides from location $i$ to location $j$ at timestamp $t$. 
    \item We define $b_i(t)$ to be the total number of external (outbound) rides requested from location $i$ to outside of the borough.
    \item We define $c_i(t)$ to be the total number of internal (inbound) rides requested for location $i$ from outside the borough. 
\end{compactenum}

\cref{fig:tlc_statistics} displays the statistics of the FHV TLC data-set for the Manhattan borough to other boroughs. We display both aggregate statistics across all zones (\cref{fig:tlc_statistics_aggregate}) and individual statistics for the 5 ``busiest`` (in terms of total rides) zones (\cref{fig:tlc_statistics_internal_outbound,fig:tlc_statistics_external_inbound,fig:tlc_statistics_external_outbound}). 

\medskip

\noindent \emph{Fractional interventions.} We run the simulations with budget $B \in \{ 0, 100, 500 \}$ and report the results in \cref{fig:tlc_days} where the intervention amounts for each node can range from 0 to $B$. More specifically, we report the clearing vectors for the 5 busiest zones (in terms of total clearing amount over the course of a month) as well as the cumulative reward.    

\medskip

\noindent \emph{Discrete interventions.} We run an experiment with budget $B = 100$ and $L = 10 \cdot \one$. \cref{fig:tlc_discrete} displays the clearing payments sequence and the indicator variables for the interventions for the 5 busiest zones.

\subsection{Financial Web Data: Venmo Data}

We use publicly available data (\url{https://github.com/sa7mon/venmo-data}) from public Venmo transactions. The dataset is split into three distinct time periods: (i) July 2018 to September 2018 (3.8M transactions) October 2018 (3.2M transactions) January 2019 to February 2019 (167K transactions). For our experiments we used the first period (July 2018 to September 2018) as it was the period with the most transactions. The amounts of the transactions are not provided in the data so we generate random transactions and use the provided topology. We have ignored data points for which the sender or the receiver of the transaction were missing. We construct the dynamic contagion instances as follows:

\begin{figure}
    \centering
    \subfigure[Clearing Payments and Cumulative Reward]{\includegraphics[width=0.47\textwidth]{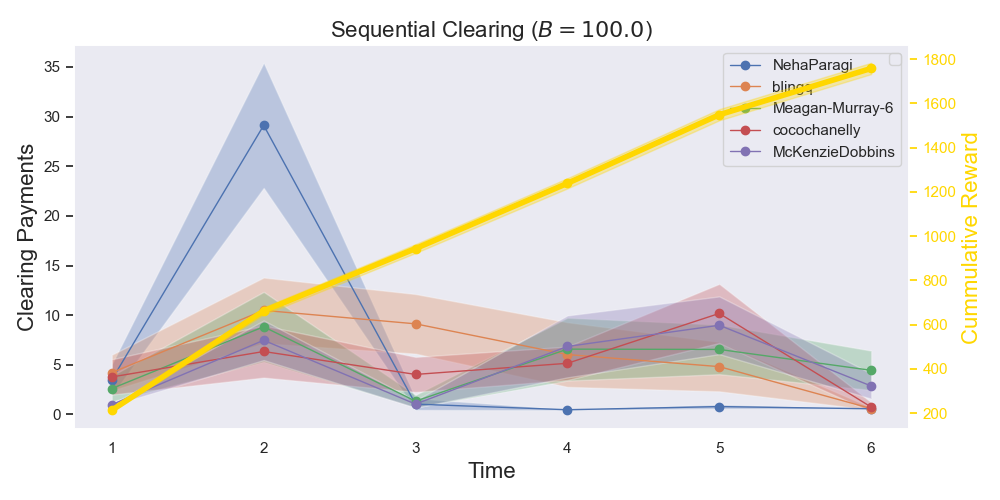}}
    \subfigure[Interventions]{\includegraphics[width=0.47\textwidth]{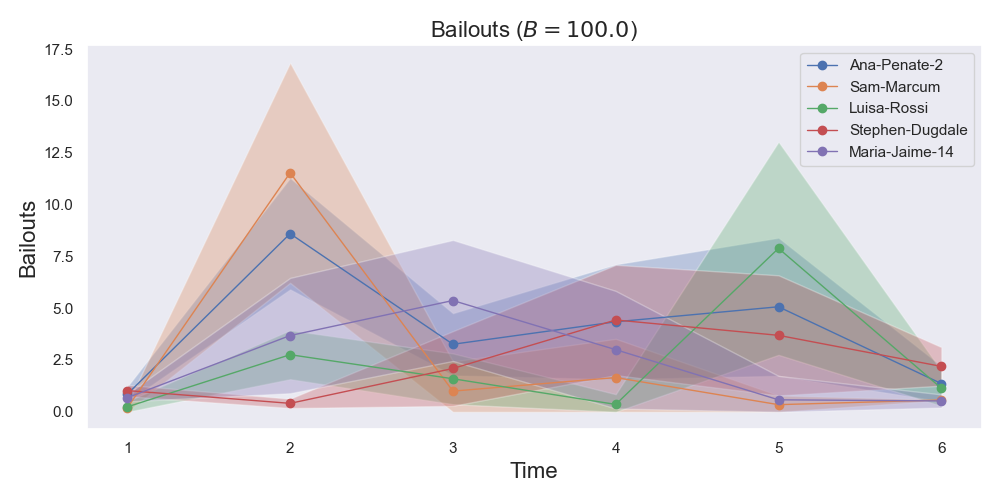}}
    \caption{\small \em Fractional interventions; Venmo data-set (July 2018-September 2018); $B = 100$ and $L = 100$}
    \label{fig:venmo_fractional}
\end{figure}

\begin{figure}
    \centering
    \subfigure[Clearing Payments and Cumulative Reward]{\includegraphics[width=0.47\textwidth]{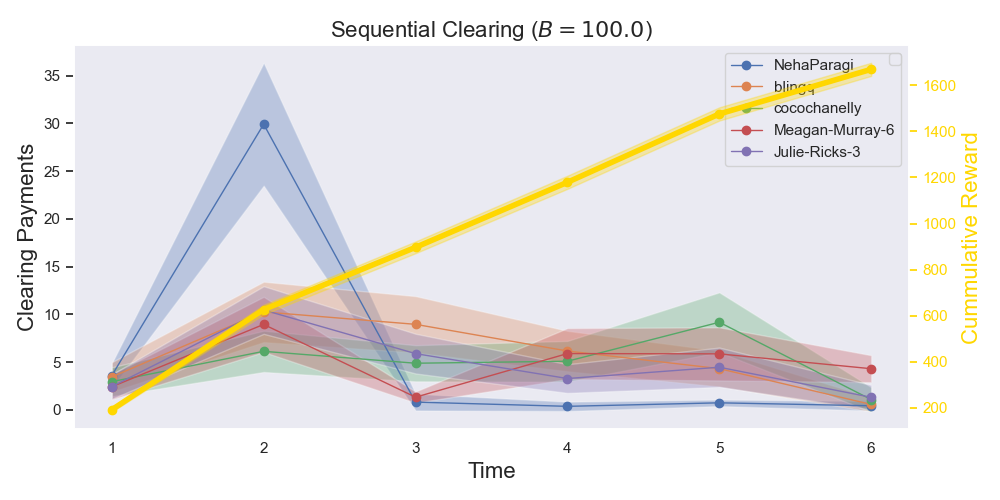}}
    \subfigure[Interventions]{\includegraphics[width=0.47\textwidth]{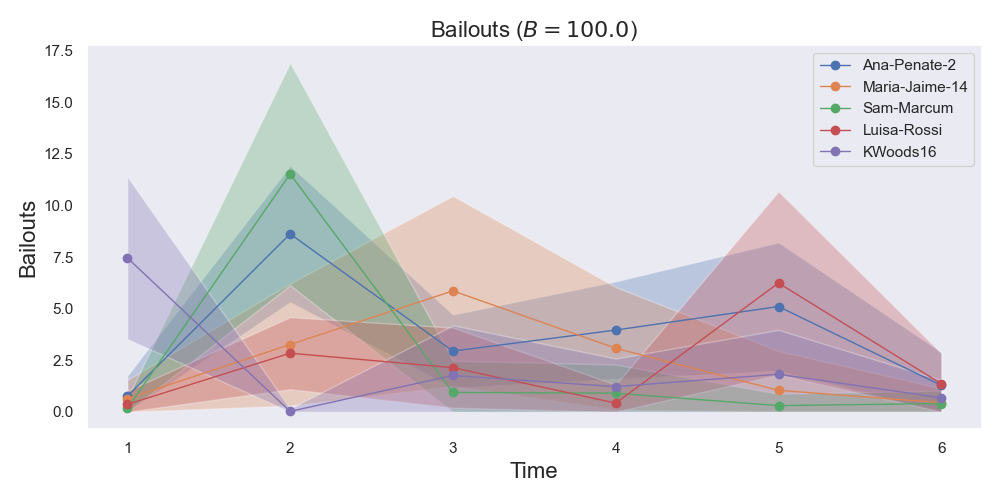}}
    \caption{\small \em Discrete interventions; Venmo Data-set (July 2018-September 2018); $B = 100$ and $L = 100$}
    \label{fig:venmo_discrete}
\end{figure}
\begin{figure}
    \centering
    \subfigure[Clearing Payments and Cumulative Reward]{\includegraphics[width=0.47\textwidth]{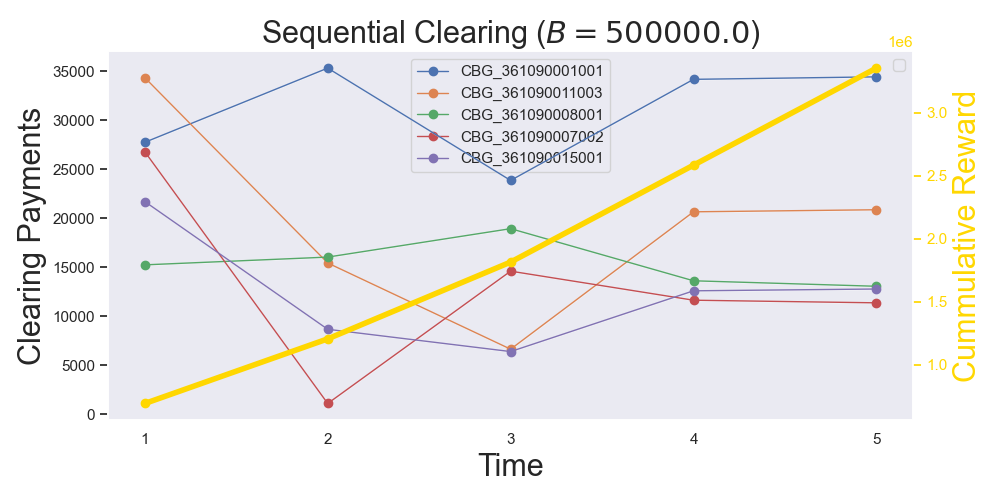}}
    \subfigure[Interventions]{\includegraphics[width=0.47\textwidth]{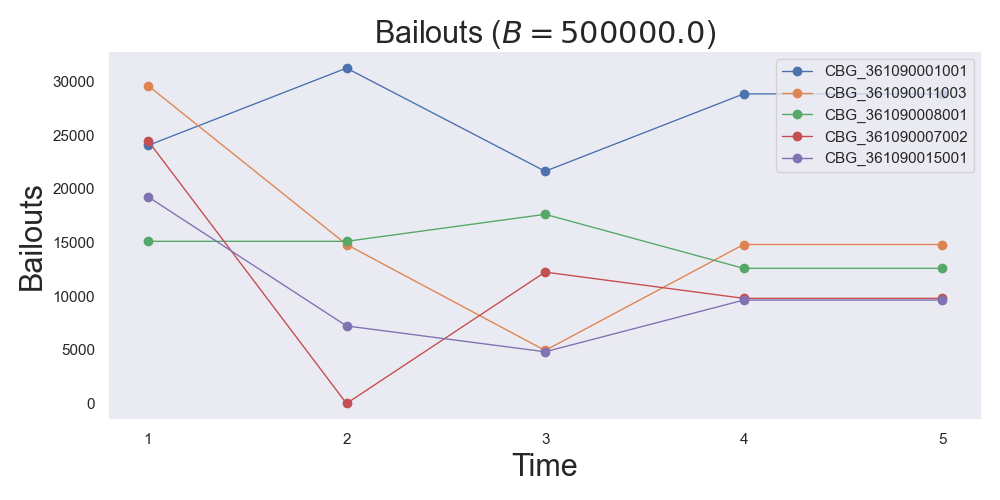}}
    \caption{\small \em Fractional interventions; SafeGraph data-set (December 2020-April 2021); $B = 500 \mathrm K$, custom $L$.}
    \label{fig:safegraph_fractional}
\end{figure}

\begin{compactitem}
    \item The timestamps are grouped in a weekly basis according to which year and week of the year they correspond to. 
    \item For the whole dataset (corresponding to the July-September 2018 period), we create 2 sets: $V_1$ and $V_2$. $V_1$ corresponds to the top-100 nodes in terms of the number of incoming transactions, and $V_2$ corresponds to the top-100 nodes in terms of the number of outgoing transactions. As the vertex set we use $V = V_1 \cup V_2$. 
    \item We count the transactions between nodes of $V$, as well as the transactions from $V$ to the outside system, and from $V$ to the inside system for each round (week of the year 2018). For nodes with zero outgoing transactions we add one transaction. 
    \item We create random liabilities as follows 
    \begin{align*}
        \ell_{ij}(t) & = \one \left \{ \text {\# of transactions $i \to j$ at round $t$ is $> 0$} \right \} \\ & \times \mathrm{Gamma} \left ( \text {\# of transactions $i \to j$ at round $t$}, 1  \right ) & \forall i, j \in [n], t \in [T], \\
        b_i(t) & = \max \left \{ 1, \mathrm{Gamma} \left ( \text{\# of transactions from $i$ to outside}, 1 \right ) \right \} & \forall i \in [n], t \in [T], \\
        c_i(t) & = \one \left \{ \text {\# of transactions from outside to $i$ is $> 0$} \right \} \\ & \times \mathrm{Gamma} \left ( \text{\# of transactions from outside to $i$}, 1 \right ) & \forall i \in [n], t \in [T].
    \end{align*}
    
    Note that for $b_i(t)$ we assert a \emph{positive} value in order for \cref{assumption:uniqueness} to hold. 
\end{compactitem}

In \cref{fig:venmo_fractional} we use a budget of $B = 100$ and $L = 100 \cdot \one$ for $N = 50$ instances of randomly initialized networks and fractional interventions. In \cref{fig:venmo_discrete}, we use the same parameters but for discrete interventions. We plot the top-5 Venmo account usernames (the usernames were included in the \textbf{public} data release and are available online) with the highest clearing payments, as well as the cumulative reward amassed by \cref{alg:fractional_interventions} and \cref{alg:randomized_rounding}.

\subsection{Financial Networks Through Cellphone Mobility Data} 

We use mobility data from the SafeGraph platform spanning the period of December 2020 to April 2021. The static version of the data (i.e. with one round) has been introduced in \cite{papachristou2021allocating}. Our paper extends the experiment to the multi-period setting, resulting in five financial networks each one of which for each month of the period of December 2020 to April 2021. SafeGraph offers insights on mobility patterns of people between Census Block Groups (CBGs) and Points of Interest (POIs), such as grocery stores, fitness centers, and religious establishments. We create a sequence of bipartite liability networks with monthly granularity between CBGs and POIs using the mobility data as a proxy to create liabilities. Roughly, given an initial pair of coordinates (latitude and longitude) we identify $k$-nearest neighboring CBGs and, subsequently, the POIs that these CBGs interact with. Each POI contains data from the visits of unique mobile devices from the corresponding CBGs, where we assume that each distinct device represents a unique person. SafeGraph also logs data about the dwelling time of devices, namely for how long each device stays in a POI, which we use to classify the visitors to two categories: customers and employees. For the employees, we add a financial liability that corresponds to the average monthly wage of such an employee for the corresponding POI as it is determined by its NAICS code. For the customers, we add edges from the corresponding CBGs to the POIs with value being the average consumption value for the specific POI category (given by the POI's NAICS code) as it is defined by the U.S. Economic Census. For each CBG node, we estimate the average size of households per CBG, the average income level and the percentage of people that belong to a minority group, which we use in order to estimate the assets and liabilities (internal and external) of CBGs. For the interventions of CBGs we use the US Cares Act rules to determine the interventions. Regarding the POI interventions we used data from loans provided in April 2020 as part of the SBA Paycheck Protection Program (PPP) and impute the interventions of the nodes that do not have intervention information. The processed data was \emph{calibrated} similarly to \cite{papachristou2021allocating} to reflect the specific interaction between POIs and CBGs across the different time spans. Regarding the intervention budget we use a budget of $B = 500 \mathrm K$. The results for the most important nodes (either POIs or CBGs) for the fractional (resp. discrete) allocations are depicted in \cref{fig:safegraph_fractional} (resp. \cref{fig:safegraph_discrete}). A more detailed explanation of the dataset creation process is presented in \cref{sec:data_addendum}. 

\subsection{Fairness} \label{sec:experiments_fairness}

To induce fair allocations (in the fractional case) we run \cref{alg:fractional_interventions} with the constraints of \cref{sec:fairness}. In \cref{fig:fairness} we perform fairness-constrained optimization for the case of synthetic data for $B = 50$ and $L = 50 \cdot \one$ (as in \cref{sec:synthetic_data}). We plot $\mathrm{SGC}(t)$ for the case of unconstrained fairness ($g(t) = 1$) and the case of constrained fairness ($g(t) = 0.5$). We plot the interventions after the fairness constrained optimization (for the non-fair allocations see \cref{fig:synthetic_data}). For all datasets we run fairness-constrained optimization experiments using the SGC and the (Standard) GC as fairness measures and present the results in \cref{fig:interventions_vs_payments} and \cref{fig:betas_vs_payments}. We report the PoF for the corresponding experiments in \cref{tab:pof}. For the datasets that involve randomness, we take the average over 50 independent runs. 

\begin{table}[h]
    \footnotesize
    \centering
    \begin{tabular}{lllll}
    \toprule
        Fairness Constraint & Synthetic ($B = 50$) & TLC ($B = 100$) & Venmo ($B = 50$) & Safegraph ($B = 500 \mathrm K$) \\
    \midrule
        Spatial GC ($g(t) = 0.5$) & 1.001 & 1.007 & 1.019 & 1.037 \\
        Standard GC ($g(t) = 0.5$) & 1.011 & 1.009 & 1.017 & 1.102 \\
    \bottomrule
    \end{tabular}
    \caption{\small \em Price of Fairness. The payments and allocations can be found at \cref{fig:interventions_vs_payments}.}
    \label{tab:pof}
\end{table}

\begin{figure}
    \centering
    \subfigure[Aggregate statistics]{\includegraphics[width=0.24\textwidth]{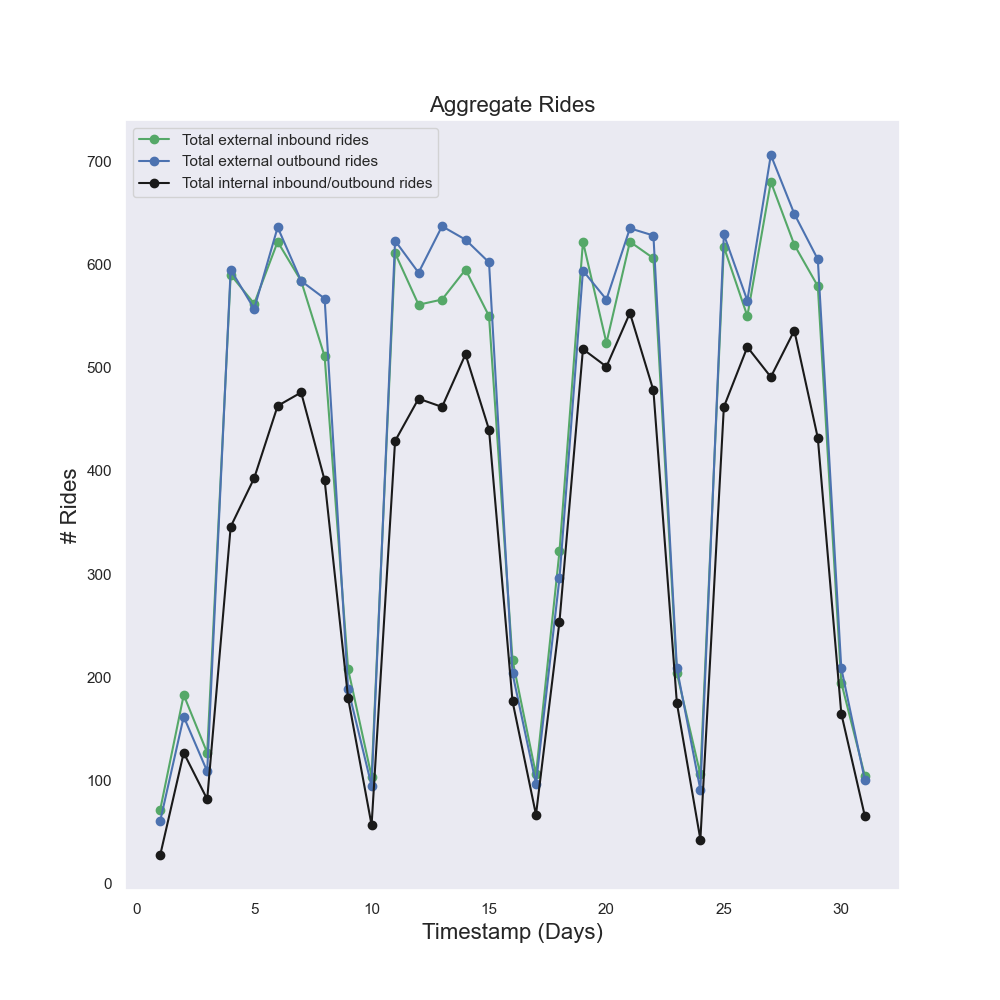}\label{fig:tlc_statistics_aggregate}}
    \subfigure[Top-5  Zones]{\includegraphics[width=0.24\textwidth]{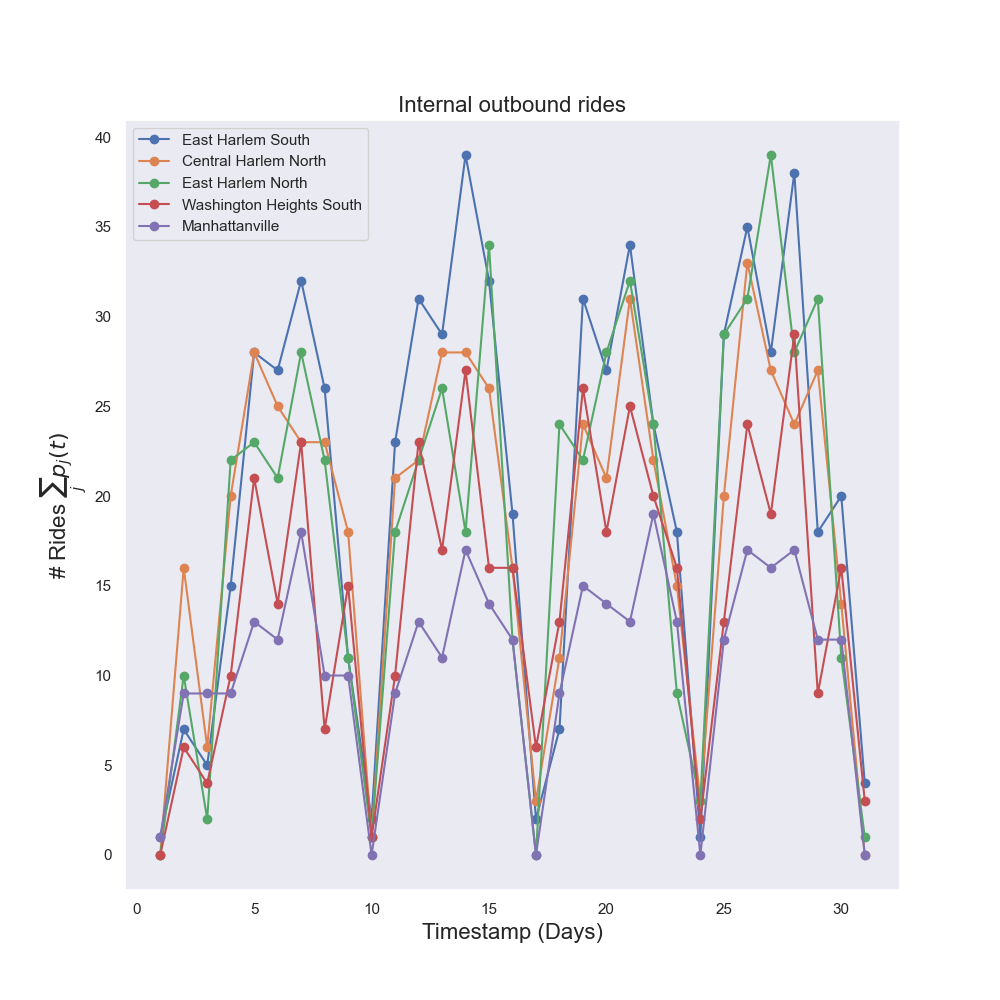}\label{fig:tlc_statistics_internal_outbound}}
    \subfigure[Top-5 Zones]{\includegraphics[width=0.24\textwidth]{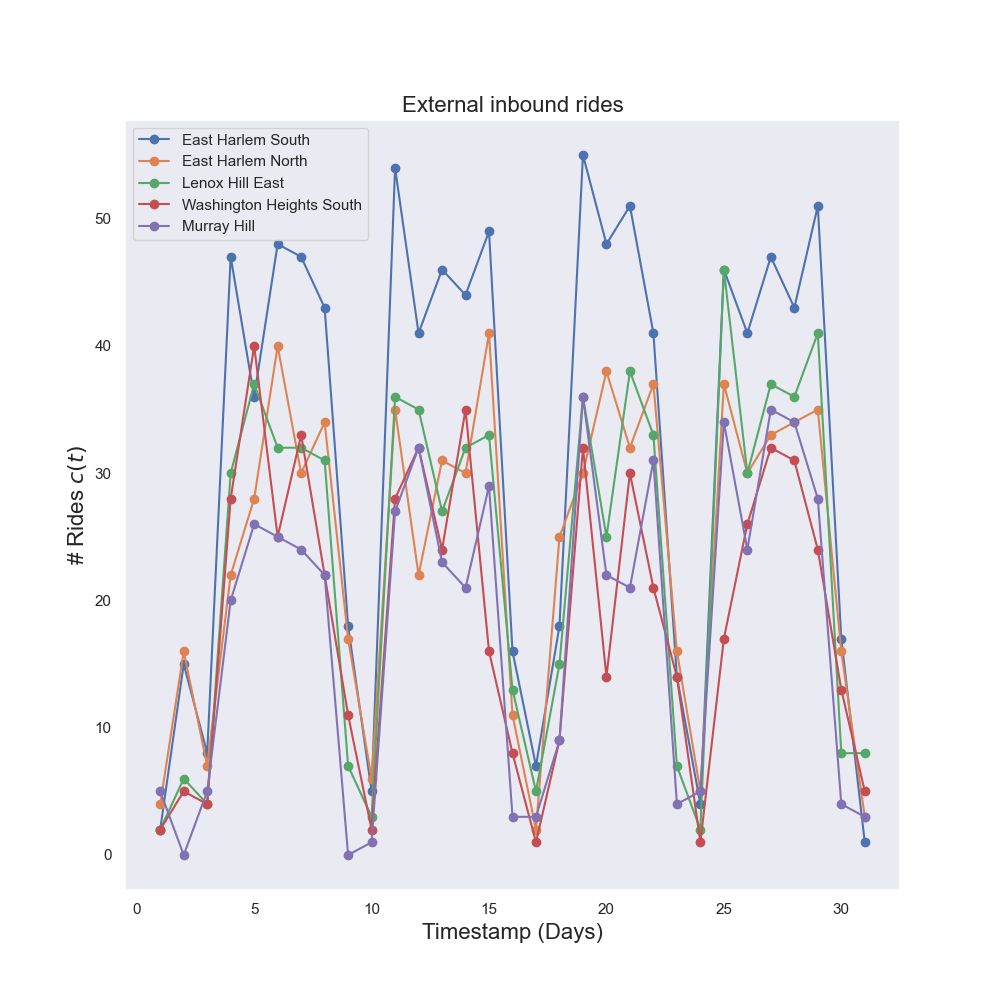}\label{fig:tlc_statistics_external_inbound}} 
    \subfigure[Top-5 Zones]{\includegraphics[width=0.24\textwidth]{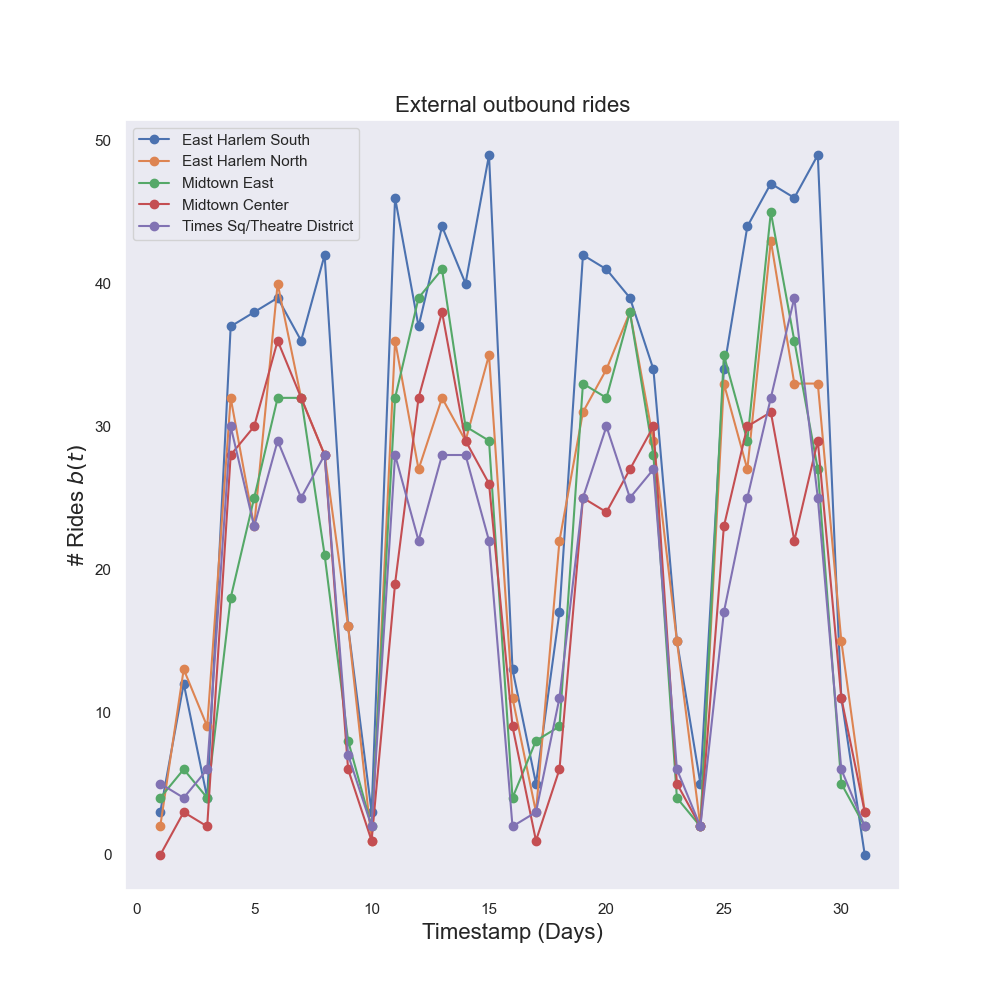}\label{fig:tlc_statistics_external_outbound}}
    \caption{\small \em NYC TLC data statistics for FHV during January 2021 for Manhattan.}
    \label{fig:tlc_statistics}
\end{figure}

\begin{figure}
    \centering
    \subfigure[Clearing Payments and Cumulative Reward]{\includegraphics[width=0.49\textwidth]{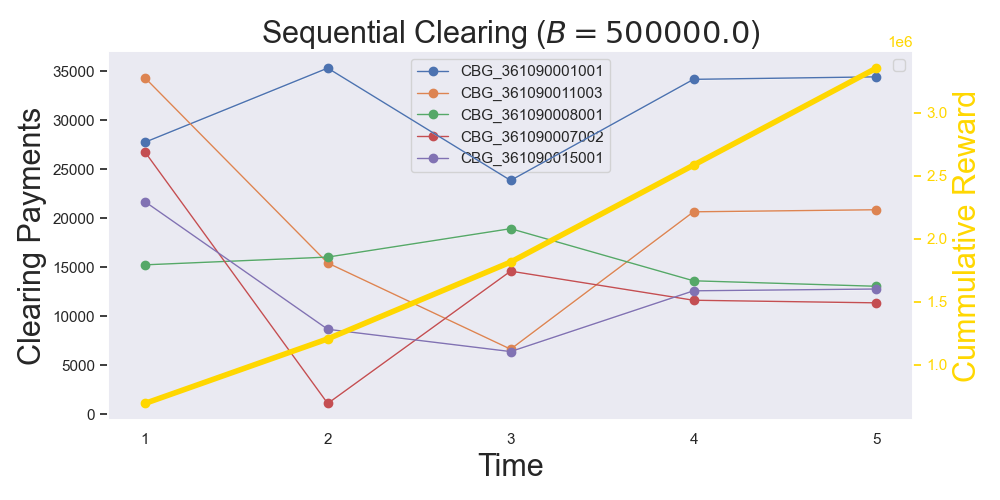}}
    \subfigure[Interventions]{\includegraphics[width=0.49\textwidth]{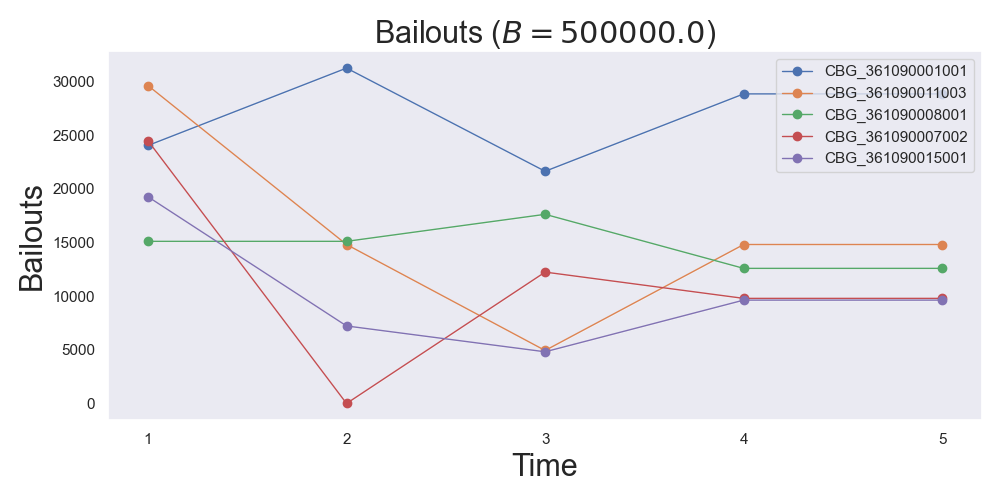}}
    \caption{\small \em Discrete intervention Scenario for SafeGraph Data (December 2020 to April 2021) for $B = 500 \mathrm K$ and custom $L$.}
    \label{fig:safegraph_discrete}
\end{figure}

\begin{figure}
    \centering
    \subfigure[SGC ($g(t) = 1$)]{\includegraphics[width=0.49\textwidth]{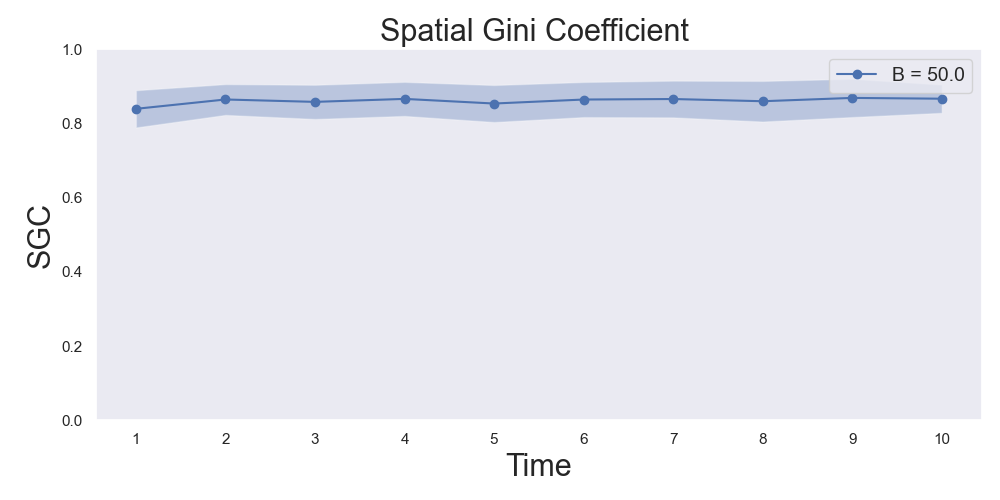}}
    \subfigure[SGC ($g(t) = 0.5$)]{\includegraphics[width=0.49\textwidth]{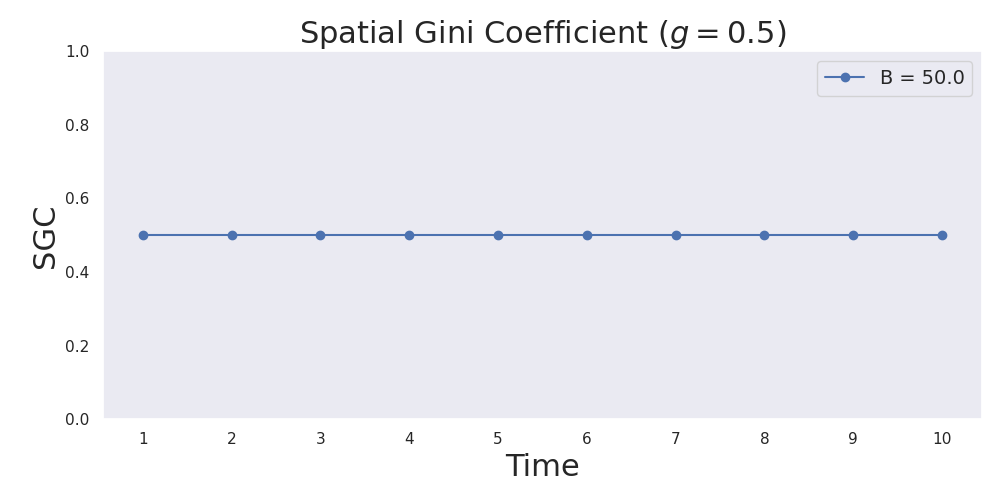}}   
    \subfigure[Interventions ($g(t) = 0.5$)]{\includegraphics[width=0.49\textwidth]{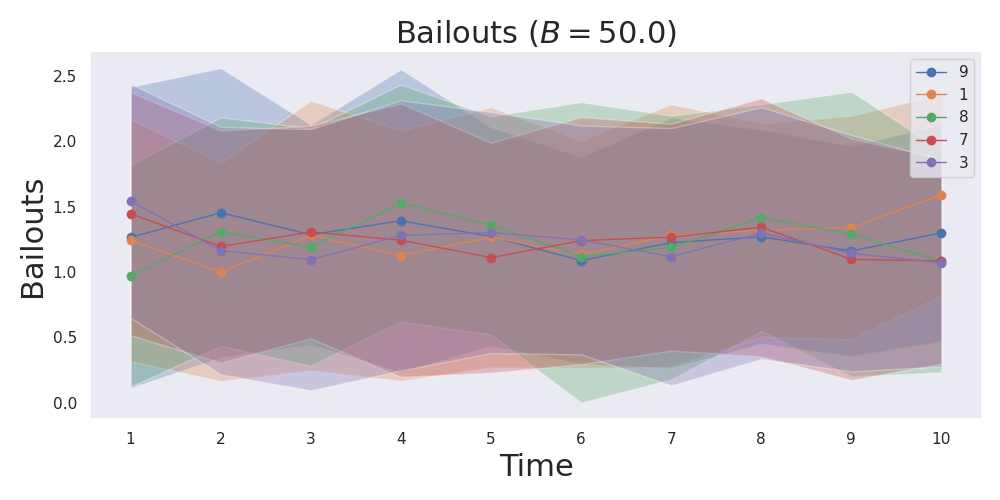}}    
    \subfigure[Clearing Payments ($g(t) = 0.5$)]{\includegraphics[width=0.49\textwidth]{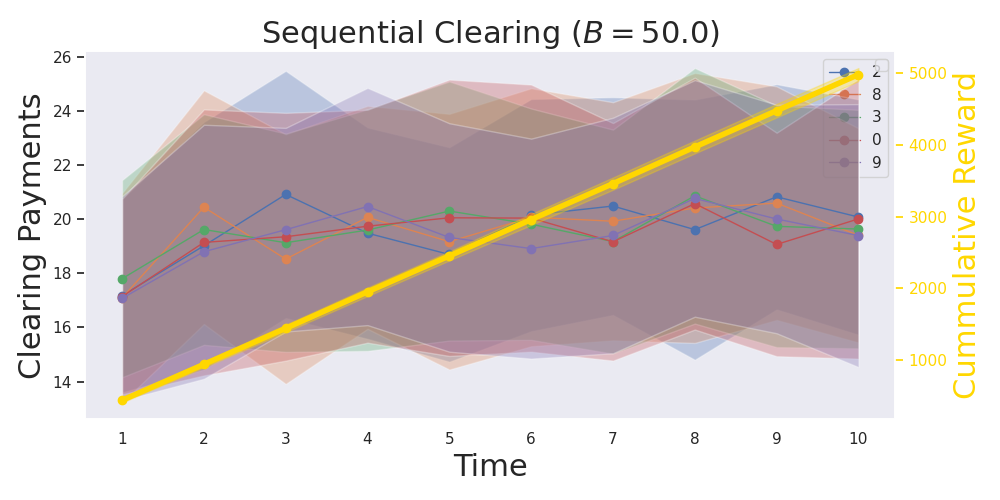}}    
    \caption{\small \em Fairness in Synthetic Core-periphery Data Experiments (see \cref{sec:synthetic_data}).} 
    \label{fig:fairness}
\end{figure}

\section{Insights} \label{sec:insights}

We study the fractional intervention allocations\footnote{Because of the rounding procedure of \cref{sec:discrete_allocations}, discrete allocations are concentrated around their expectation and therefore, the scatter plots look similar to the fractional case.}  as a function of network characteristics. We consider the relationship between interventions, total payments, and between interventions and (mean) financial connectivity for the three datasets of interest. We plot the relationship between the average total interventions and the network characteristics subject to fairness ($g(t) = 0.5$) and no-fairness ($g(t) = 1$) constraints. For the fairness constraints, we consider both the Spatial Gini Coefficient (SGC) and the (Standard) Gini Coefficient (GC). We fit an ordinary least squares model (OLS) to the points and a robust linear model via iteratively reweighted least squares with Huber's T criterion.

\medskip

\noindent \emph{Interventions vs. Payments.} In \cref{fig:interventions_vs_payments},  we plot the average total interventions and the average total payments a node can make. We observe that for the synthetic core-periphery data, the core and the peripheral nodes are split in the plot yielding a correlation coefficient of $R^2 = 0.74$. For the TLC data ($R^2 = 0.36$), we observe that the two quantities generally follow a linear relation (except for a few outliers). The Venmo data shows a positive correlation ($R^2 = 0.61$) between the interventions and the payments. This is suggestive of the fact that nodes that are generally bailed out are central in the system, so bailing them out benefits the connections they have. For the SafeGraph data, we observe a high correlation between the payments and the allotted interventions ($R^2 = 0.887$). However, we observe a ``banding phenomenon'' where nodes that have varying payments in the network receive a small amount of total interventions across time. This phenomenon is due to the following reasons: (i) at each round, some CBGs do not participate in the contagion (i.e., POIs do not get visits from them), which amount to a small total intervention, and (ii) the imputed interventions (see \cref{sec:data_addendum}) are similar for relatively large groups of nodes. With respect to the fairness-adjusted metrics wrt. the SGC, we observe that imposing the SGC constraint reduces the correlation between the two variables since the SGC constraint balances interventions between adjacent nodes, which may result in nodes that are less important to the contagion process getting high interventions. In the SafeGraph data, imposing fairness constraints reinforces the ``banding phenomenon'' and creates zones of interventions. Also, when the (Standard) GC constraint is enforced, the correlation reduces even further since this constraint penalizes unequal distribution of the interventions across all nodes. 

\medskip

\noindent \emph{Price of Fairness.} Regarding the PoF, \cref{tab:pof} indicates that all fairness measures achieve a PoF close to 1 in all datasets. This suggests that, generally, our allocation algorithm can respect algorithmic fairness constraints with a minimal cost to the total welfare. 

\medskip

\noindent \emph{Financial Connectivity vs. Payments.} Similarly, in \cref{fig:betas_vs_payments}, we observe that in the no-fairness setting the interventions are positively correlated with the financial connectivities, i.e. nodes with more (in-total) liabilities towards the internal network get higher interventions. More specifically, on the synthetic core-periphery data, core nodes and periphery nodes are separated similarly to \cref{fig:interventions_vs_payments}, as well as in Venmo data there is a cluster of nodes that has both very little exposure to the outside system as well as small interventions. For the SafeGraph data, we observe the same banding phenomenon as in \cref{fig:interventions_vs_payments} which is structurally due to the same reason (i.e. CBGs which have limited participation  in most rounds). Furthermore, when fairness is introduced the correlation between the two variables changes, as in the case of the synthetic data, since we trade-off fairness with adjacent nodes to bailing out nodes with potentially high contribution in the network. Most, notably, in the synthetic core-periphery network the correlation between the variables becomes negative when fairness is introduced, from positive that was in the absence of fairness, indicating that the peripheral nodes (points on the left ``cluster'') get higher interventions than the core (points on the right ``cluster''). In the TLC dataset, the correlation remains positive however the outliers become more spread-out, again owing to the equitable spread of interventions across graph neighborhoods. The Venmo data exhibits similar behaviour as in the case of no-fairness. We observe the banding phenomenon (similarly to \cref{fig:interventions_vs_payments}) for the SafeGraph data after imposing fairness constraints. Also, when the (Standard) GC constraint is enforced, the correlation reduces further similarly to \cref{fig:interventions_vs_payments}.

\begin{figure}
    \centering
    \subfigure[Synthetic Core-periphery ($g(t) = 1$)]{\includegraphics[width=0.24\textwidth]{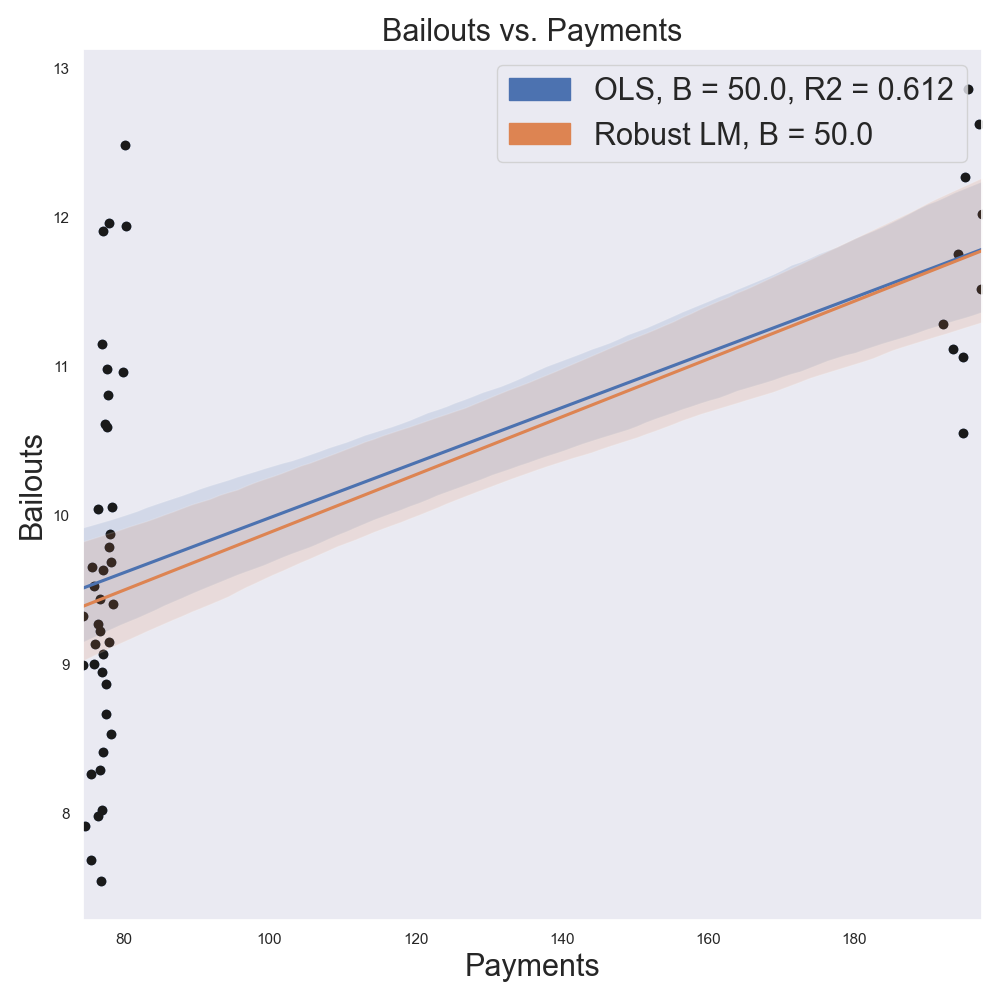}}
    \subfigure[TLC ($g(t) = 1$)]{\includegraphics[width=0.24\textwidth]{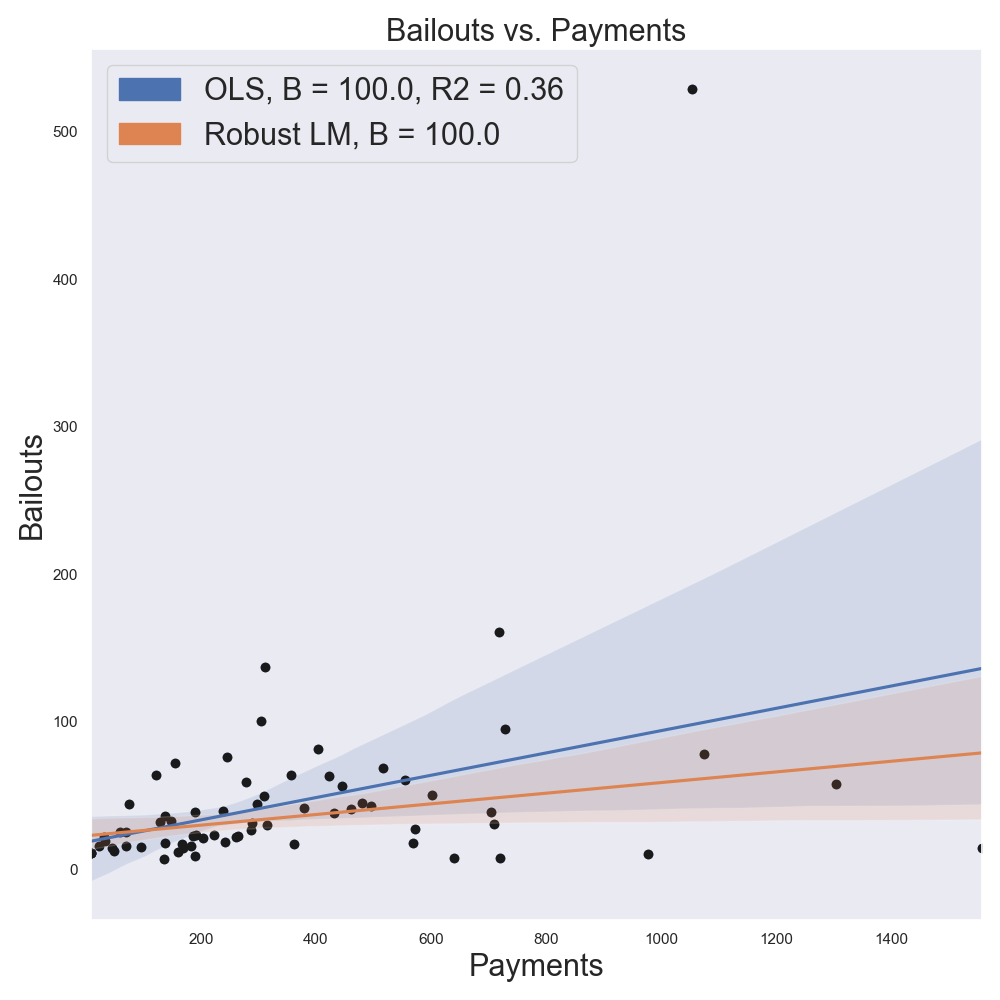}}
    \subfigure[Venmo ($g(t) = 1$)]{\includegraphics[width=0.24\textwidth]{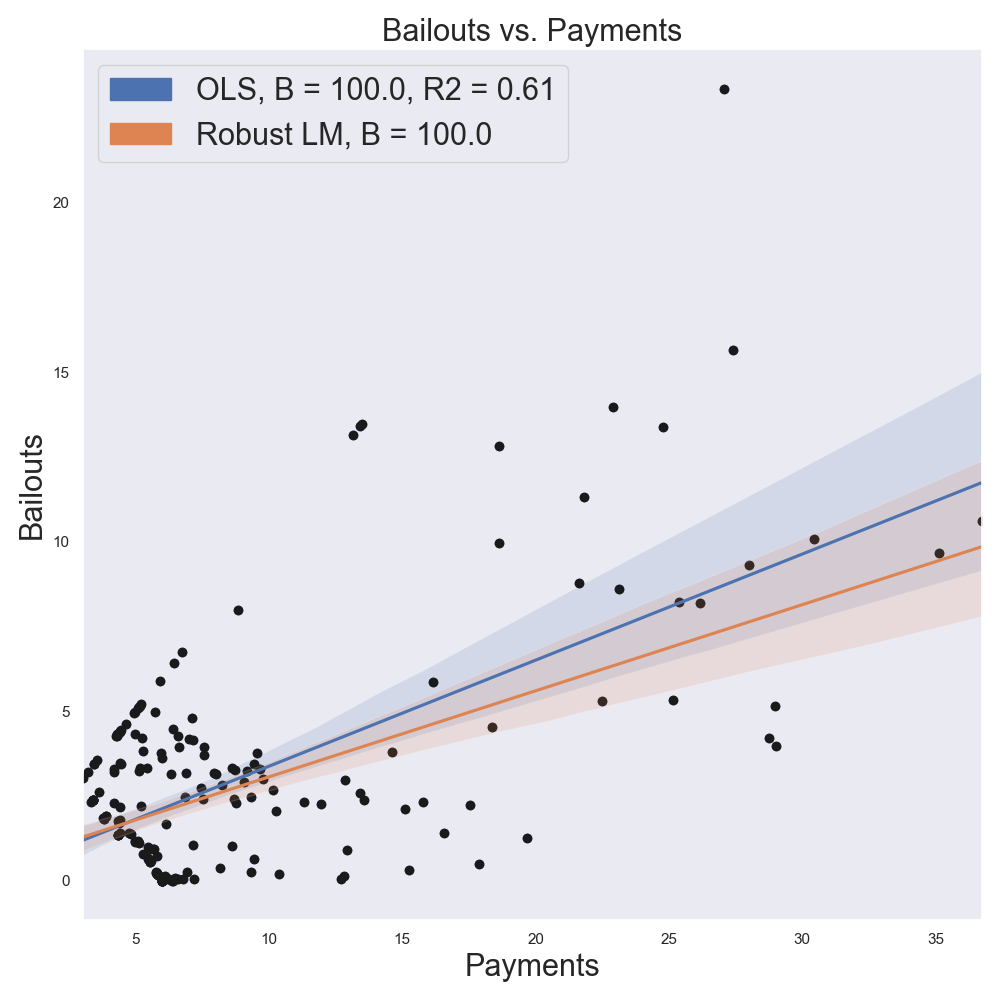}}
    \subfigure[SafeGraph ($g(t) = 1$)]{\includegraphics[width=0.24\textwidth]{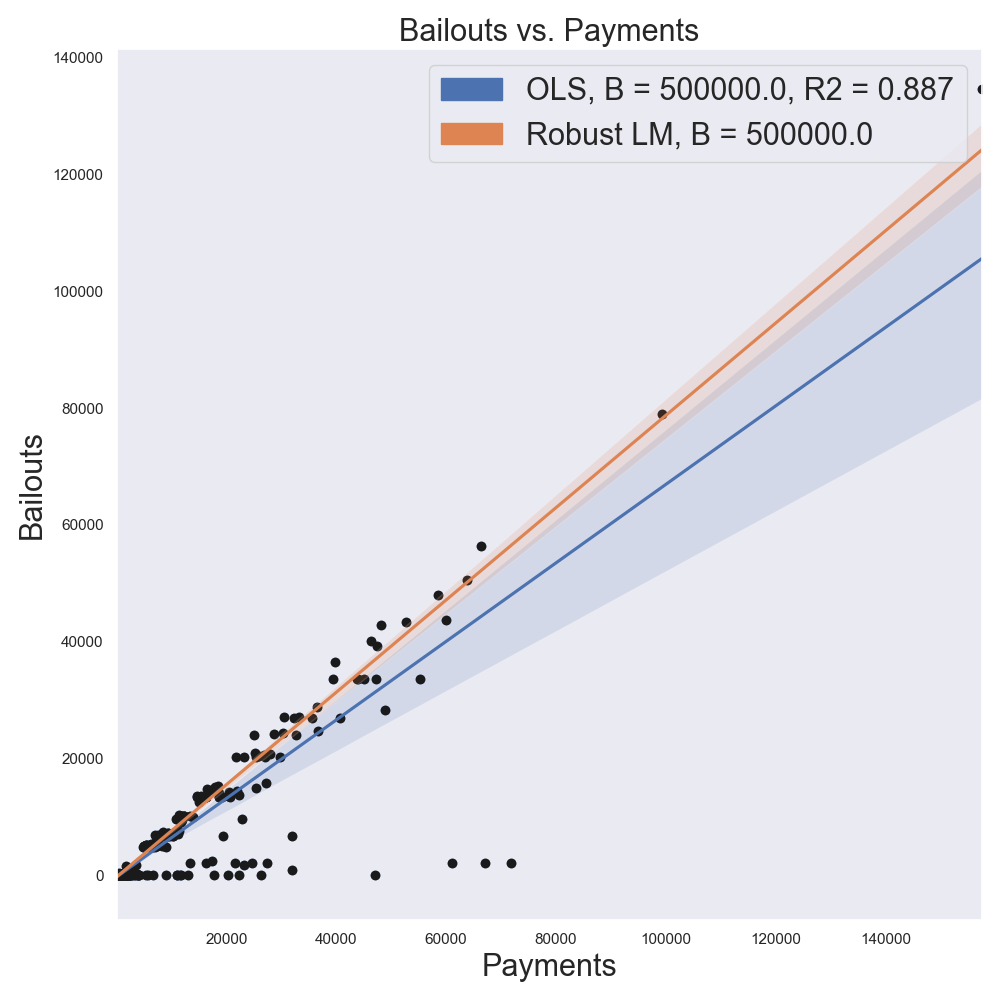}}
    \subfigure[Synthetic Core-periphery (SGC, $g(t) = 0.5$)]{\includegraphics[width=0.24\textwidth]{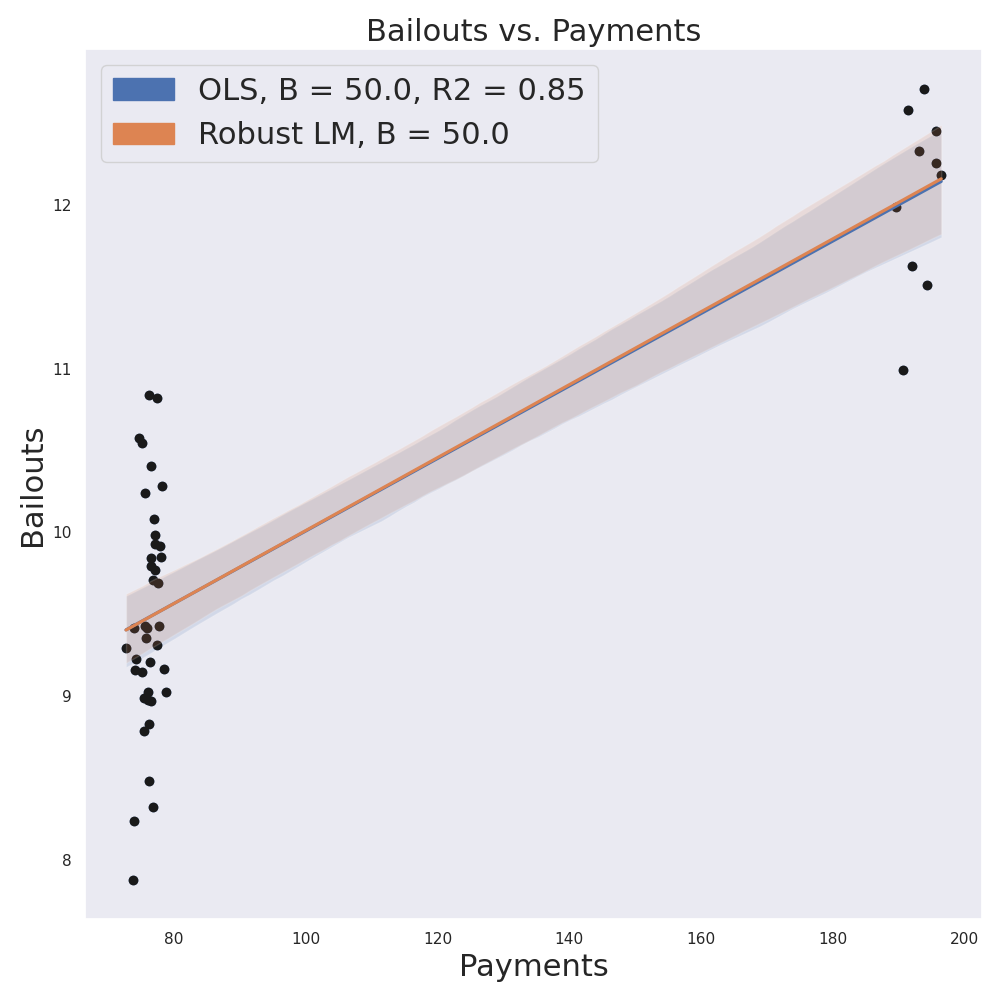}}
    \subfigure[TLC (SGC,  $g(t) = 0.5$)]{\includegraphics[width=0.24\textwidth]{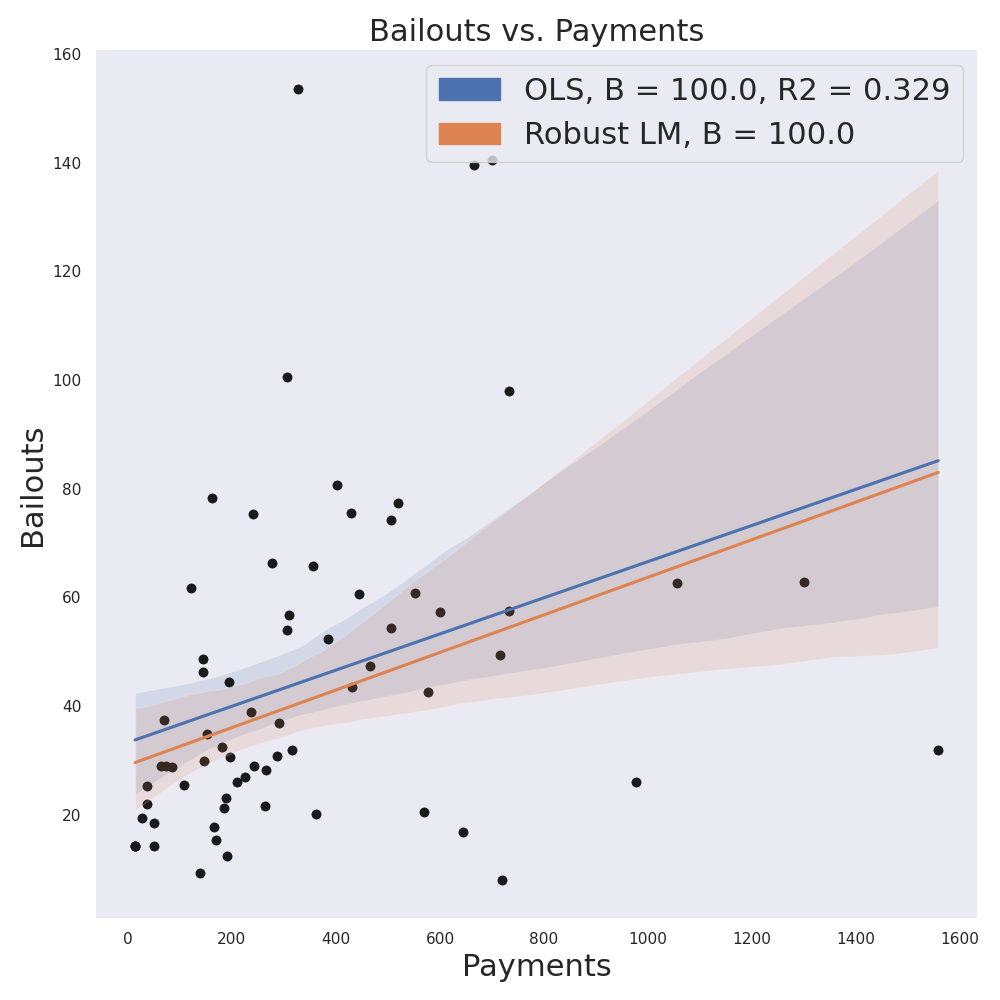}}
    \subfigure[Venmo (SGC,  $g(t) = 0.5$)]{\includegraphics[width=0.24\textwidth]{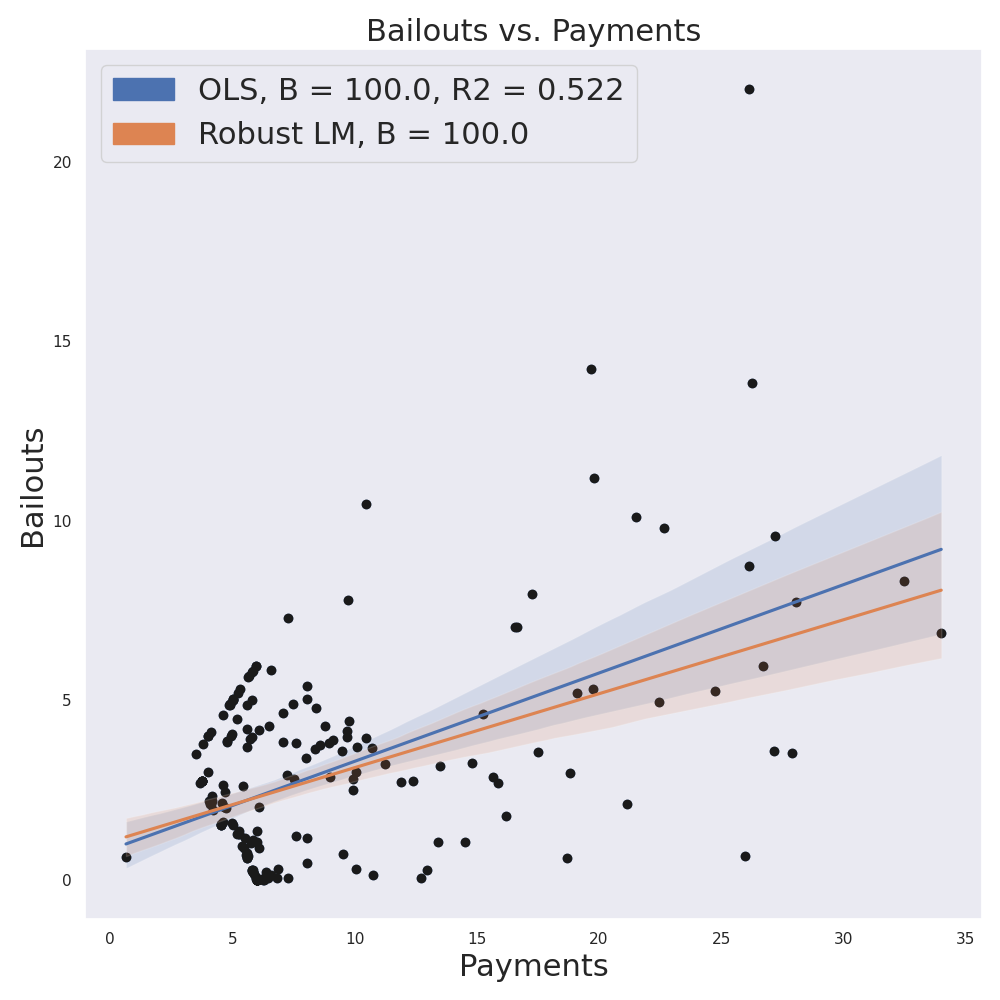}}
    \subfigure[SafeGraph (SGC, $g(t)=0.5$)]{\includegraphics[width=0.24\textwidth]{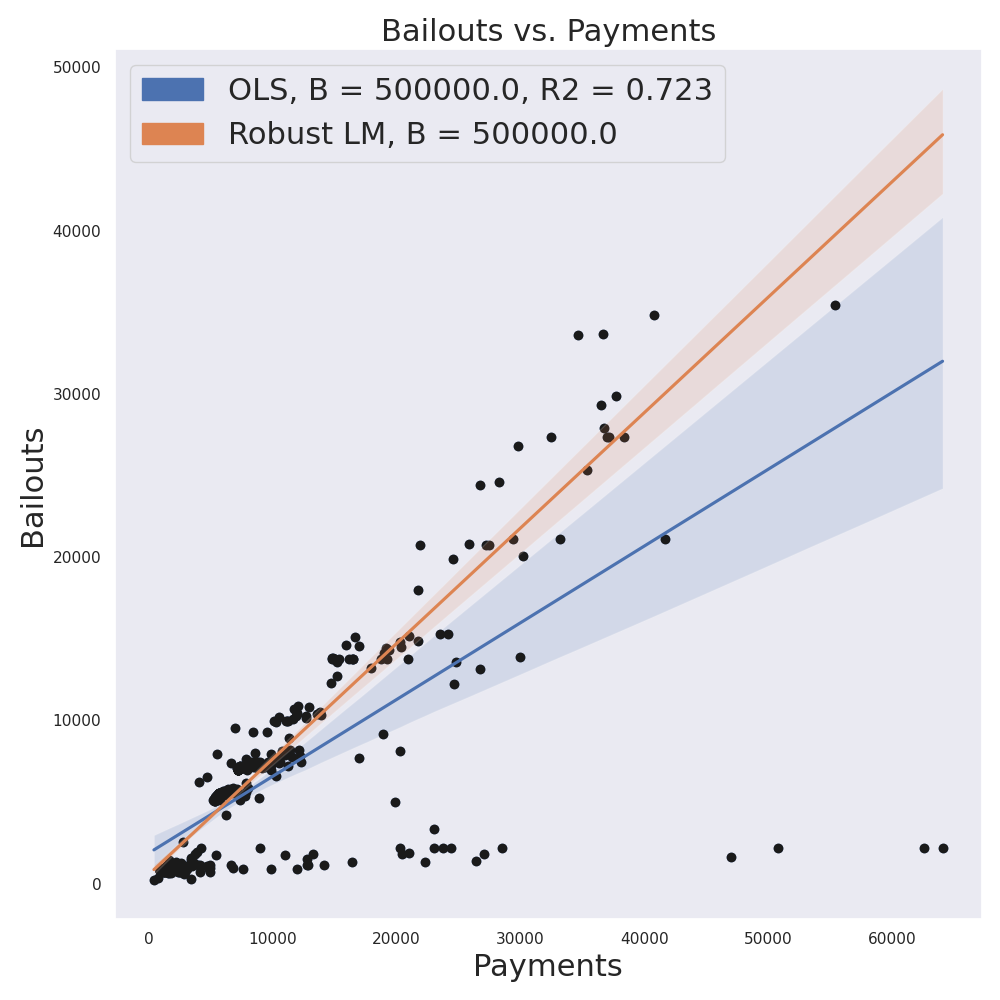}}
    \subfigure[Synthetic Core-periphery (GC,  $g(t) = 0.5$)]{\includegraphics[width=0.24\textwidth]{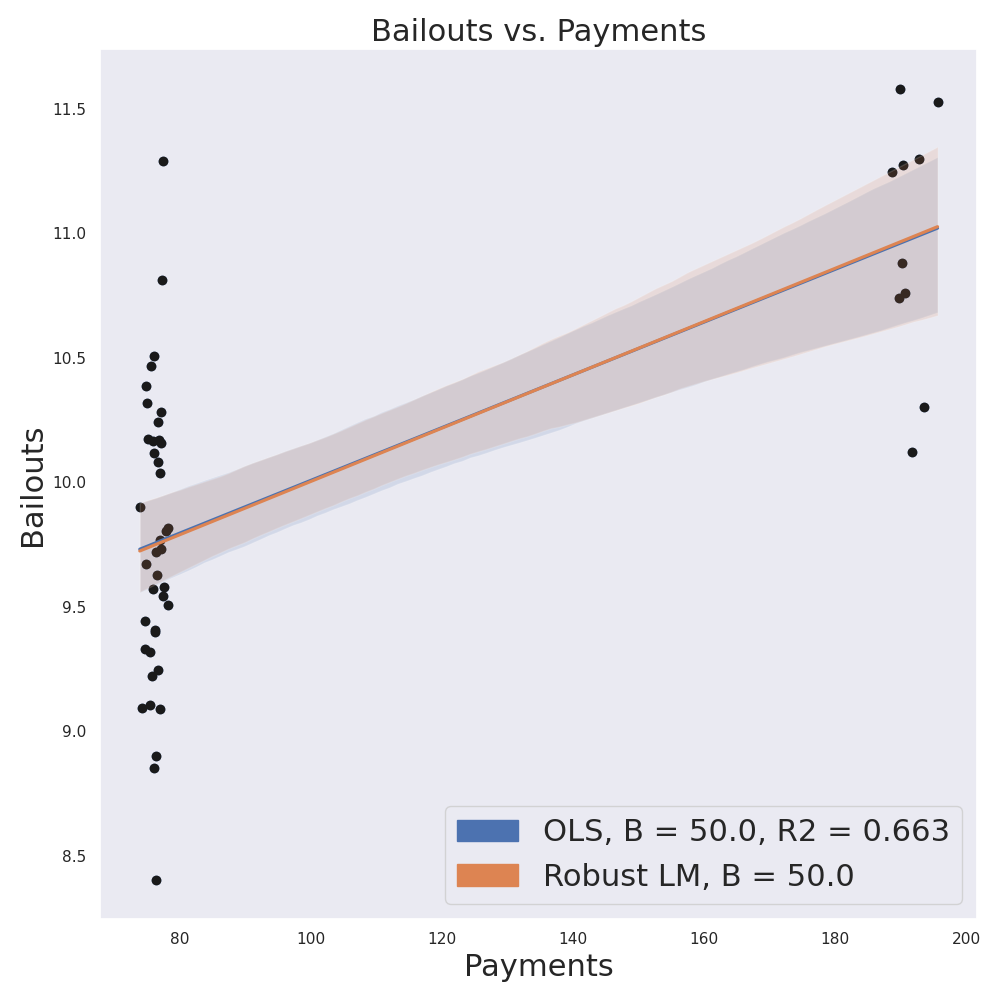}}
    \subfigure[TLC (GC, $g(t) = 0.5$)]{\includegraphics[width=0.24\textwidth]{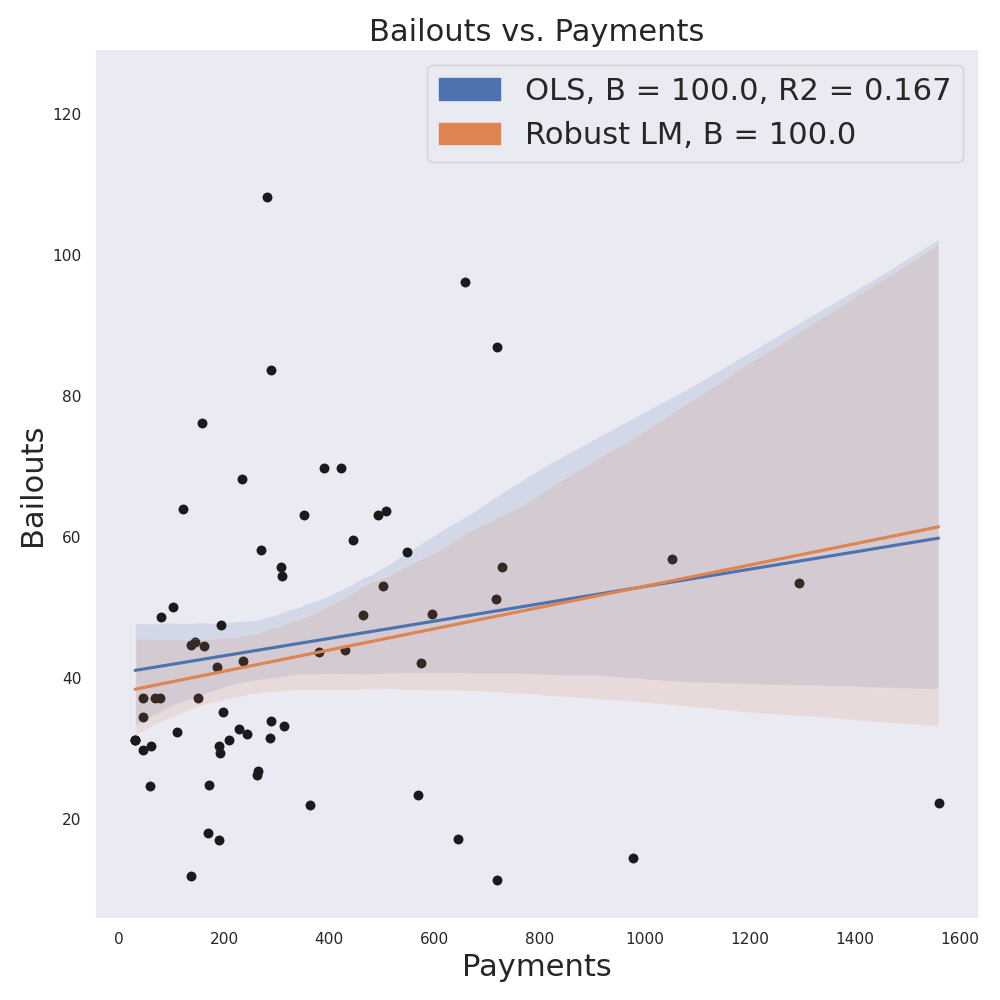}}
    \subfigure[Venmo (GC, $g(t) = 0.5$)]{\includegraphics[width=0.24\textwidth]{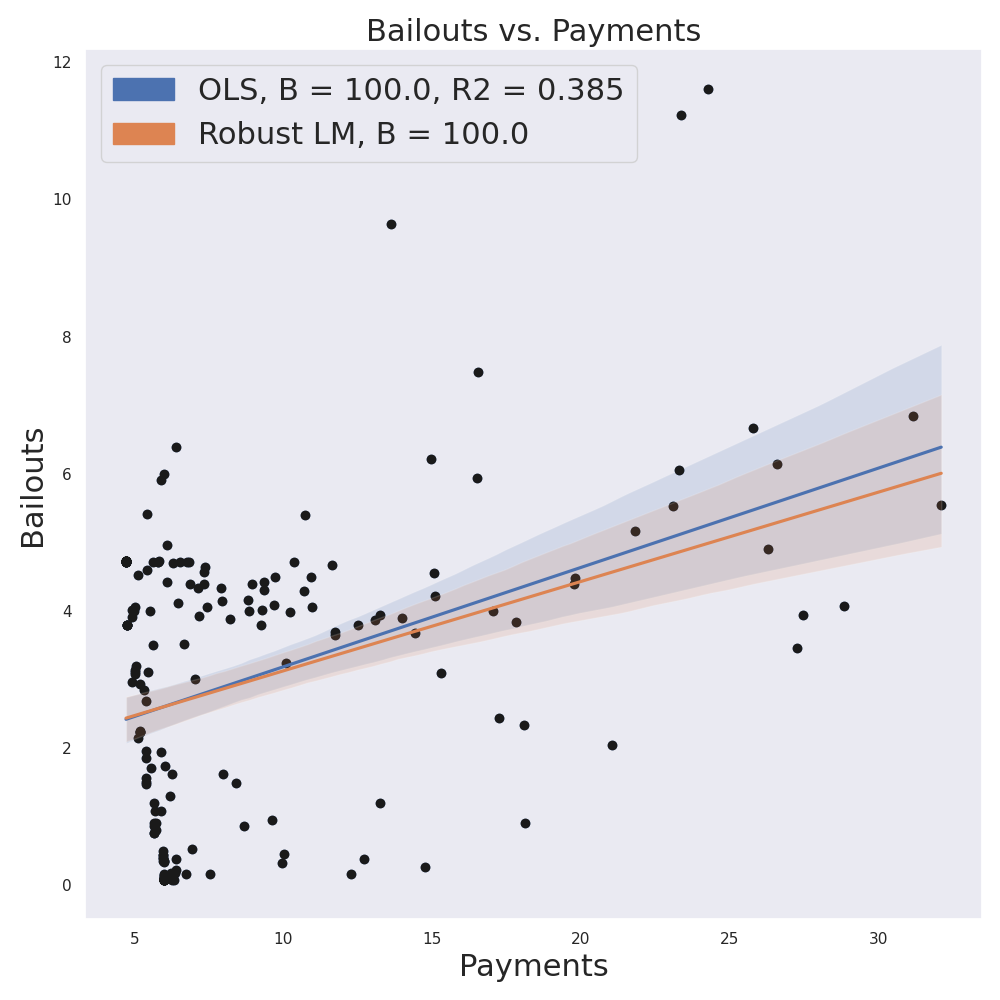}}
    \subfigure[SafeGraph (GC, $g(t) = 0.5$)]{\includegraphics[width=0.24\textwidth]{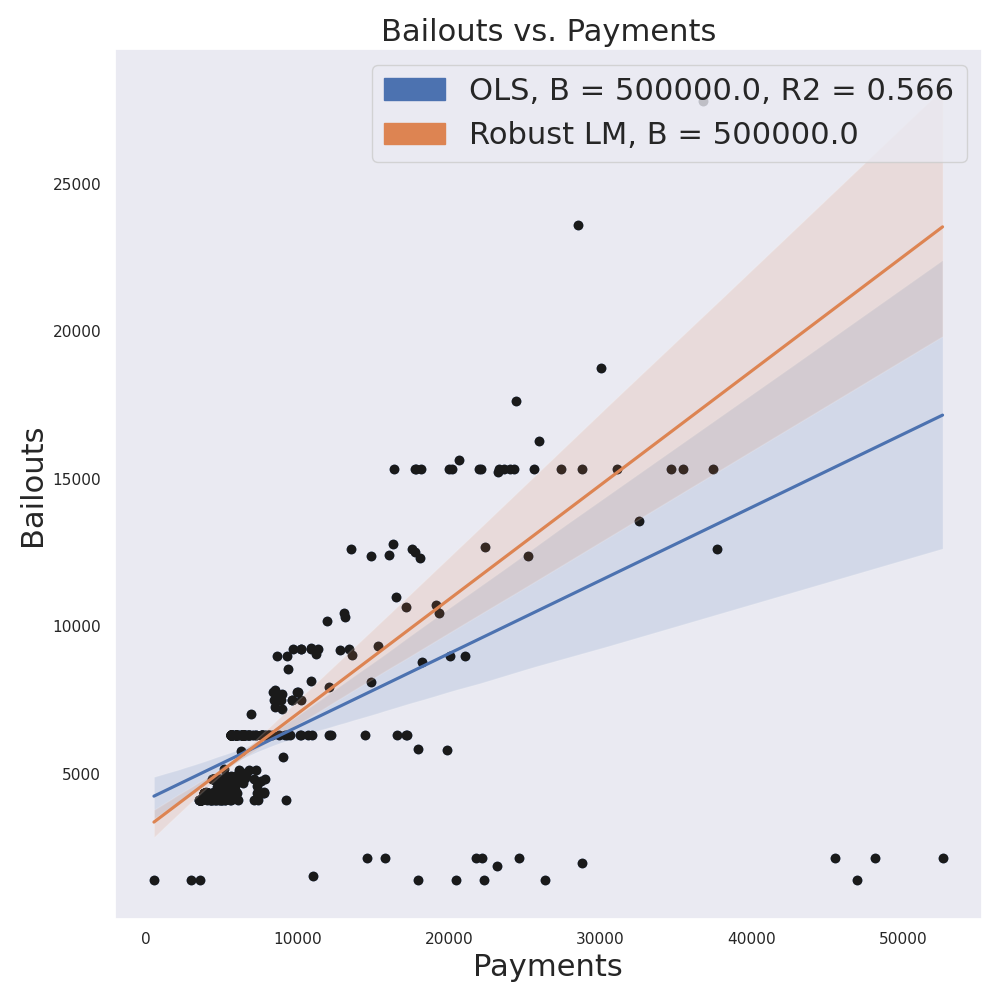}}
    \caption{\small \em Relation between the total payments of nodes and the total interventions received. We use $L = B \cdot \one$. SGC = Spatial Gini Coefficient, GC = (Standard) Gini Coefficient.}
    \label{fig:interventions_vs_payments}
\end{figure}

\begin{figure}
    \centering
    \subfigure[Synthetic Core-periphery ($g(t) = 1$)]{\includegraphics[width=0.24\textwidth]{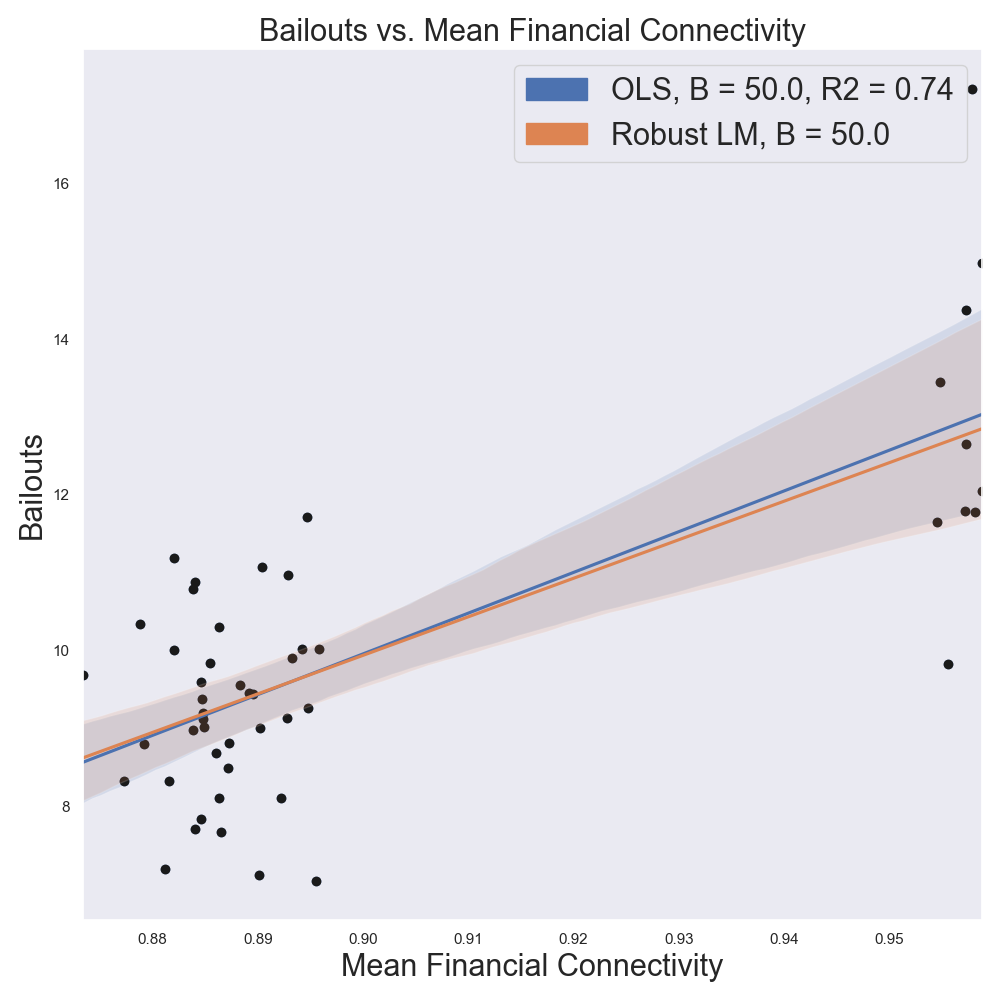}}
    \subfigure[TLC ($g(t) = 1$)]{\includegraphics[width=0.24\textwidth]{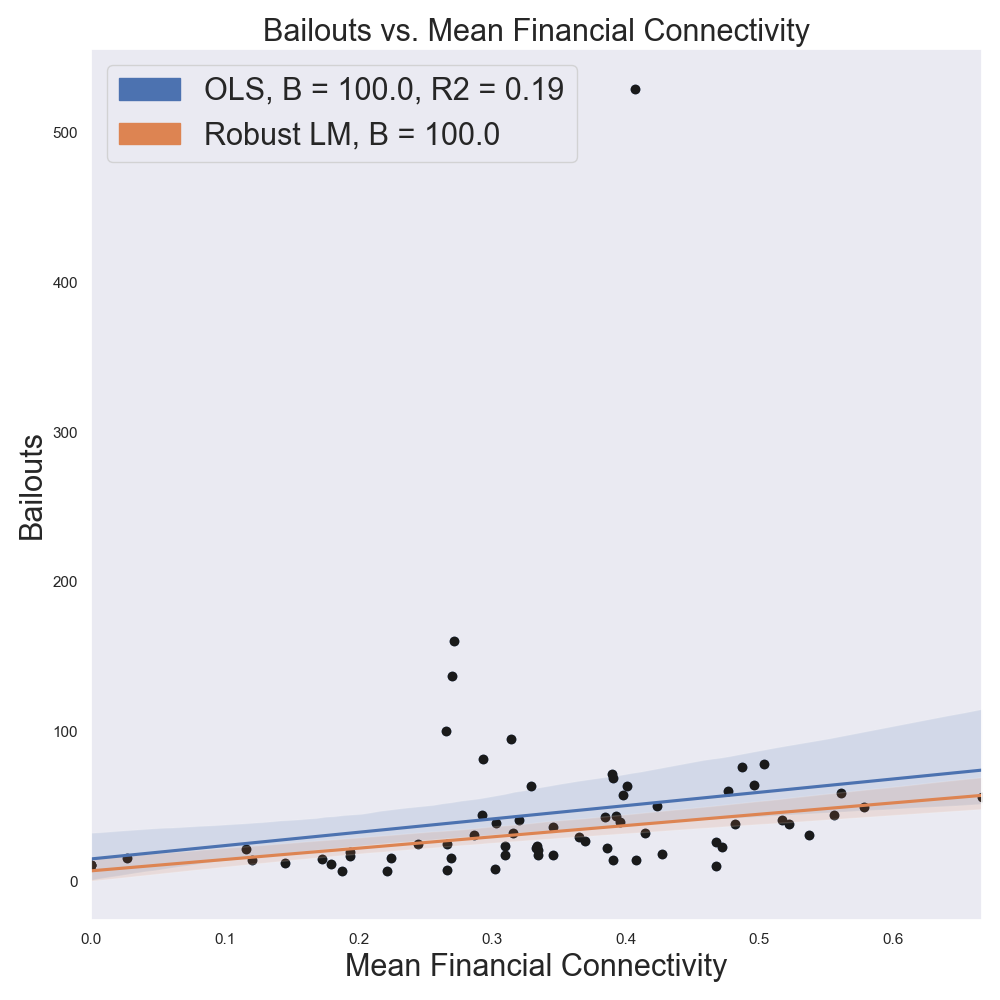}}
    \subfigure[Venmo ($g(t) = 1$)]{\includegraphics[width=0.24\textwidth]{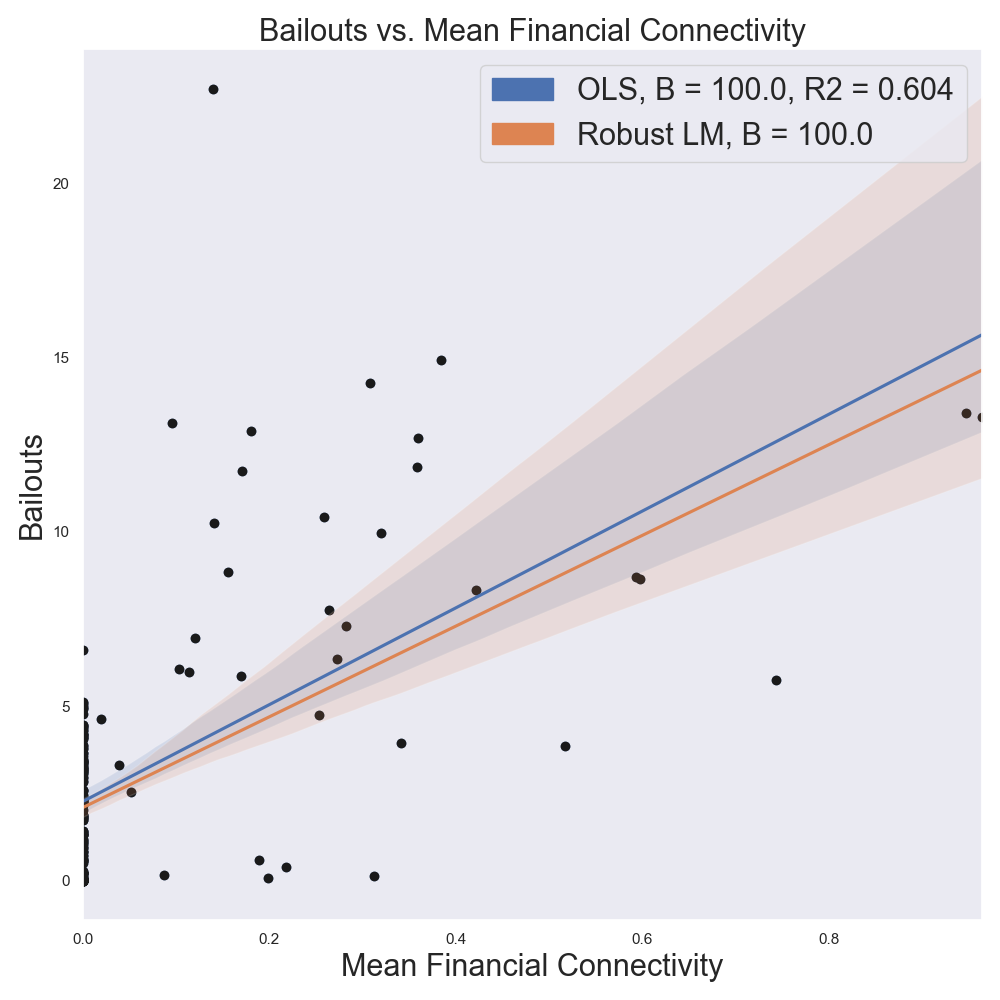}}
    \subfigure[SafeGraph ($g(t) = 1$)]{\includegraphics[width=0.24\textwidth]{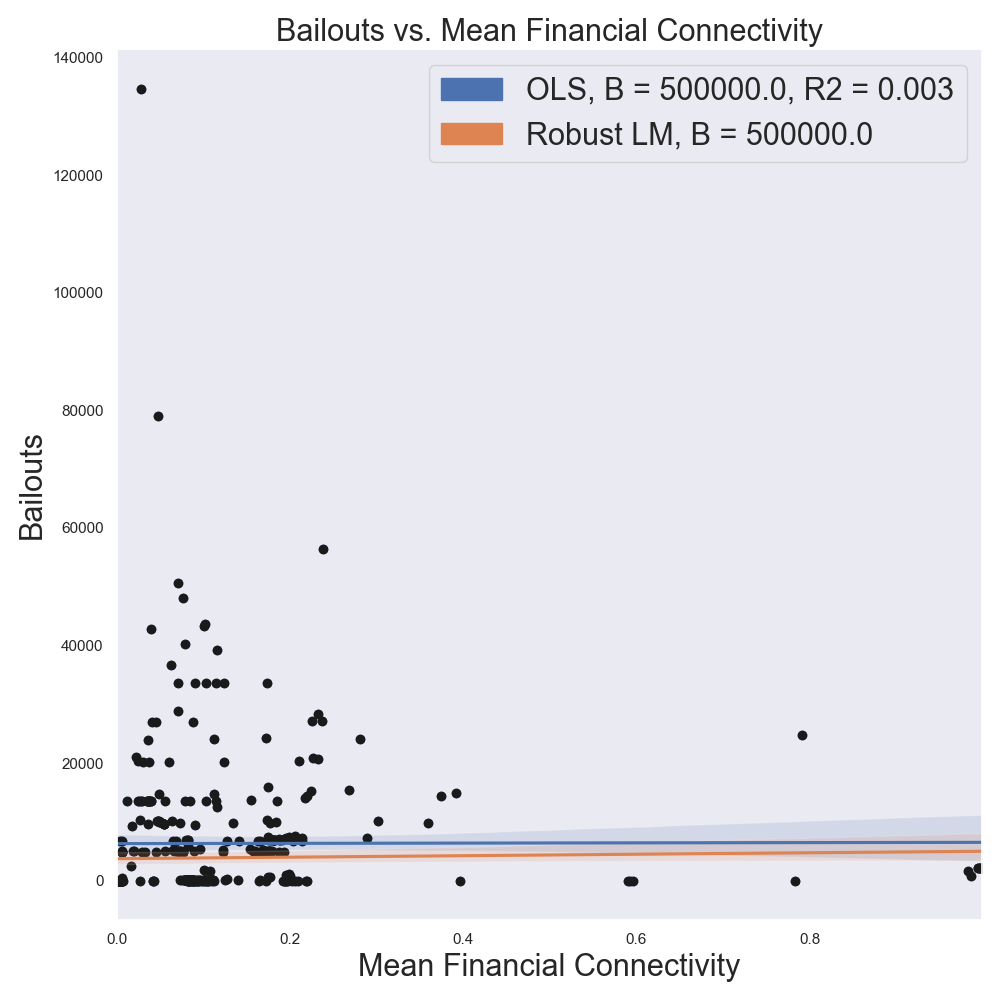}}
    \subfigure[Synthetic Core-periphery (SGC,  $g(t) = 0.5$)]{\includegraphics[width=0.24\textwidth]{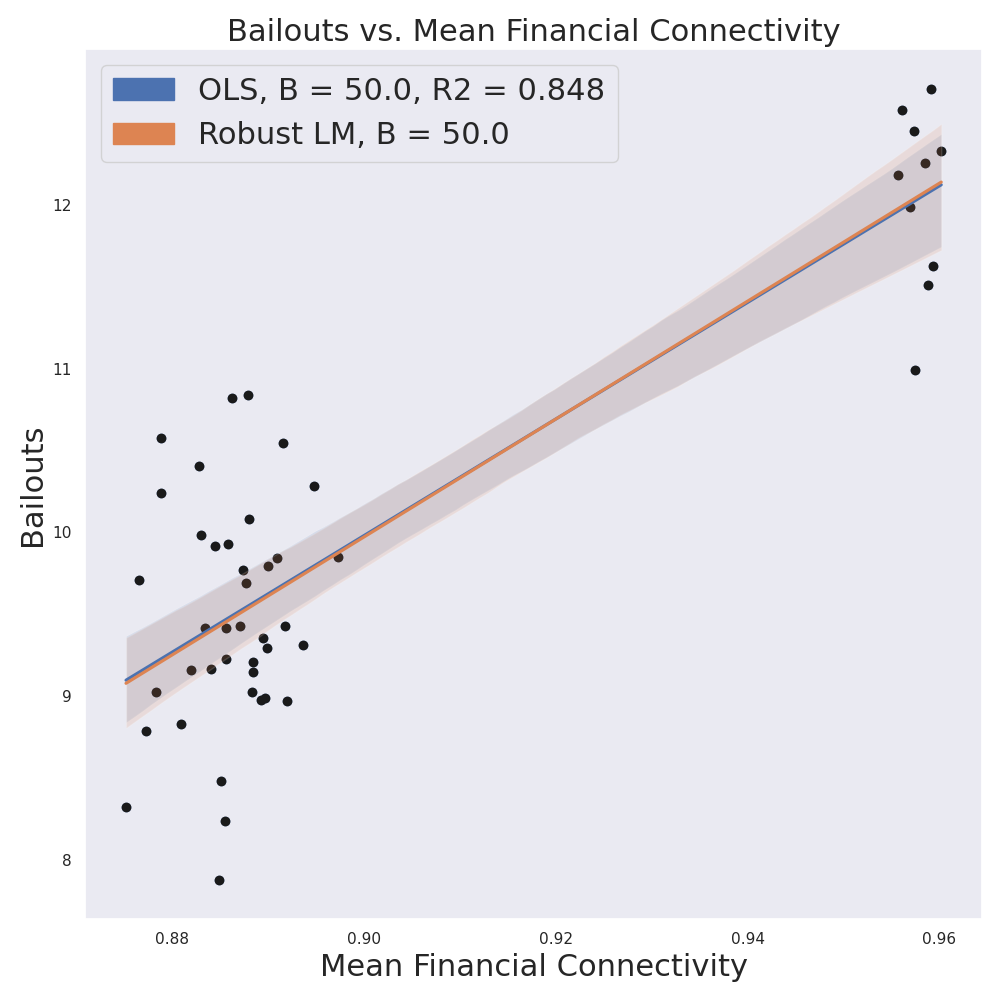}}
    \subfigure[TLC ($g(t) = 0.5$)]{\includegraphics[width=0.24\textwidth]{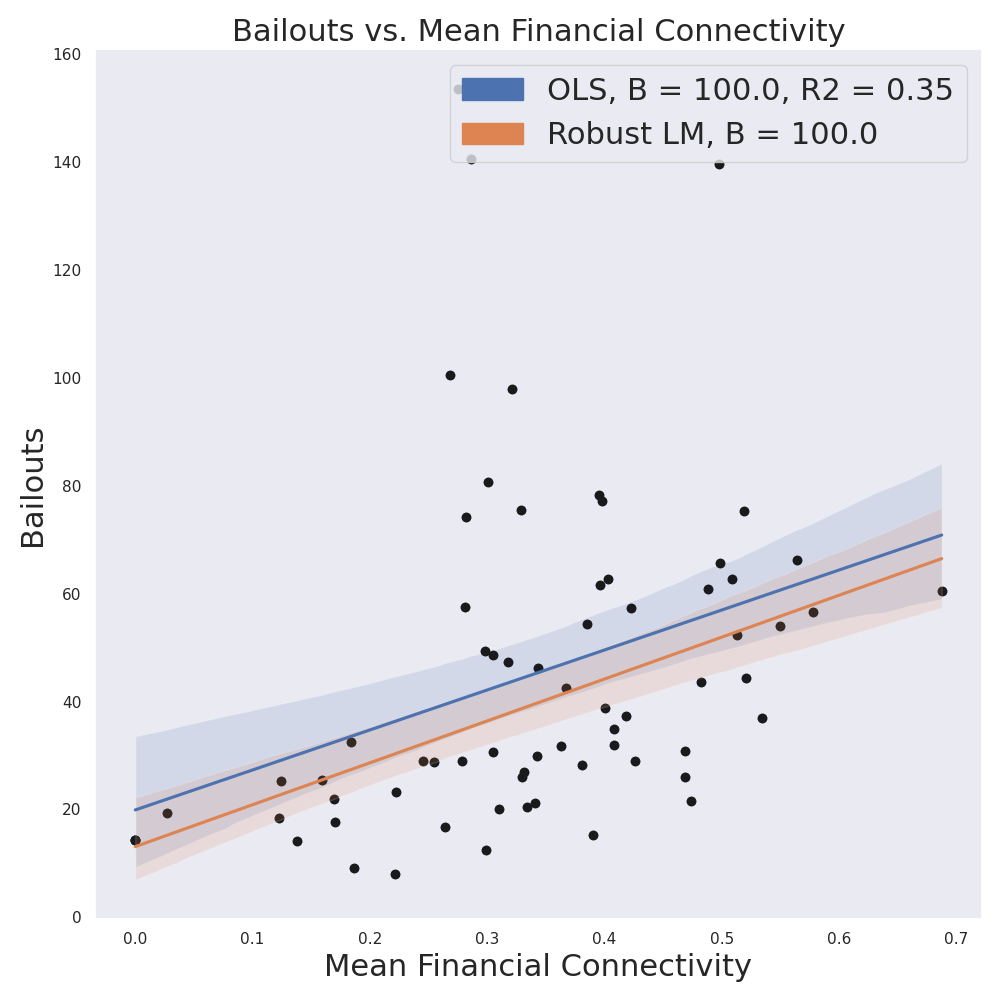}}
    \subfigure[Venmo (SGC,  $g(t) = 0.5$)]{\includegraphics[width=0.24
    \textwidth]{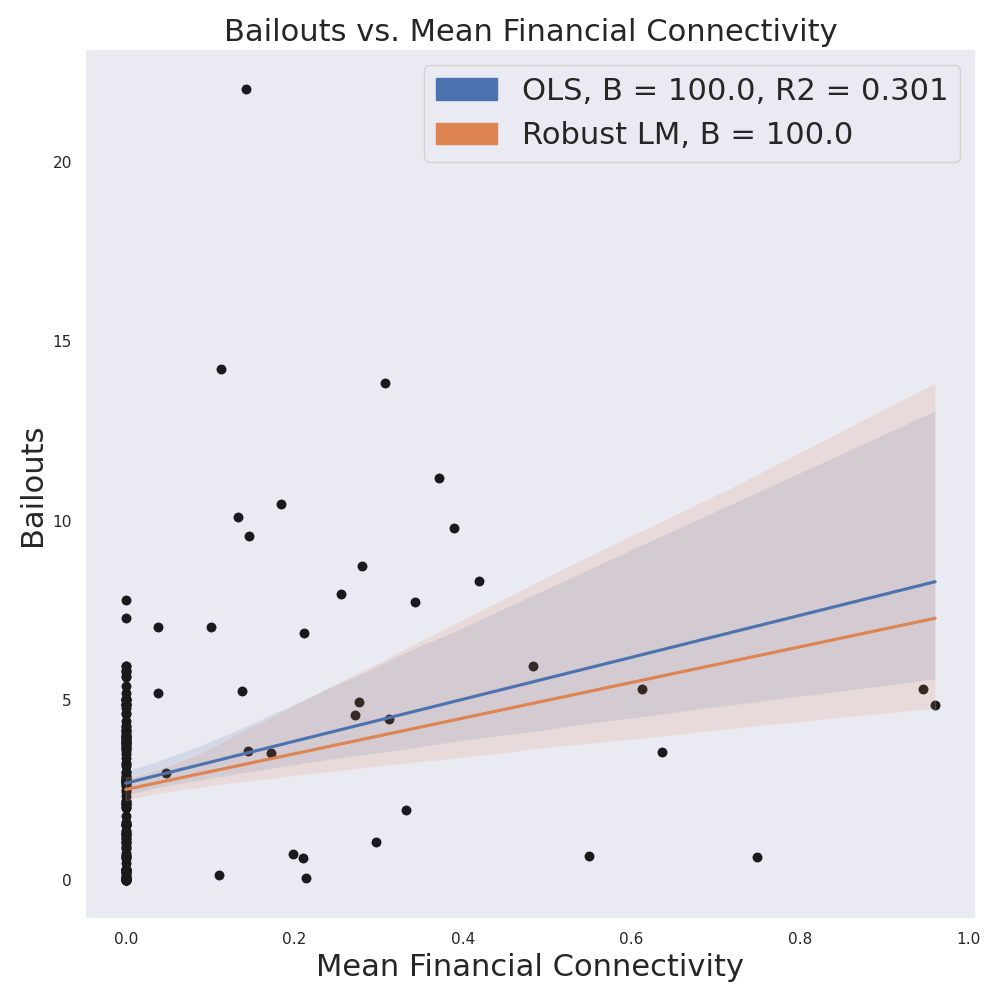}}
    \subfigure[SafeGraph (SGC,  $g(t) = 0.5$)]{\includegraphics[width=0.24\textwidth]{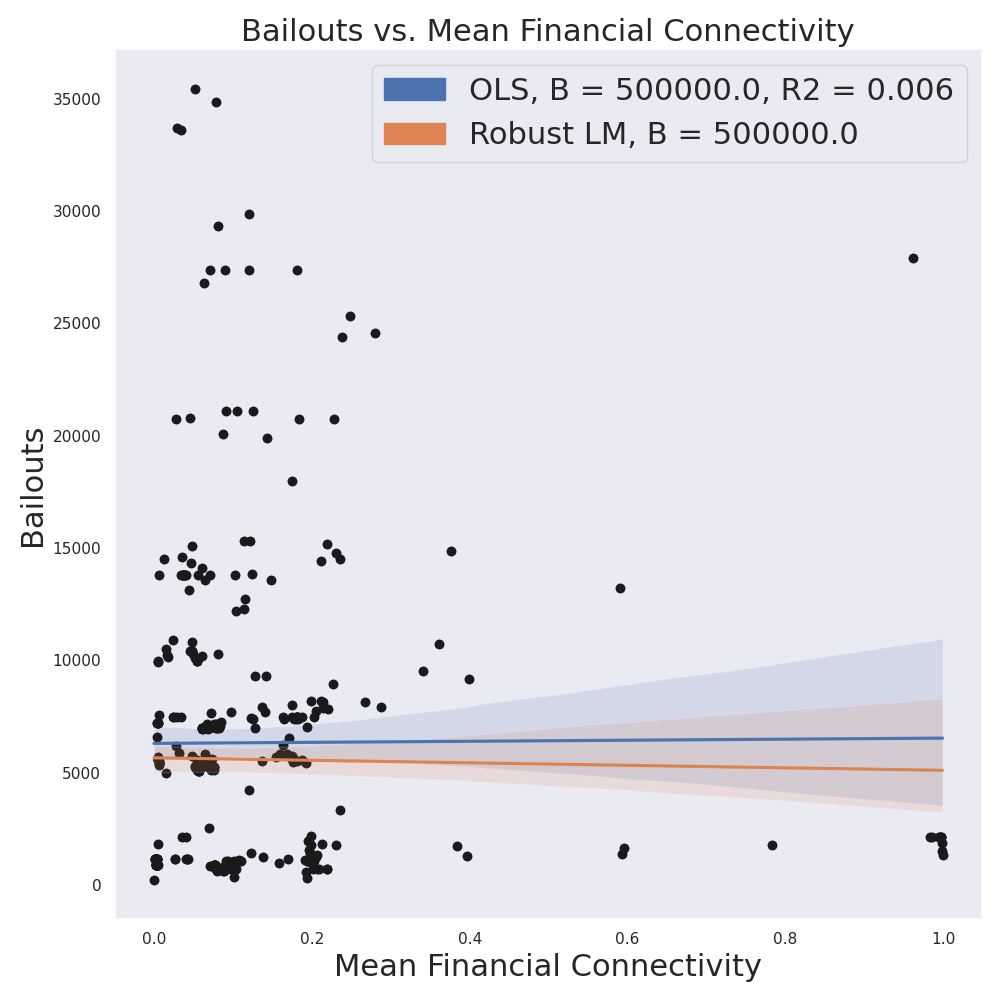}}
    \subfigure[Synthetic Core-periphery (GC,  $g(t) = 0.5$)]{\includegraphics[width=0.24\textwidth]{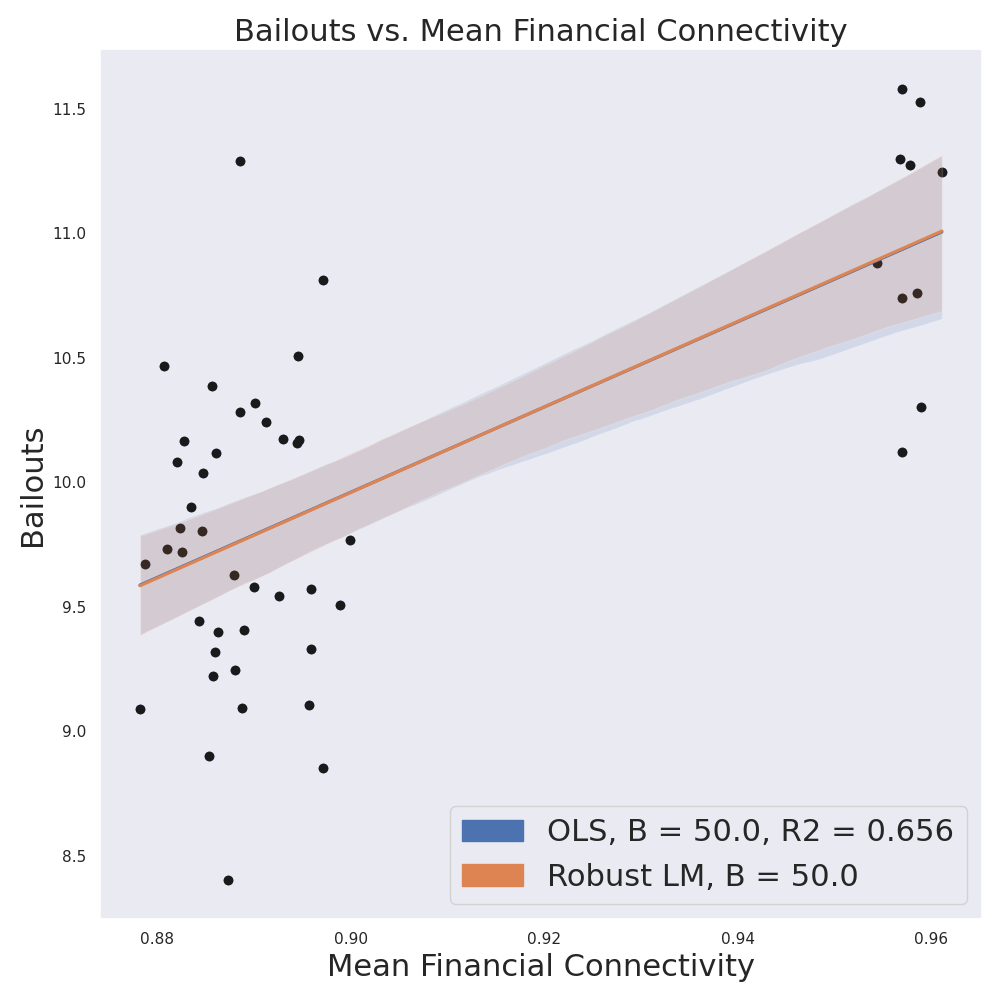}}
    \subfigure[TLC (GC,  $g(t) = 0.5$)]{\includegraphics[width=0.24\textwidth]{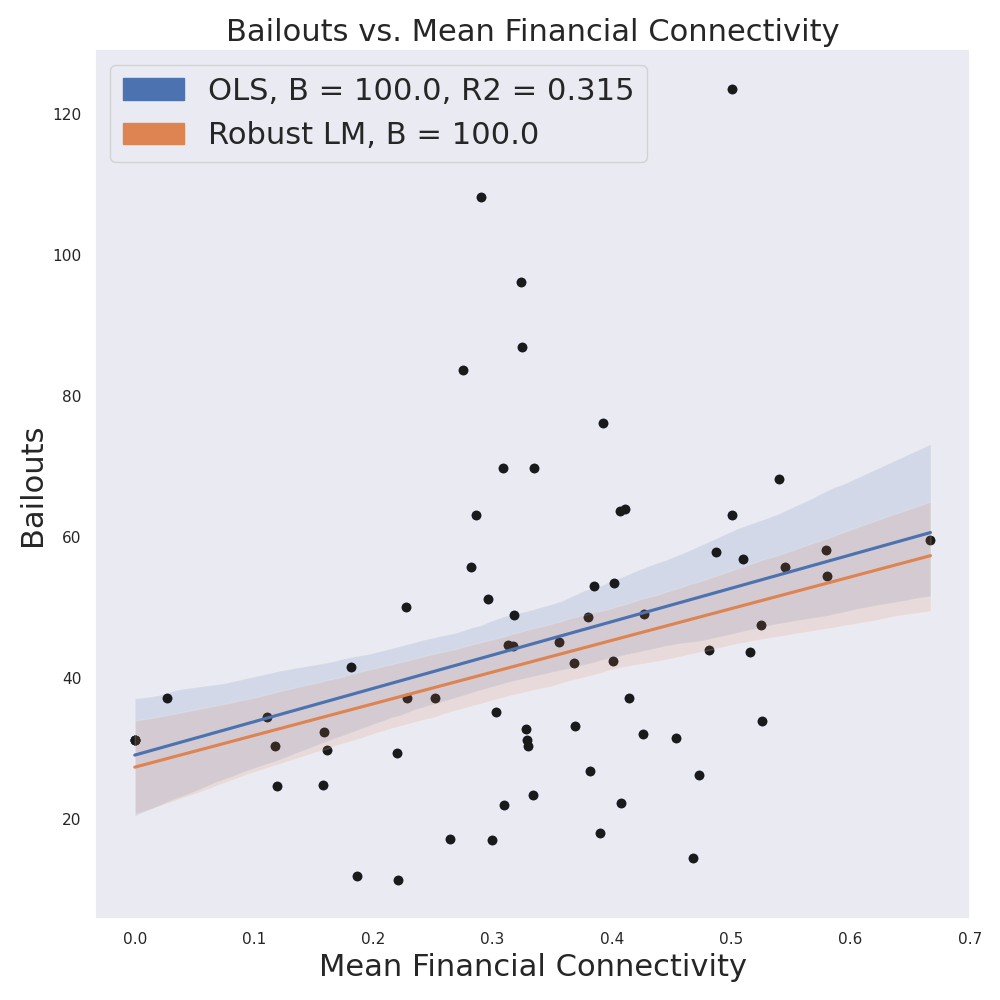}}
    \subfigure[Venmo (GC,  $g(t) = 0.5$)]{\includegraphics[width=0.24
    \textwidth]{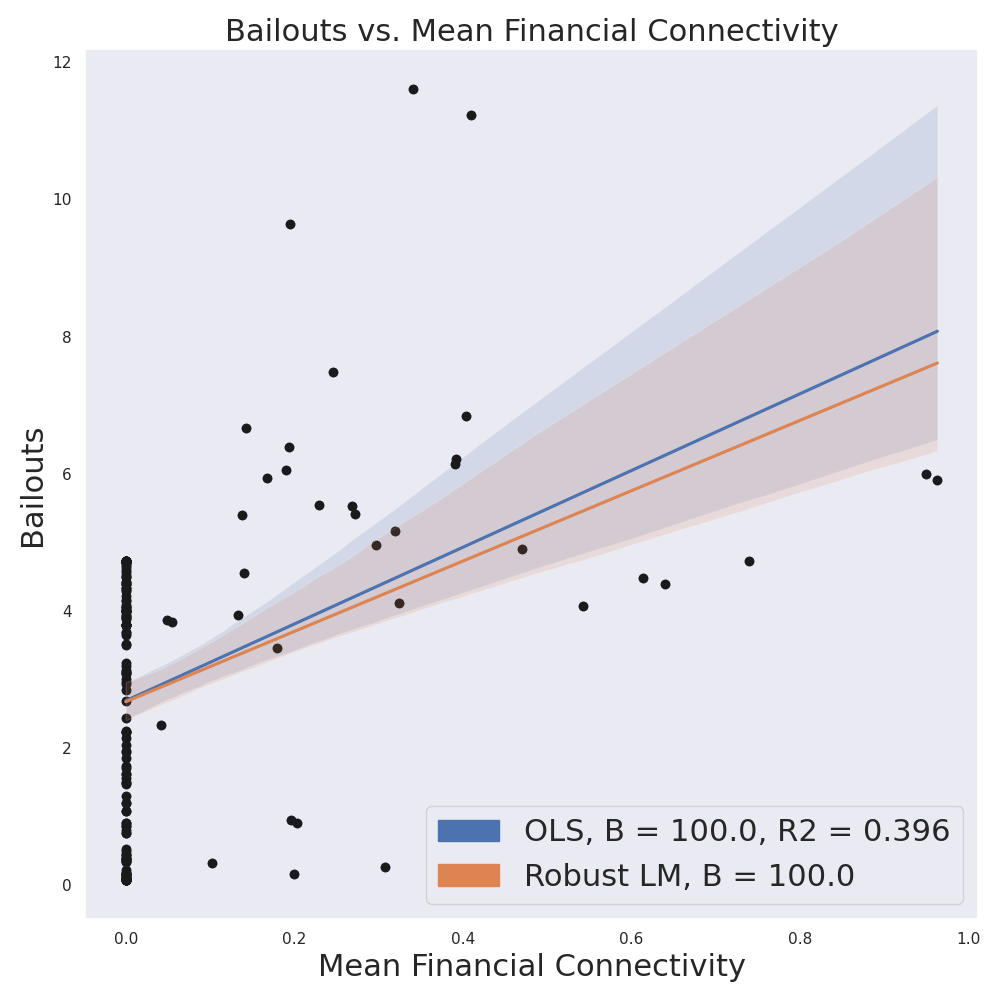}}
    \subfigure[SafeGraph (GC,  $g(t) = 0.5$)]{\includegraphics[width=0.24\textwidth]{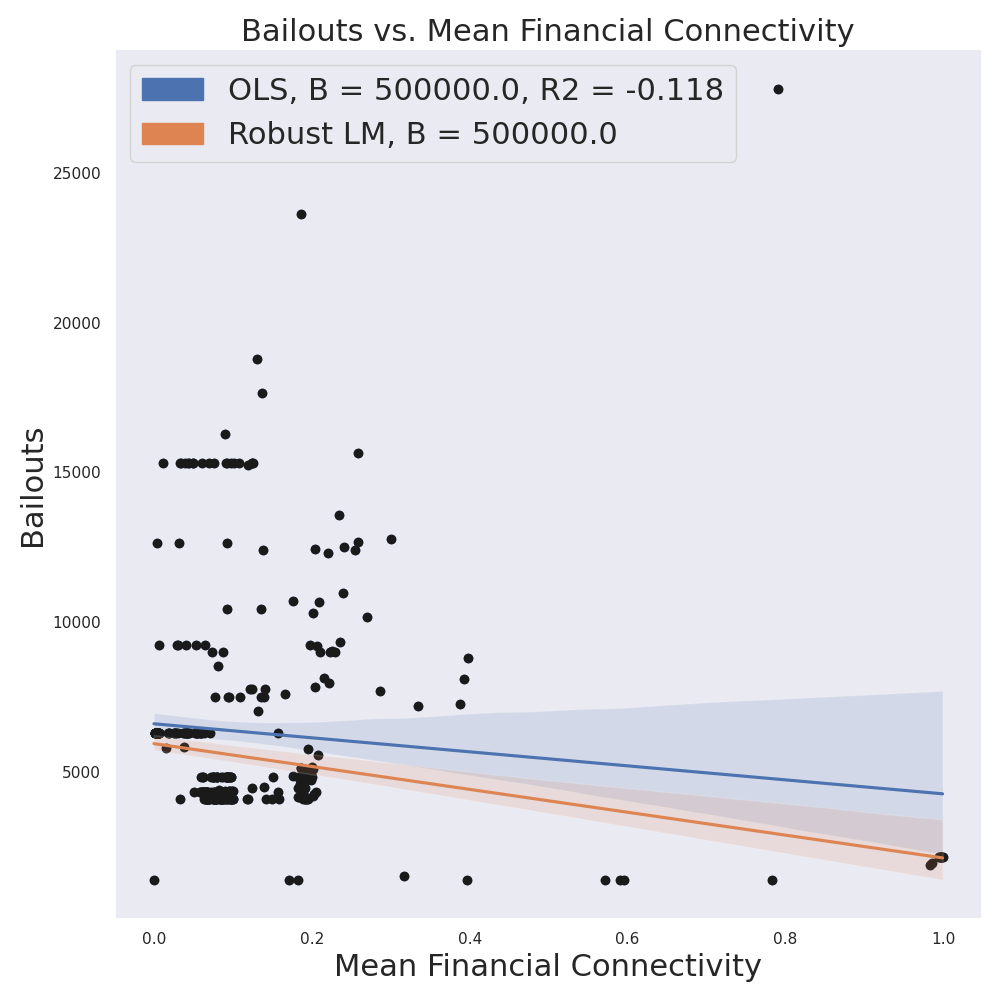}}
    \caption{\small \em Relation between the mean financial connectivity (over all rounds) of nodes and the total interventions received. We use $L = B \cdot \one$. SGC = Spatial Gini Coefficient, GC = (Standard) Gini Coefficient.}
    \label{fig:betas_vs_payments}
\end{figure}

\section{Conclusion} \label{sec:conclusion}

In this paper, we have studied a dynamic resource allocation problem through the lens of a dynamic financial contagion model, which we formulate as an MDP that evolves under a stochastic environment.  We derive the value function of the problem with fractional interventions and formulate the optimal dynamic intervention problem when the interventions are discrete, and give an approximation algorithm whose approximation factor depends on the worst-case financial connectivity of the dynamics.  Based on the development of the aforementioned algorithms, we run experiments and display the results of our method on synthetic data, ridesharing data for For-Hire-Vehicles in NYC, publicly available transaction data from Venmo, and financial networks derived from mobility data (SafeGraph).  We incorporate fairness measures that involve fairly distributing the resources and experimentally observe that the proposed fairness measures can achieve near-optimal fairness to welfare relationship (i.e., PoF close to 1). 

There are several potential future roadmaps for our work.  Firstly, it is currently unknown whether this approximation ratio can be improved (and by how much) even in the case of $T = 1$, and how the policy that solves the MDP conditioned on the fact that the system responds optimally at each step (i.e., \cref{assumption:optimal_response}) compares towards the globally optimal policy that solves the general optimization problem (see \cref{sec:general_dynamics}), as well as algorithms for finding the globally optimal policy.  Lastly, interesting is the question of rounding the clearing variables so that we achieve a solution that is provably within some factor from the optimal policy. 

\section*{Ethics Statement} 

The work studies a theoretical model of financial contagion and applies experiments to data. Both parts of the paper does not pose ethical concerns. The Venmo data and the corresponding usernames used in the paper are publicly available in \url{https://github.com/sa7mon/venmo-data}. The SafeGraph data were obtained from \url{https://www.safegraph.com} with an academic license. Differential privacy techniques have been applied by the data provider, and, thus, no human subjects are identifiable. Because all data consists of public, pre-existing datasets without identifiable individuals, the current work is exempt from IRB review.

\newpage

\section*{Acknowledgements}

Supported in part by a Simons Investigator Award, a Vannevar Bush Faculty Fellowship, AFOSR grant FA9550-19-1-0183, a Simons Collaboration grant, and a grant from the MacArthur Foundation. SB gratefully acknowledges funding from the NSF under grants ECCS 1847393 and CNS 1955997. 

The authors would like to thank Constantine Caramanis, Akhil Jalan, Emma Pierson, David Shmoys, and Andrew Wang for helpful discussions regarding our paper.

\begin{figure}[t]
    \centering
    \subfigure[$B = 0$ (no external assistance)]{\includegraphics[width=0.49\textwidth]{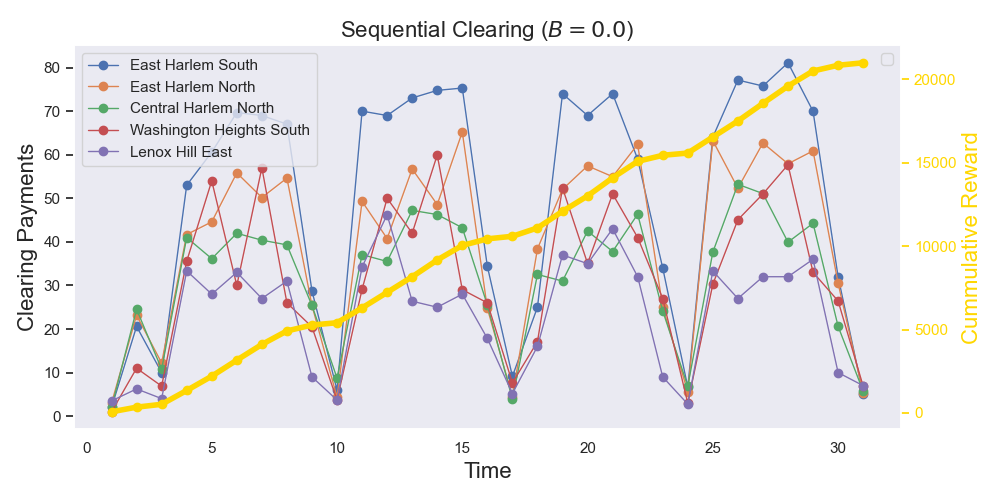}}
    \subfigure[$\max_{i \in [n]} \beta_i(t)$]{\includegraphics[width=0.49\textwidth]{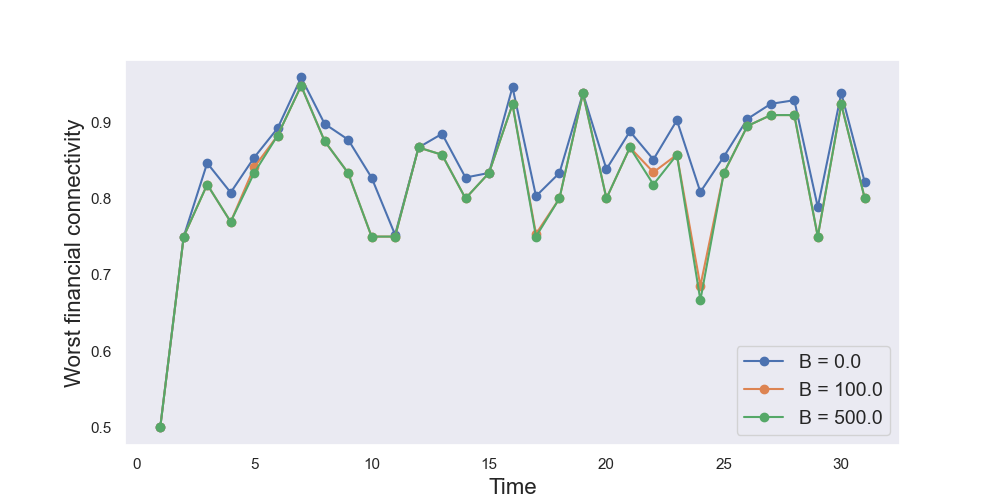}}
    \subfigure[$B = 100$ (fractional allocations)]{\includegraphics[width=0.49\textwidth]{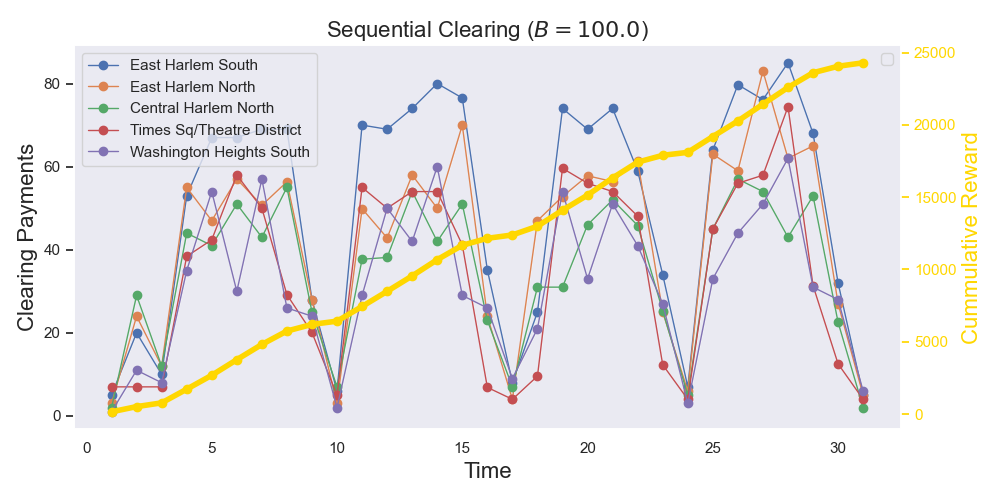}}
    \subfigure[$B = 100$ (fractional allocations)]{\includegraphics[width=0.49\textwidth]{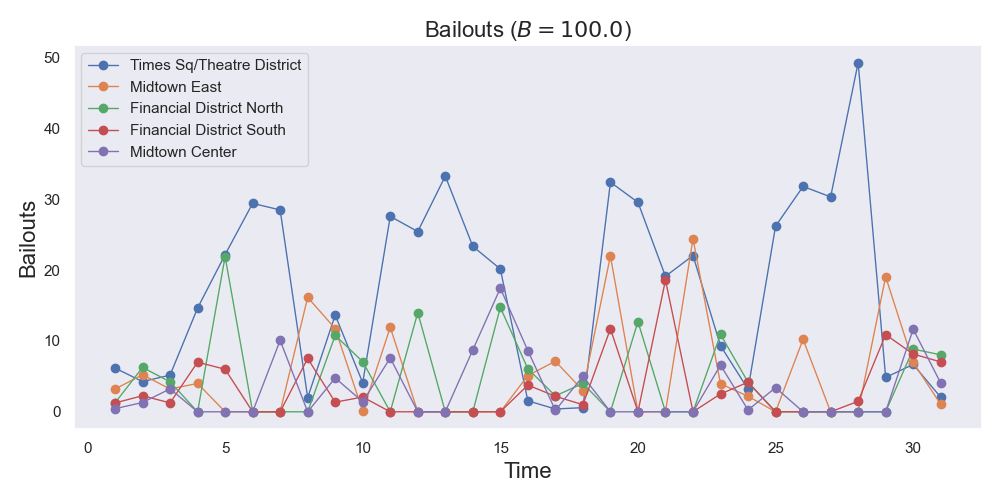}}
    \subfigure[$B = 500$ (fractional allocations)]{\includegraphics[width=0.49\textwidth]{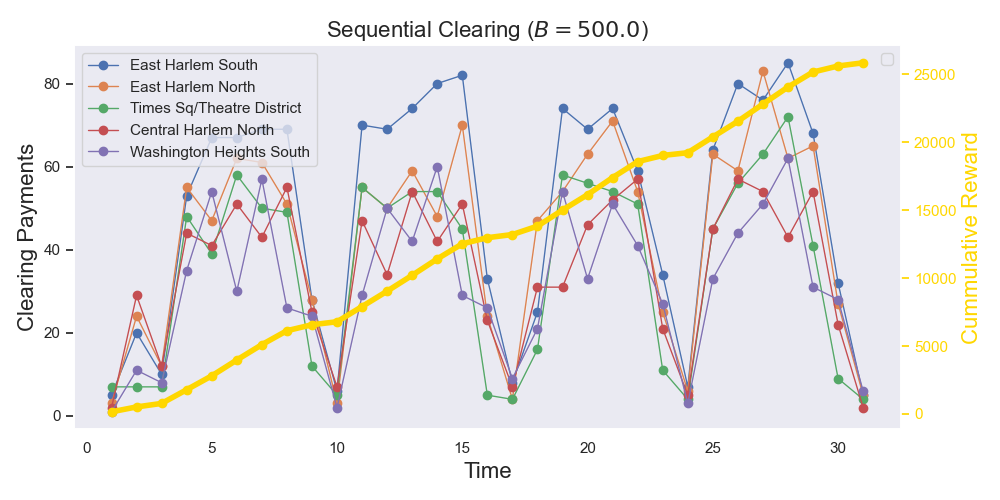}}
    \subfigure[$B = 500$ (fractional allocations)]{\includegraphics[width=0.49\textwidth]{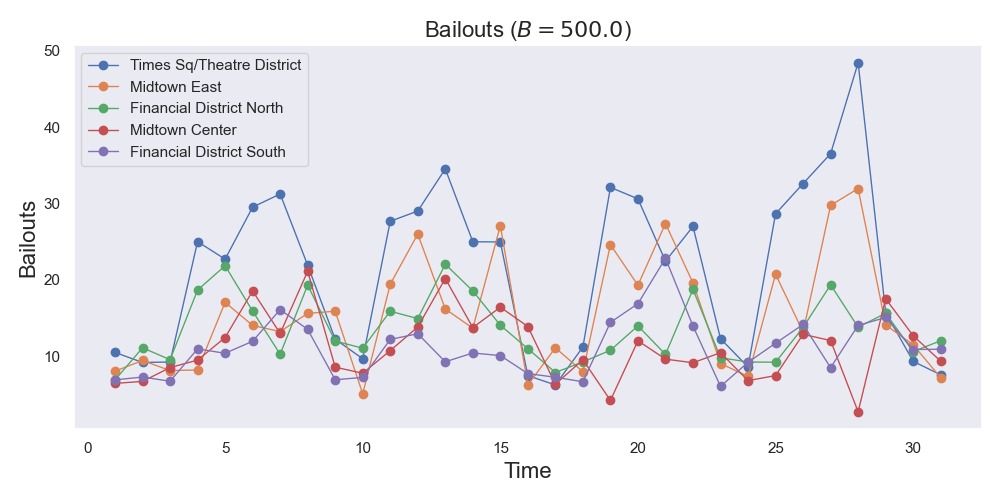}}
    \caption{\small \em TLC FHV Data for January 2021 split into days. We report the 5 most active neighborhoods (in terms of total cleared value) as well as the cumulative reward amassed.}
    \label{fig:tlc_days}
\end{figure}

\begin{figure}[t]
    \centering
    \subfigure[]{\includegraphics[width=0.49\textwidth]{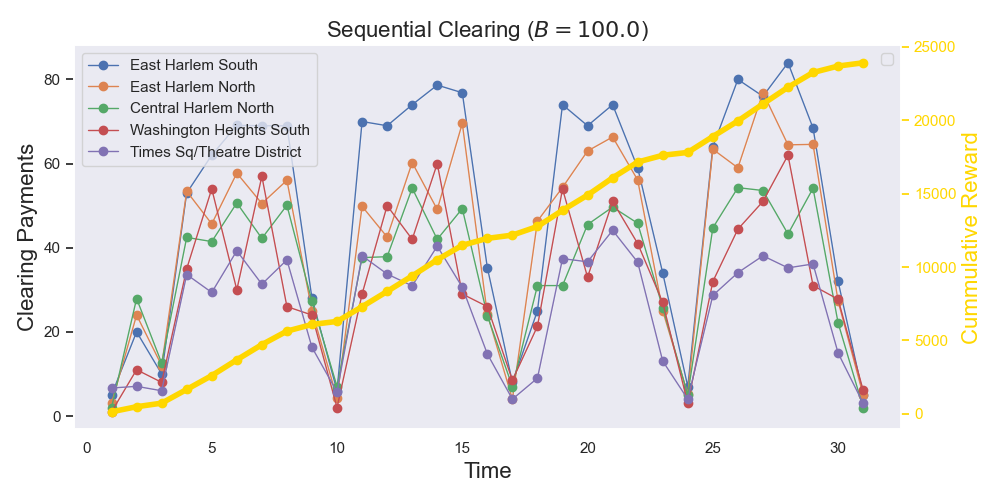}}
    \subfigure[]{\includegraphics[width=0.49\textwidth]{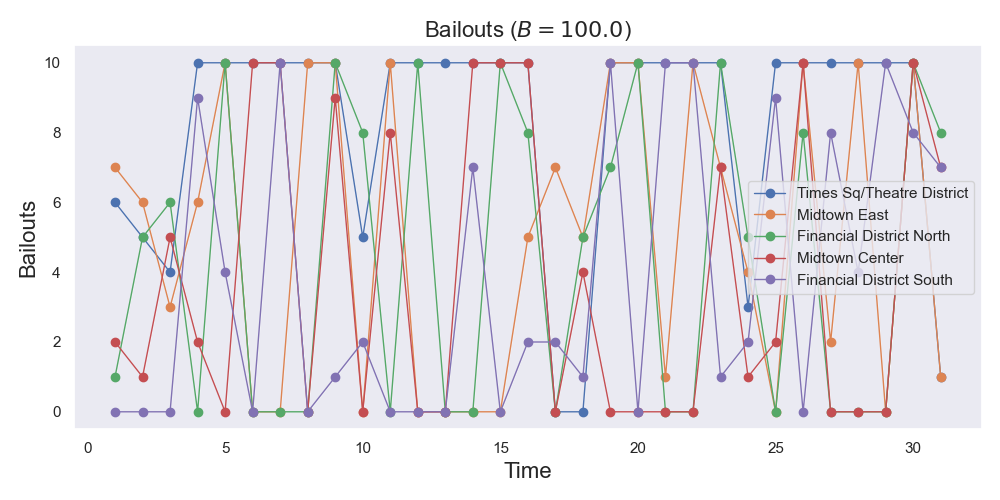}}
    \caption{\small \em Discrete intervention Scenario for $B = 100$ and $L = 10 \cdot \one$.}
    \label{fig:tlc_discrete}
\end{figure}

\bibliographystyle{alpha}
\bibliography{main}

\appendix

\section{Proofs} \label{sec:proofs}

\subsection{Proof of \cref{theorem:value_function}}

\noindent \emph{Base Case $t' = T$.} Let $f_T$ be the PDF of $U(T)$. We have that 

\begin{equation*}
    \begin{split}
        V(T, s) & = \max_{z \in \mathcal Z, \til p \text{ fixed point given $z$}} \langle f_T, \one^T \til p \rangle \\
        & = \max_{z \in \mathcal Z} \int_{\mathcal U} f_T(u) \max_{\til p \text{ fixed point given $z$}} \one^T \til p(u) du \\
        & = \max_{z \in \mathcal Z} \lim_{N \to \infty} \sum_{i = 1}^N f(u_i) \Delta u_i \max_{\til p(u_i) \text{ fixed point given $z$}}  \one^T \til p(u_i) \\
        & = \max_{z \in \mathcal Z} \lim_{N \to \infty} \max_{(p(u_1)^T, \dots, p(u_N)^T)^T \text{fixed point given } (z^T, \dots, z^T)^T} \sum_{i = 1}^N f(u_i) \Delta u_i \sum_{j \in [n]} \til p_j(u_i) \\
        & =  \lim_{N \to \infty} \max_{(p(u_1)^T, \dots, p(u_N)^T)^T \text{fixed point given } (z^T, \dots, z^T)^T} \sum_{i = 1}^N f(u_i) \Delta u_i \sum_{j \in [n]} \til p_j(u_i) \\
        & = \int_{\mathcal U} f_T(u) \max_{z, \til p(u)} \one^T \til p(u) du \\
        & = \ev {U(T)} {\max_{z, \til p} \one^T \til p}
    \end{split}
\end{equation*}

The equalities follow from: (i) definition of expectation, (ii) pushing the maximization wrt $s$ inside since $z$ is fixed, (iii) definition of Riemannian integral, (iv) the $N$ optimization problems being decoupled since $z$ is fixed and thinking of the optimization as a large problem with a state vector of dimension $2 \times N \times n$, (v) pushing the optimization inside the limit (rewards are bounded regardless of the value of $N$ and the corresponding mappings are continuous), (vi) definition of integration, and (vii) definition of the expected value.

\noindent \emph{Inductive Hypothesis.} Assume that for $t' = t + 1$ we have that 

\begin{equation*}
    V(t+1, s_t) = \ev {U(t+1:T)} {\max_{z_{t + 1}, \til p_{t + 1}} \left \{  \one^T \til p_{t + 1} + \max_{z_{t + 2}, \til p_{t + 2}} \left \{ \one^T \til p_{t + 2} + \max_{z_{t + 3}, \til p_{t + 3}} \left \{ \one^T \til p_{t + 3} + \dots \right \} \right \} \right \} }
\end{equation*}

\noindent \emph{Inductive Step.} For $t' = t$ we have that

\begin{equation*}
\footnotesize
    \begin{split}
        V(t, s_t) & = \max_{z_t} \left \{ r(s_t, z_t) + \ev {s_{t + 1} \sim \mathcal T(s_t, z_t)} {V(t + 1, s_{t + 1})} \right \} \\
        & = \max_{z_t} \mathbb E_{U(t)} \bigg [\max_{\til p_{t}} \one^T \til p_t + {\ev {U(t+1:T)} {\max_{z_{t + 1}, \til p_{t + 1}} \left \{  \one^T \til p_{t + 1} + \max_{z_{t + 2}, \til p_{t + 2}} \left \{ \one^T \til p_{t + 2} + \max_{z_{t + 3}, \til p_{t + 3}} \left \{ \one^T \til p_{t + 3} + \dots \right \} \right \} \right \} }} \bigg ] \\
        & = \max_{z_t} \mathbb E_{U(t:T)} \bigg [\max_{\til p_{t}} \left \{ \one^T \til p_t +  {\max_{z_{t + 1}, \til p_{t + 1}} \left \{  \one^T \til p_{t + 1} + \max_{z_{t + 2}, \til p_{t + 2}} \left \{ \one^T \til p_{t + 2} + \max_{z_{t + 3}, \til p_{t + 3}} \left \{ \one^T \til p_{t + 3} + \dots \right \} \right \} \right \} } \right \}
        \\
        & = \ev {U(t:T)} {\max_{z_{t}, p_{t}} \left \{  \one^T \til p_{t} + \max_{z_{t + 1}, \til p_{t + 1}} \left \{ \one^T \til p_{1 + 2} + \max_{z_{t + 2}, \til p_{t + 2}} \left \{ \one^T \til p_{t + 2} + \dots \right \} \right \} \right \} }
    \end{split}
\end{equation*}

The equalities follow from: (i) the HJB equations, (ii) the inductive hypothesis, (iii) the fact that the maximization over $\til p_t$ is independent of the sample paths from round $t + 1$ onwards and thus we can reorganize the expectations into one expectation over sample paths $U(t:T) \sim \mathcal U$, (iv) identically to the base case argument.

At any point, with probability 1, the value functions $V_{u_i(t:T)}$ are between $0$ and $\sum_{t} \one^T (b(t) + \ell(t)) \le (T - t + 1) \cdot \Delta$, where $\Delta = \sup_{\mathcal U} \left ( \| b \|_1 + \| \ell \|_1 \right )$ since the maximum reward can be achieved when all debts are paid and all nodes are solvent. Thus, by standard Chernoff bounds, one needs to choose $N = \tfrac {\log (2 / \delta) (T - t + 1)^2 \Delta^2 } {2 \varepsilon^2}$ samples to get an $\varepsilon$-accurate estimation of the actual value function with probability at least $1 - \delta$. 

\qed

\subsection{Proof of \cref{theorem:approximation_randomized_rounding}} 

\subsubsection{Approximation Guarantee}

Let $D_d(t')$ and $R_d(t')$ are the sets of default nodes and solvent nodes under discrete interventions at round $t'$. We have that 

\begin{enumerate}
    \item If $i \in D_d(t')$ we have that $\ev {z_d(t:T)} {\til p_{d, i}(t') | i \in D_d(t')} \ge c_i(t') + z_{r, i}^*(t')$.
    \item If $i \in R_d(t')$ we have that 
    
    \begin{align*}
        \ev {z_d(t:T)} {\til p_{d, i}(t') | i \in R_d(t')} & = \ev {z_d(t:T)} {p_{d, i}(t') | i \in R_d(t')} \\ & = \ev {z_d(t:T)} {\sum_{t'' < t'} (b_i(t'') + \ell_i(t'')) - \sum_{t'' < t'} \til p_{d, i}(t'')} \\ & \ge \ev {z_d(t:T)} {\sum_{t'' < t'} (b_i(t'') + \ell_i(t'')) - \sum_{t'' < t'} \til p^*_{r, i}(t'')} \\ & = \ev {z_d(t:T)} {p_{r, i}(t')} \\ & \ge \til p_{r, i}^*(t') \\ & \ge (1 - \max_{i \in [n]} \beta_{r, i}(t') ) \til p_{r, i}^*(t'). 
    \end{align*}
    
    The statement follows from: (i) definition of solvent node in the rounded solution, (ii) recursively using the definition of $p_{d, i}(t'')$ for all $1 \le t'' < t'$, (iii) point-wise optimality of the fractional clearing vector, (iv) definition of the solvency constraint for the fractional relaxation, (v) feasibility of the fractional solution, (vi) $\max_{i \in [n]} \beta_{r, i}(t') > 0$ by \cref{assumption:uniqueness}. 
\end{enumerate}

Moreover, by \cite{papachristou2021allocating} we know that for every subset $S \subseteq [n]$ we have that the fractional solution satisfies

\begin{equation} \label{eq:fractional_interventions_ub}
    \left ( 1 - \max_{i \in S} \beta_{r, i}(t') \right ) \sum_{i \in S} \til p_{r, i}^* \le \sum_{j \in S} \left [ c_i(t') + z_{r, i}^*(t') \right ].
\end{equation}

By letting $S = D_d(t')$ on \cref{eq:fractional_interventions_ub} and since $\max_{i \in D_d(t')} \beta_{r, i}(t') \le \max_{i \in [n]} \beta_{r, i}(t')$ we have that

\begin{equation*}
    \left ( 1 - \max_{i \in [n]} \beta_{r, i}(t') \right ) \sum_{i \in D_d(t')} \til p_{r, i}^* \le \sum_{j \in D_d(t')} \left [ c_i(t') + z_{r, i}^*(t') \right ].
\end{equation*}

Moreover, 

\begin{equation*}
    p_{r, i} (t') \ge \til p_{r, i}^*(t') \ge \left ( 1 - \max_{i \in [n]} \beta_{r, i} (t')\right ) \til p_{r, i}^*(t')
\end{equation*}

where the second inequality is due to feasibility and the last inequality is because we multiply with a quantity that is strictly in $(0, 1)$. Therefore we have that the expected reward of the rounded solution at time $t'$ is at least $\left ( 1 - \max_{i \in [n]} \beta_{r, i} (t') \right )$ the optimal reward, i.e

\begin{equation*}
    \ev {z_d(t:T)} {R(s(t'), z(t') = SOL)} \ge \left ( 1 - \max_{i \in [n]} \beta_{r, i} (t') \right ) \cdot R(s(t'), z(t') = OPT) \quad \forall t' \in [t,T]
\end{equation*}

We sum over $t' \in [t, T]$ and have that 

\begin{equation*}
    \begin{split}
        \ev {z_d(t:T)} {V_{u(t:T)}^{SOL} (t, s(t))} & = \sum_{t' \in [t,T]} \ev {z_d(t:T)} {R(s(t'), z(t') = SOL)} \\
        & \ge \sum_{t' \in [t,T]} (1 - \max_{i \in [n]} \beta_{r, i}(t')) \cdot R(s(t'), z(t') = OPT) \\
        & \ge \min_{t \in [t,T]} (1 - \max_{i \in [n]} \beta_{r, i}(t')) \cdot \sum_{t' \in [t,T]} R(s(t'), z(t') = OPT) \\
        & = \left ( 1 - \max_{t' \in [t, T], i \in [n]} \beta_{r, i}(t') \right ) \cdot V_{u(t:T)}^{OPT}(t, s(t)).
    \end{split}
\end{equation*}

Taking expectation with respect to $u(t:T) \sim U(t:T)$ and arrive to

\begin{equation*} 
\begin{split}
    \ev {z_d(t:T), u(t:T)} {V^{SOL} (t, s(t))} & \ge \ev {u(t:T)} {\left ( 1 - \max_{t' \in [t, T]} \max_{i \in [n]} \beta_{i}(t') \right ) V_{u(t:T)}^{OPT} (t, s(t))} \\
    & = \ev {u(t:T)} {V^{OPT} (t, s(t))} - \ev {u(t:T)} {\left | V_{u(t:T)}^{OPT} (t, s(t)) \max_{t' \in [t, T]} \max_{i \in [n]} \beta_{i}(t') \right |} \\
    & \overset {\text{H\"older}} \ge \ev {u(t:T)} {V^{OPT} (t, s(t))} \\ & - \ev {u(t:T)} {\left \| V_{u(t:T)}^{OPT} (t, s(t)) \right \|_1} \sup_{u(t:T)} {\max_{t' \in [t, T]} \max_{i \in [n]} \beta_{i}(t')} \\
    & \overset {V \ge 0} = \ev {u(t:T)} {V^{OPT} (t, s(t))} \\ & - \ev {u(t:T)} {V^{OPT} (t, s(t))} \sup_{u(t:T)} \max_{t' \in [t, T]} \max_{i \in [n]} \beta_{i}(t') \\
    & = \left ( 1 - \sup_{u(t:T)} \left \{ \max_{t' \in [t, T]} \max_{i \in [n]} \beta_{i}(t') \right \} \right ) \cdot \ev {u(t:T)} {V^{OPT} (t, s(t))}.
\end{split}
\end{equation*}

\noindent \textbf{Simplification of approximation guarantee.} Having $B > \Delta$ implies that at each round all liabilities can be covered and therefore $P_i(t') = b_i(t') + \ell_i(t') \le \Delta$, thus 

\begin{align*}
    \beta_i(t') = 1 - \frac {b_i(t')} {P_i(t')} \le 1 - \frac {\delta_b} {\Delta} \implies \max_{t' \in [t, T], i \in [n]} \beta_i(t') \le 1 - \frac {\delta_b} {\Delta},
\end{align*}

Otherwise, it always holds that

\begin{align*}
    \beta_i(t') = 1 - \frac {b_i(t')} {P_i(t')} \le 1 - \frac {\delta_b} {(t' - t + 1) \Delta} \implies \max_{t' \in [t, T], i \in [n]} \beta_i(t') \le 1 - \frac {\delta_b} {(T - t + 1) \Delta},
\end{align*}

since $P_i(t')$ (generally) is maximized when no liabilities are cleared for every $t'' \in [t, t']$. 

\subsubsection{Runtime Analysis for Equal interventions}

We briefly describe the process of the runtime analysis of our algorithm for the case when $L = \one \cdot \lambda$ for some $\lambda  \in \mathbb N^*$, i.e. with equal interventions. The analysis with unequal interventions is similar, but after a change in the rounding scheme from independent rounding to the dependent rounding scheme of \cite{srinivasan2001distributions} (see also \cite[pp. 38-40]{papachristou2021allocating}). Let $\varepsilon > 0$ be a quantity to be determined later. For brevity we let $$\gamma = \sup_{u(t:T)} \left \{ \max_{t' \in [t, T]} \max_{i \in [n]} \beta_{i}(t') \right \} \in (0, 1).$$

We fix a sample $u_i(t:T)$ of the financial environment and we let $V^{SOL}_{i, :t'}(t, s(t))$ be the value function corresponding to this realization truncated to step $t'$. Let $\mathcal F_i(t')$ be the event that the algorithm fails under the $i$-th realization up until step $t' \in [t,T]$. Clearly, $\mathcal F_i(t')$ is true when \emph{(i)} either the approximation guarantee fails, or \emph{(ii)} the constraints are not met. 

For \emph{(i)} we have that by Markov's Inequality

\begin{equation*}
\begin{split}
    \Pr {z_d(t:t')} {V^{SOL}_{i, :t'}(t, s(t)) \le (1 - \gamma - \varepsilon^2) V^{OPT}_{i, :t'}(t, s(t)) \bigg | \bigcap_{t''< t} (\mathcal F_i(t''))^c} \le \frac {1} {1 + \varepsilon^2 / \gamma} \le 1 - \frac {\varepsilon^2} {2 \gamma} \le 1 - \frac {\varepsilon^2} 2.
\end{split}
\end{equation*}

For \emph{(ii)} Let $\delta = \tfrac {\lambda} {B} \sqrt {n \log (2/\varepsilon^2) / 2}$. Since the rounded variables are binomial with a maximum number of trials $\lambda$ the constraint $\zero \le z_d(t') \le \one \lambda$ is satisfied with probability 1 for all $t' \in [t,T]$. For the budget constraint, by the Chernoff bound, we have that

\begin{equation*}
    \Pr {z_d(t:t')} {\one^T z_d(t') \ge (1 + \delta) B  \bigg | \bigcap_{t''< t} (\mathcal F_i(t'))^c} \le e^{-2 \delta^2 B^2 / (n \lambda^2)} \le \frac {\varepsilon^2} 2
\end{equation*}
    
Round $t'$ is repeated $\tau$ times and thus the probability that the algorithm fails at round $t'$ given that it has succeeded on rounds $t'' < t'$ is at most $\left ( 1 - \varepsilon^2 \right )^\tau \le e^{-\tau \varepsilon^2}$. Therefore, we have that for step $T$ 

\begin{equation*}
    \Pr {} {\text{algorithm succeeds given $i$-th shock}} = \Pr {} {(\mathcal F_i(T))^c} \ge \left ( 1 - e^{-\tau \varepsilon^2} \right )^{T - t} \ge 1 - (T - t + 1) e^{-\tau \varepsilon^2}
\end{equation*}

And thus the probability that the algorithm fails in any of the $N$ shocks is at most $(T - t + 1)N e^{-\tau \varepsilon^2}$. Let $\bar V^{OPT}_i(t, s(t))$ be the sample value function for the optimal policy $OPT$ bounded above by $\Delta T $, for $\Delta = \cdot \sup_{u \in \mathcal U} ( \| b \|_1 + \| \ell \|_1)$. The Chernoff bound implies that 

\begin{equation} \label{eq:shock_average_chernoff}
    \Pr {u(t:T)} {\bar V^{OPT}(t, s(t)) \ge (1 - \varepsilon) \ev {u(t:T)} {\bar V^{OPT}(t, s(t))}} \le e^{-2N\varepsilon^2 \ev {u(t:T)} {\bar V^{OPT}(t, s(t))} / (T\Delta)^2}
\end{equation}

Also note that for $B > 0$, $\ev {u(t:T)} {\bar V^{OPT}(t, s(t))} \ge (T - t + 1) \Theta > 0$, since for every round either all budget is used or all nodes are solvent, with $\Theta = \min \{ B, \min_{u} \| b \|_1 \}$, and therefore choosing 

\begin{equation*}
    N = \frac {\tau \Delta^2} {\Theta^2}
\end{equation*}

we get that \cref{eq:shock_average_chernoff} is at most $e^{-\tau \varepsilon^2}$. Therefore with probability at least $1 - (1 + (T - t + 1) N e^{-\tau \varepsilon^2} = 1 - O \left ( (T - t + 1) N e^{-\tau \varepsilon^2} \right )$ we have that 

\begin{equation*}
\begin{split}
    \bar V^{SOL}(t, s(t)) & \ge (1 - \gamma - \varepsilon^2 ) \bar V^{OPT} (t, s(t)) \\ &  \ge (1 - \gamma - \varepsilon^2 ) (1 - \varepsilon) \ev {u(t:T)} {V^{OPT}(t, s(t))}  \\ & \ge (1 - \gamma - 2 \varepsilon) \ev {u(t:T)} {V^{OPT}(t, s(t))}. 
\end{split}
\end{equation*}

\qed

\section{General Response Dynamics} \label{sec:general_dynamics}

In this Section, we study the dynamics when agents can respond with any feasible vector (not necessarily a fixed point as \cref{assumption:optimal_response} implies) satisfying \cref{eq:dynamics}. We restate the optimization problem again, for simplicity, for a given sample path $u(1:T)$.  

\begin{subequations} \label{eq:restated_optimization}
\begin{align} 
        \max_{\til P(1:T), Z(1:T)} \quad &  \sum_{t \in [T]} \one^T \til P(t) \\
        \text{s.t.} \quad & \zero \le \til P(t) \le P(t) \quad \forall t \in [T] \label{eq:solvency_constraint} \\
        & \til P(t) \le A^T(t) \til P(t) + c(t) + Z(t) \quad \forall t \in [T] \label{eq:default_constraints} \\
        & \zero \le Z(t) \le L  \quad \forall t \in [T] \\
        & \one^T Z(t) \le B \quad \forall t \in [T] \\
        & P(t) = b(t) + \ell(t) + P(t - 1) - \til P(t - 1) \quad \forall t \in [T]  \\
        & p_{ij}(t) = \ell_{ij}(t) + p_{ij}(t - 1) \left ( 1 - \frac {\til P_i(t - 1)} {P_i(t - 1)} \right ) \quad \forall i, j \in [n] \forall t \in [T] \\
        & a_{ij}(t) = p_{ij}(t) / P_i(t) \quad \forall i, j \in [n], \forall t \in [T]
\end{align}
\end{subequations}

\subsection{Convexity} \label{sec:convexity}

We investigate conditions under which \cref{eq:restated_optimization} corresponds to a convex program. Firstly, all the constraints and the objective are convex with respect to $Z(1:T)$. The objective is also convex in $\til P(1:T)$. The solvency constraints (\cref{eq:solvency_constraint}) are convex with respect to $\til P(1:T), Z(1:T)$. For the default constraints (\cref{eq:default_constraints}), we let $\varphi_{ji}(t) = a_{ji}(t) \til P_j(t)$ and $g_i(t) = \til P_i(t) - \sum_{j \in [n]} \varphi_{ji}(t) - c(t) - Z(t)$. Note that if $-\varphi_{ji}(t)$ are convex for all $i, j \in [n]$ and $t \in [T]$ then the solvency constraints are also convex, and vice versa since $\nabla^2 g_i(t) = \mathrm{diag} \left ( \left \{ - \nabla^2 \varphi_{ji}(t) \right \}_{j \in [n]} \right )$. All other second derivatives are zero since $\varphi_{j_1 i}(t)$ is independent from $\varphi_{j_2 i}(t)$ for all $j_1 \neq j_2$. The Hessian of $-\varphi_{ji}(t)$ equals: 

\begin{equation} \label{eq:vaphi_hessian}
    \nabla^2 \varphi_{ji}(t) = 
    \begin{pmatrix} 
        - \frac {\partial \varphi_{ji}^2(t)} {\partial \til P_j(t)^2} & -\frac {\partial \varphi_{ji}^2(t)} {\partial \til P_j(t) \partial \til P_j(t - 1)} \\ 
        - \frac {\partial \varphi_{ji}^2(t)} {\partial \til P_j(t - 1) \partial \til P_j(t)} &  - \frac {\partial \varphi_{ji}^2(t - 1)} {\partial \til P_j(t)^2} 
        \end{pmatrix} = 
        \begin{pmatrix} 
            0 & - \frac {a_{ji}(t) - a_{ji}(t - 1)} {P_j(t)} \\
           - \frac {a_{ji}(t) - a_{ji}(t - 1)} {P_j(t)} & - \frac {2 \til P_j(t) (a_{ji}(t) - a_{ji}(t))} {\til P_j(t)^2}
        \end{pmatrix} 
\end{equation}

The leading principal minors are $\Delta_1 = 0$, and $\Delta_2 = - \frac {(a_{ji}(t) - a_{ji}(t - 1))^2} {P_j^2(t)} \le 0$. In order to make $- \nabla^2 \varphi_{ji}(t) \succeq 0$  we should make the leading principal minors of $- \nabla^2 \varphi_{ji}(t)$ non-negative. For this to happen, we should set $\Delta_2 = 0$ (since we require it to be $\ge 0$ for positive semi-definiteness and $\Delta_2 \le 0$ by \cref{eq:vaphi_hessian}). This requires setting $a_{ji}(t)$ to be some constant $\zeta_{ji} \in [0, 1)$ for all $t \in [T]$. Therefore the necessary and sufficient condition for convexity of the dynamics is that for every $i \in [n], t \in [T]$

\begin{equation*}
    a_{ji}(t) = \zeta_{ji} \; \forall j \in [n] \iff -\nabla^2 \varphi_{ji} \succeq 0 \iff \nabla^2 g_{i}(t) \succeq 0 \iff g_i(t) \text{ is convex}.
\end{equation*}

\subsection{A Necessary and Sufficient Condition for Convexity}

We restrict $A(t)$ to be some constant row sub-stochastic matrix $\mathfrak Z$ with entries $\zeta_{ij} \in [0, 1)$ for all $t \in [T]$. 

This yields the following condition and the following Condition and Theorem: 

\begin{condition} \label{condition:weaker_financial_environment}
    The financial environment $U(t) = (b(t), c(t),\{ \ell_ij(t) \}_{i, j \in [n]})$ is a MC subject to the constraint that $\frac {\ell_{ij}(t)} {b_i(t) + \sum_{k \in [n]} \ell_{ik}(t)}$ is constant for all $i, j \in [n]$ and $t \in [T]$. Equivalently the internal and external instantaneous liabilities are samples from the polytope $\mathcal K = \{ (b, \{ \ell_{ij} \}_{i, j \in [n]}) \in \mathbb R^{3n} : \zeta_{ij} + \sum_{k \in [n]} (\zeta_{ij} - \one \{ k = j \}) \ell_{ij} = 0, \forall i, j \in [n] \}$
\end{condition}

Which yields \cref{theorem:sufficient_condition_for_linearity},

\begin{theorem}[Necessary and Sufficient Condition for Convexity] \label{theorem:sufficient_condition_for_linearity}
    If for every $i, j \in [n]$ the quantity  $\frac {\ell_{ij}(t)} {b_i(t) + \sum_{k \in [n]} \ell_{ik}(t)}$ is constant and equal to $0 \le \zeta_{ij} < 1$ for a (row)-substochastic matrix $\mathfrak Z = \{ \zeta_{ij} \}_{i, j \in [n]}$, if and only if \cref{eq:restated_optimization} corresponds to a convex program. Moreover, under \cref{theorem:sufficient_condition_for_linearity}, the optimization program of \cref{eq:restated_optimization} is a linear program described in \cref{sec:linear_program}. 
\end{theorem}

\begin{proof}
    
For all $i, j \in [n]$ and $t \in [T]$ we have that  

\begin{align*}
    p_{ij}(t) & = \zeta_{ij} P_i(t) & \iff \\ 
    \ell_{ij}(t) + p_{ij} (t - 1) \left ( 1 - \frac {\til P_i(t - 1)} {P_i(t - 1)} \right ) & = \zeta_{ij} \left ( b_i(t) + \ell_i(t) + P_i(t - 1) - \til P_i (t - 1) \right ) & \iff  \\
    \ell_{ij}(t) + \zeta_{ij} P_i(t - 1) \left ( 1 - \frac {\til P_i(t - 1)} {P_i(t - 1)} \right ) & = \zeta_{ij} \left ( b_i(t) + \ell_i(t) + P_i(t - 1) - \til P_i (t - 1) \right ) & \iff \\ 
    \ell_{ij}(t) + \zeta_{ij} P_i(t - 1) - \zeta_{ij} \til P_i(t - 1) & = \zeta_{ij} \left ( b_i(t) + \ell_i(t) + P_i(t - 1) - \til P_i(t - 1) \right ) & \iff \\
    \ell_{ij}(t) & = \zeta_{ij} \left ( b_i(t) + \ell_i(t) \right ) & \iff \\ 
    \frac {\ell_{ij}(t)} {b_i(t) + \sum_{k \in [n]} \ell_{ik}(t)} & = \zeta_{ij}. 
\end{align*}

The primal is easily obtained from \cref{eq:restated_optimization} by replacing $P(t)$ with its definition and recursing. The dual is obtained by standard LP duality arguments (see e.g. \cite{boyd2004convex}). 

\end{proof}

\subsection{Linear Programming Formulation under \cref{condition:weaker_financial_environment} and \cref{theorem:sufficient_condition_for_linearity}} \label{sec:linear_program}

In the following, the $\mathsf {vec ( \cdot )}$ operator flattens a matrix to a vector (by row and then by column).

The primal and dual LPs (for a fixed sample path) that correspond to \cref{eq:restated_optimization} are (we define $h(t) = \sum_{t' \le t} (b(t) + \ell(t))$ 

\medskip

\begin{minipage}[h]{.5\textwidth}
\small
\textbf{Primal LP}
\begin{subequations} \label{eq:linear_program_primal}
\begin{align}
        \max_{\til P, Z} \quad &  \sum_{t \in [T]} \one^T \til P(t) \\
        \text{s.t.} \quad &  \sum_{t' \le t} \til P(t') \le h(t) & \quad \forall t \in [T]  \label{eq:solvency_constraint} \\
        & \zero \le \til P(t) \le \mathfrak Z^T \til P(t) + c(t) + Z(t) & \quad \forall t \in [T] \label{eq:default_constraint} \\
        & \zero \le Z(t) \le L & \quad \forall t \in [T] \label{eq:interventions_nonnegative_constraint} \\
        & \one^T Z(t) \le B & \quad \forall t \in [T] \label{eq:budget_constraint} 
\end{align}
\end{subequations}
\end{minipage}
\begin{minipage}[h]{.5\textwidth}
\small
\textbf{Dual LP}
\begin{subequations} \label{eq:linear_program_dual}
\begin{align}
    \min_{\til \lambda, \til \mu, \til \nu, \til \xi} \quad & \sum_{t \in [T]} \left [ \til \lambda^T(t) h(t) + \til \mu^T(t) c(t) + B \til \nu(t) + L^T \xi(t) \right ] \\
    \text{s.t.} \quad & \til \lambda(1:T) \ge \zero \\ & \til \mu(1:T) \ge \zero \\ &  \til \nu(1:T) \ge \zero  \\ & \til \xi(1:T) \ge \zero \\
    & \til \mu(t) \ge \mathfrak Z \til \mu(t) + \one - \sum_{t' \ge t} \til \lambda(t') \quad \forall t \in [T] \\
    & \til \xi(t) + \one \til \nu(t) \ge \til \mu(t) \quad \forall t \in [T]
\end{align}
\end{subequations}
\end{minipage}

We also consider the prefix payments $\til Q(t) = \sum_{t' \le t} \til P(t')$, prefix interventions $W(t) = \sum_{t' \le t} Z(t')$ and prefix assets $f(t) = \sum_{t' \le t} c(t')$. From the properties of the EN model we know that the optimal payments and optimal interventions can be found by optimizing \emph{any} strictly increasing function of the clearing payment vectors. For this reason we consider the objective

\begin{equation} \label{eq:q_objective}
    \sum_{t \in [T]} t \cdot ( \one^T \til P(t) ) = \sum_{t \in [T]} \one^T \til Q(t) = \one^T \mathsf{vec}(\til Q(1:T))
\end{equation}

The constraints for $\til Q(t)$ become 

\begin{subequations} \label{eq:q_constraints}
\begin{align}
    \til Q(t) & \le h(t) & \quad \forall t \in [T] \\
    \til Q(t) & \ge \zero & \quad \forall t \in [T] \\
    \til Q(t) & \le \mathfrak Z^T \til Q(t) + f(t) + W(t) & \quad \forall t \in [T] \\
    W(t) & \ge \zero & \quad \forall t \in [T] \\
    W(t) & \le t L & \quad \forall t \in [T] \\
    \one^T W(t) & \le tB & \quad \forall t \in [T]
\end{align}
\end{subequations}

The optimal values of $\til Q^*(1:T), \til W^*(1:T)$ found by solving \cref{eq:q_objective} subject to \cref{eq:q_constraints} can be used to recover the optimal $\til P^*(1:T), Z^*(1:T)$. Moreover, the optimization problem of \cref{eq:q_objective,eq:q_constraints} corresponds to solving an one-period clearing problem on the vectorized payments $\mathsf{vec}(\til Q(1:T))$ and the vectorized interventions $\mathsf{vec}(W(1:T))$ with

\begin{compactitem}
    \item Relative liability matrix $\mathrm{diag}(\mathfrak Z, \dots, \mathfrak Z)$.
    \item Asset vector $\mathsf{vec}(f(1:T))$.
    \item Liability vector $\mathsf{vec}(h(1:T))$.
    \item intervention constraint $\mathsf{vec}(W(1:T)) \le \begin{pmatrix} L \\ 2L \\ \vdots \\ TL \end{pmatrix}$.
    \item Budget constraint $\begin{pmatrix} \one^T & \zero^T & \dots & \zero^T \\ \zero^T & \one^T & \dots & \zero^T \\ \vdots & \vdots & \dots & \vdots \\ \zero^T & \zero^T & \dots & \one^T  \end{pmatrix} \mathsf{vec}(W(1:T)) \le \begin{pmatrix} B \\ 2B \\ \vdots \\ TB  \end{pmatrix}$.
\end{compactitem}

\subsection{Optimality of the Myopic (Sequential) Policy under \cref{condition:weaker_financial_environment}}

The sequential/myopic policy solves the contagion problem at each round $t$ and then supplies the solution to round $t + 1$, and so on. Below we prove that under \cref{condition:weaker_financial_environment} and \cref{theorem:sufficient_condition_for_linearity} the myopic policy is in fact the optimal one 

\begin{theorem}
    Under \cref{condition:weaker_financial_environment} the myopic policy is optimal for \cref{eq:linear_program_primal}, i.e. for all $t \in [T]$ the reward accumulated from time $t$ onwards (for the myopic policy) for every $\til P(1:t-1), Z(1:T-1)$ equals 
    
    \begin{align*}
        \max_{\til P(t:T), Z(t:T)} \quad & \sum_{t' \ge t} \one^T \til P(t') \\
        \text{s.t.} \quad & \text{\cref{eq:solvency_constraint,eq:default_constraint,eq:interventions_nonnegative_constraint,eq:budget_constraint}}
    \end{align*}
    
    Subsequently, for $t = 1$ we get that the optimal policy can be found by solving \cref{eq:linear_program_primal}. 
    
\end{theorem}

\begin{proof}

    Fix some financial environment obeying \cref{condition:weaker_financial_environment}. Then define the value function at which the environment responds myopically, i.e. in a way that maximizes its reward from round $t$ onwards, that is $V^\flat(t) =  \max_{\til p \text{ feasible clearing vector}} \max_{z \in \mathcal Z} \left \{\one^T \til p + V^\flat(t + 1) \right \}$. In our case, for $t' \ge T$ we have that $V^\flat(t') = 0$. For $t' = T$ we have that for all $\til P(1:T-1), Z(1:T-1)$ 
    
    \begin{align*}
        V^\flat(T) & = \max_{\til P(T)} \max_{Z(T)} \one^T \til P(T) \overset {\text{LP}} {=}  \max_{\til P(T), Z(T)} \one^T \til P(T) \\
        \text{s.t.} \quad & \text{\cref{eq:solvency_constraint,eq:default_constraint,eq:interventions_nonnegative_constraint,eq:budget_constraint}}
    \end{align*}
    
    Now assume that for $t' = t + 1$ and for all $\til P(1:t), Z(1:T)$ that 
    
    \begin{align*}
        V^\flat(t + 1) & = \max_{\til P(t+1:T), Z(t+1:T)} \sum_{t' \ge t + 1} \one^T \til P(t') \\
        \text{s.t.} \quad & \text{\cref{eq:solvency_constraint,eq:default_constraint,eq:interventions_nonnegative_constraint,eq:budget_constraint}}
    \end{align*}
    
    For $t' = t$ and for all $\til P(1:t-1), Z(1:t-1)$ the value function of the myopic equals
    
    \begin{align*}
        V^\flat(t) & = \max_{\til P(t)} \left \{ \max_{Z(t)} \one^T \til P(t) + V^\flat(t + 1) \right \} \\
        & \overset {\text{Inductive Hypothesis}} {=} \max_{\til P(t)} \left \{ \max_{Z(t)} \one^T \til P(t) + \max_{\til P(t+1:T), Z(t+1:T)} \sum_{t' \ge t + 1} \one^T \til P(t')  \right \} \\
        & \overset {\text{LP}} {=} \max_{\til P(t)} \max_{Z(t:T), \til P(t+1:T)} \sum_{t' \ge t} \one^T \til P(t) \\
        & \overset {\text{LP}} {=} \max_{\til P(t:T), Z(t:T)} \sum_{t' \ge t} \one^T \til P(t) \\
        \text{s.t.} \quad & \text{\cref{eq:solvency_constraint,eq:default_constraint,eq:interventions_nonnegative_constraint,eq:budget_constraint}}
    \end{align*}
    
    We have merged the maximizations because the objective functions are linear and the constraints are also linear (and separable). 
    Recursing to $t = 1$ we get that $V^\flat(1) \equiv \text{\cref{eq:linear_program_primal}}$. 

\end{proof}

\section{Data Addendum for SafeGraph} \label{sec:data_addendum}

\subsection{Network Topology}

The network is bipartite and consists of two types of nodes: POI nodes that correspond to businesses (such as restaurants, gyms, grocery stores, etc.) and CBG nodes (which represent unidentifiable groups households of people). The data period spans the period of December 2020 to April 2021 with monthly granularity. The POI nodes are constructed as follows: we start from some geographical coordinates (in our case a small rural US city center) and then we find the POIs that belong to the $k$-closest CBGs (in terms of their Haversine distance). In our experiment we have chosen $k$ to be 3. Then for each POI and each period of interaction we find out devices from which CBGs visited the specific POI by using  the \texttt{visitor\_home\_cbg} column to find the number of visitors for each POI from each CBG assuming that each device represents a different person. For each CBG $j$ we find the average household size $\bar h_j$, the average number of dependents $\bar d_j$, the average income $\bar \iota_j$ and the probability of someone being unemployed $p_j^{\mathrm{employment}}$ using data from the 2020 US Census. 

\subsection{Internal Liabilities}

For each POI-CBG interaction there are two types of edges: consumption edges and salary edges. The former type of edge is from a CBG to a POI whereas the latter type of edge is from a POI to a CBG. The weights on the edges are calculated following a two-step approach: Firstly, we use the \emph{bucketed dwelling times} to find the percentage of people that are \emph{employed} on these businesses (which we denote by $p_j^{\mathrm{worker}}(t)$), where we consider someone to be employed if they spend more than 4 hours in the POI, and we consider the rest to be consumers (i.e. the ones that spend less than 4 hours in the POI). Between a POI $i$ and a CBG $j$ whereas $n_{ij}(t)$ people have commuted from the CBG to the POI we add the two edges: An edge $(j, i)$ for the worker nodes that we estimate to be $n_{ij}^{\mathrm{workers}}(t) = \lfloor n_{ij}(t) \times p^{\mathrm{worker}}_j(t) \rfloor$ workers, and an edge $(i, j)$ for the non-workers with $n_{ij}^{\mathrm{non-workers}}(t) = n_{ij}(t) - \lfloor n_{ij}(t) \times p^{\mathrm{workers}}_j(t) \rfloor $. For each type of POI $i$ we find the average salary $\bar s_i$, and the average consumption $\bar o_i$ using data from the NAICS and the US Economic Census (Consumer Expenditure Survey), and calculate the final internal liabilities between POI $i$ and CBG $j$ to be

\begin{equation*}
    \ell_{ij}(t) = n_{ij}^{\mathrm{workers}}(t) \times \bar s_i, \; \ell_{ji}(t) = n_{ij}^{\mathrm{non-workers}}(t) \times \bar o_i
\end{equation*}

The total number of workers number of workers (resp. non-workers) for each POI $j$ and each CBG $i$ is estimated to be 

\begin{equation*}
    n_j^{\mathrm{workers}}(t) = \sum_{i \text{ is CBG}} n_{ij}^{\mathrm{workers}} (t), \; n_j^{\mathrm{non-workers}}(t) = \sum_{i \text{ is CBG}} n_{ij}^{\mathrm{non-workers}} (t)
\end{equation*}

\begin{equation*}
        n_i^{\mathrm{workers}}(t) = \sum_{j \text{ is POI}} n_{ij}^{\mathrm{workers}} (t), \; n_i^{\mathrm{non-workers}}(t) = \sum_{j \text{ is POI}} n_{ij}^{\mathrm{non-workers}} (t)
\end{equation*}

The estimated total number of households in the CBG $j$ that are related to interaction with the corresponding POIs is

\begin{equation*}
    n_j^{\mathrm{households}}(t) = \left \lceil \frac {\frac {n_j^{\mathrm{workers}}(t)}  {p_j^{\mathrm{employment}}} + n_j^{\mathrm{non-workers}}(t)} {\bar h_j} \right \rceil .
\end{equation*}

\subsection{Interventions} 

For each CBG $j$ we use the US CARES Act rule to calculate the interventions as follows 

\begin{equation*}
    L_j(t) = n_j^{\mathrm{households}(t)} \times \begin{cases}
        1200 & \bar h_j < 2, \bar \iota_j \le 75000 \\
        \left ( 1200 \times \frac {150000 - \bar \iota_j} {75000} \right ) \vee 0 & \bar h_j < 2, \bar \iota_j > 75000 \\
        2400 + 500 (\bar h_j - 2) &  \bar h_j \ge 2, \bar \iota_j \le 150000 \\
        \left ( (2400 + 500 (\bar h_j - 2) \times \frac {300000 - \bar \iota_j} {150000} \right ) \vee 0 & \bar h_j \ge 2, \iota_j > 150000  
    \end{cases}. 
\end{equation*}

For each POI $i$ the Safegraph data provides annual business loans from the Paycheck Protection Program (PPP), whereas each loan has value $\psi_i$. Also the businesses report the number of employees $n_{i}^{\mathrm{PPP-employees}}$, which we use to calculate the intervention as 

\begin{equation*}
    L_i(t) = \frac {1} {12} \times \frac {\psi_i} {n_i^{\mathrm{PPP-employees}} } \times n_{i}^{\mathrm{workers}}(t).
\end{equation*}

In our data, there are nodes that after processing have missing interventions, or interventions that equal 0. For these nodes, we proceed with the following \emph{imputation scheme:} For each POI node (resp. CBG node), we replace the missing interventions with the average intervention of the POI node (resp. CBG node) across rounds where the interventions are well defined. If a POI node (resp. CBG node) has no interventions across all rounds, we set the intervention of the node (for all rounds) to equal the average POI intervention (resp. average CBG intervention). 
    
\subsection{External Assets and Liabilities}

For each POI $i$ we use the total assets earned annually from all establishments, normalized them by month, divide by the total number of workers for the specific NAICS code that the POI belongs, and multiply $n_i^{\mathrm{workers}}(t)$. Finally, from this amount we subtract the revenue due to nodes within the network $\sum_{j \text{ is CBG}} \ell_{ji}(t)$ (i.e. inbound edges) and take the positive part in case the result is negative to deduce $c_i(t)$. $b_i(t)$ is determined in a similar way with the only change that the total weight of the outbound edges is subtracted in the end. To ensure that \cref{assumption:uniqueness} holds throughout the experiments we assert a minimum of \$100 for the external liabilities. Such a number is smaller than the order of most quantities and thus does not significantly affect the results.
    
For each CBG $j$ the process is similar where for determining $c_j(t)$ we use $\bar \iota_j$, following a similar procedure as in the POI case, but now normalized with respect to $n_j^{\mathrm{households}}(t)$. For $b_j(t)$, we assume that each household spends the national US average of $\sim \! \$ 63,000$.

\end{document}